\documentclass [oneside,reqno,12pt] {article}
\usepackage{amsfonts,latexsym,amsmath,amssymb,amsthm,lineno,graphicx,bm,changepage,enumitem,cases}
\input epsf
\def\a{\alpha}
\newcommand\be{\beta}
\newcommand\de{\delta}
\newcommand\De{\Delta}

\newcommand\vth{\vartheta}
\newcommand\et{\eta}
\newcommand\ep{\epsilon}
\newcommand\ga{\gamma}
\newcommand\Ga{\Gamma}
\newcommand\ka{\kappa}
\newcommand\la{\lambda}
\newcommand\La{\Lambda}
\newcommand\ph{\varphi}

\newcommand\ps{\psi}
\newcommand\Ps{\Psi}
\newcommand\rh{\rho}
\newcommand\si{\sigma}
\newcommand\Si{\Sigma}
\newcommand\ta{\tau}
\newcommand\om{\omega}
\newcommand\Om{\Omega}

\newcommand\A{\mathcal A}
\newcommand\B{\mathcal B}
\newcommand\I{\mathsf I}
\newcommand\E{\mathsf E}
\newcommand\F{\mathsf F}
\newcommand\G{\mathsf G}
\newcommand\K{\mathsf K}
\newcommand\Loc{\mathsf L}
\newcommand\N{\mathsf N}
\def\S{\mathsf S}
\newcommand\U{\mathsf U}
\newcommand\X{\mathsf X}
\newcommand\Real{\mathbb R}
\newcommand\Nat{\mathbb N}
\newcommand\Rat{\mathbb Q}
\newcommand\Mtwo{{\mathbf M}_2}
\newcommand\Mthree{{\mathbf M}_3}
\newcommand\SA{{\mathbf S}_{\mathcal A}}
\newcommand\SB{{\mathbf S}_{\mathcal B}}

\newcommand\IDMI{\mathbf{I}_{\rm DMI}}
\newcommand\rA{r_{\mathcal A}}
\newcommand\rB{r_{\mathcal B}}
\newcommand\mA{m_{\mathcal A}}
\newcommand\mB{m_{\mathcal B}}
\newcommand\mmax{m_{\rm max}}

\newcommand\ssQ{{\mathsf Q}}

\newcommand\ssW{{\mathsf W}}

\newcommand\Cov{\operatorname{Cov}}
\newcommand\grad{\operatorname{grad}}

\newcommand\cmin{c_{\text{min}}}
\newcommand\ctot{c_{\text{tot}}}

\newcommand\abs[1]{\left|#1\right|}
\newcommand{\der}[2]{\frac{\text{d} #1}{\text{d} #2}}
\newcommand{\pder}[2]{\frac{\partial #1}{\partial #2}}

\newcommand\Tr{{\top}}  
\newcommand\pikak{p_{i_k,\a}^{(k)}}
\newcommand\pinan{p_{i_n,\a}^{(n)}}
\newcommand\pjkak{p_{j_k,\a}^{(k)}}

\newcommand\pikk{p_{i_k}^{(k)}}
\newcommand\pinn{p_{i_n}^{(n)}}
\newcommand\piKaK{p_{i_\K,\a}^{(\K)}}
\newcommand\piNaN{p_{i_\N,\a}^{(\N)}}
\newcommand\Pinn{P_{i_n}^{(n)}}
\newcommand\PiNN{P_{i_\N}^{(\N)}}
\newcommand\Pikk{P_{i_k}^{(k)}}
\newcommand\PiKK{P_{i_\K}^{(\K)}}
\newcommand\Pjkk{P_{j_k}^{(k)}}
\newcommand\qikk{q_{i_k}^{(k)}}
\newcommand\rinan{r_{i_n,\a}^{(n)}}
\newcommand\ominn{\om_{i_n}^{(n)}}
\newcommand\wbara{{\bar w}_\a}
\newcommand\rbara{{\bar r}_\a}

\newcommand\pia{p_{i,\a}}
\newcommand\pib{p_{i,\be}}
\newcommand\pja{p_{j,\a}}

\newcommand\vba{\bar v_\a}
\newcommand\wija{w_{ij,\a}}

\newcommand\wia{w_{i,\a}}

\newcommand\wb{{\bar w}}
\newcommand\wba{{\bar w_\a}}
\newcommand\mab{m_{\a\be}}
\newcommand\tmab{\tilde m_{\a\be}}
\newcommand\mba{m_{\be\a}}
\newcommand\tmba{\tilde m_{\be\a}}
\newcommand\mb{{\bar m}}
\newcommand\tw{\tilde w}
\newcommand\pa{p_\a}
\newcommand\pb{p_\be}
\newcommand\xa{x_\a}
\newcommand\xb{x_\be}
\newcommand\ya{y_\a}
\newcommand\yb{y_\be}
\newcommand\dota{{\cdot,\a}}

\newtheorem{theorem}{Theorem}
\newtheorem{lemma}[theorem]{Lemma}
\newtheorem{remark}[theorem]{Remark}
\newtheorem*{remark*}{Remark}
\newtheorem{example}[theorem]{Example}
\newtheorem{proposition}[theorem]{Proposition}
\newtheorem*{proposition*}{Proposition}
\newtheorem{corollary}[theorem]{Corollary}

\global\voffset-0.3cm

\begin{document}
\normalbaselineskip18pt
\baselineskip18pt
\global\hoffset=-15truemm
\global\voffset=-10truemm
\numberwithin{equation}{section}
\numberwithin{theorem}{section}
\numberwithin{figure}{section}
\numberwithin{table}{section}
\setcounter{section}{0}

\parindent=0pt
\vspace{1.5cm}
{\bf\LARGE A Survey on Migration-Selection Models \\in Population Genetics\footnote{{\it Acknowledgments.}
This survey is based on lecture notes I compiled for a lectures series presented during the ``Program on Nonlinear Equations in Population Biology'' in
June 2013 at the Center for PDE, East China Normal University, Shanghai. In turn, these lecture notes were a revised version of notes written for a lecture course at the University of Vienna in 2010. I am very grateful to the participants of these courses for many helpful comments, in particular to Dr.\ Simon Aeschbacher, Benja Fallenstein, and Dr.\ Linlin Su. Comments by and useful discussions with Professor Yuan Lou helped to finalize this work. I am particularly grateful to Professor Wei-Ming Ni for inviting me to the Center for PDE, and to him and its members for their great hospitality. Financial support by grant P25188 of the Austrian Science Fund FWF is gratefully acknowledged.}}

\vspace{1.5cm}
{\bf\large Reinhard B\"urger}
 
\vspace{.5cm}
Department of Mathematics, University of Vienna

\vspace{3cm}

{
\normalbaselineskip14pt
\baselineskip14pt{
{\obeylines 
{\it Contact:}
\bigskip 
Reinhard B\"urger
Institut f\"ur Mathematik 
Universit\"at Wien 
Oskar-Morgenstern-Platz 1
A-1090 Wien 
Austria
E-mail: reinhard.buerger@univie.ac.at
Phone: +43 1 4277 50631
Fax: +43 1 4277 9506}


\vspace{.3cm}

\vspace{.3cm}
{\it Running head:} Migration-selection models

}
}

\newpage
\section*{Abstract}
This survey focuses on the most important aspects of the mathematical theory of population genetic models of selection and migration between discrete niches. Such models are most appropriate if the dispersal distance is short compared to the scale at which the environment changes, or if the habitat is fragmented.
The general goal of such models is to study the influence of population subdivision and gene flow among subpopulations on
the amount and pattern of genetic variation maintained. Only deterministic models are treated. Because space is discrete, they are formulated in terms of systems of nonlinear difference or differential equations. A central topic is the exploration of the equilibrium and stability structure under various assumptions on the patterns of selection and migration. Another important, closely related topic concerns conditions (necessary or sufficient) for fully polymorphic (internal) equilibria. First, the theory of one-locus models with two or multiple alleles is laid out. Then, mostly very recent, developments about multilocus models are presented. Finally, as an application, analysis and results of an explicit two-locus model emerging from speciation theory are highlighted.

\vspace{.5cm}

\tableofcontents

\vspace{.5cm}
\parindent=15pt

\section{Introduction}
Population genetics is concerned with the study of the genetic composition 
of populations. This composition is shaped by selection, mutation, recombination, mating behavior and reproduction,
migration, and other genetic, ecological, and evolutionary factors.
Therefore, these mechanisms and their interactions and evolutionary consequences are investigated.
Traditionally, population genetics has been applied to animal and plant breeding, to human genetics,
and more recently to ecology and conservation biology.  
One of the main subjects is the investigation of the mechanisms that generate and maintain genetic variability in 
populations, and the study of how this genetic variation, shaped by environmental
influences, leads to evolutionary change, adaptation, and speciation.
Therefore, population genetics provides the basis for understanding the evolutionary
processes that have led to the diversity of life we encounter and admire. 

Mathematical models and methods have a long history in population genetics,
tracing back to Gregor Mendel, who used his education in mathematics and physics to draw his conclusions. Francis Galton and the biometricians, notably Karl Pearson,
developed new statistical methods to describe the distribution of trait 
values in populations and to predict their change between generations. Yule (1902), Hardy (1908), and Weinberg (1908) worked out simple, but important, consequences of the particulate mode of inheritance proposed by Mendel in 1866 that contrasted and challenged the then prevailing blending theory of inheritance.
However, it was not before 1918 that the synthesis between genetics and the theory of evolution through natural selection began to take shape through
Fisher's (1918) work. By the early 1930s, the foundations of modern population genetics had been laid by the work of Ronald A.\ Fisher, J.B.S.\ Haldane, and
Sewall Wright. They had demonstrated that the theory of evolution by natural selection, proposed by Charles Darwin in 1859, can be justified on the basis of genetics as governed by Mendel's laws. A detailed account of the history of population genetics is given in Provine (1971).

In the following, we explain some basic facts and mechanisms that are needed throughout this survey. 
Mendel's prime achievement was the recognition of the particulate nature of the hereditary determinants, now called genes.
Its position along the DNA is called the \emph{locus}, and a particular sequence there is called an \emph{allele}.
In most higher organisms, genes are present in pairs, one being inherited from the mother, the other from the father. 
Such organisms are called diploid. The allelic composition is called the \emph{genotype}, and the set of observable properties derived from
the genotype is the \emph{phenotype}.

\emph{Meiosis} is the process of formation of reproductive cells, or \emph{gametes} (in animals, sperm and eggs)
from somatic cells. Under \emph{Mendelian segregation}, each gamete contains
precisely one of the two alleles of the diploid somatic cell and each gamete is equally likely to
contain either one. The separation of the paired alleles from one another and their distribution to the gametes is
called \emph{segregation} and occurs during meiosis. At mating, two reproductive cells fuse and 
form a \emph{zygote} (fertilized egg), which contains the full (diploid) genetic information.

Any heritable change in the genetic material is called a \emph{mutation}. Mutations are the ultimate source of genetic 
variability and form the raw material upon which selection acts. 
Although the term mutation includes changes in chromosome structure and number, the vast majority of
genetic variation is caused by changes in the DNA sequence. Such mutations occur in many different ways, 
for instance as base substitutions, in which one nucleotide is replaced by another, as insertions or deletions of DNA, as inversions of 
sequences of nucleotides, or as transpositions. For the population-genetic models treated in this text the molecular origin of a mutant is 
of no relevance because they assume that the relevant alleles are initially present.

During meiosis, different chromosomes assort independently and \emph{crossing over} between two homologous chromosomes
may occur. Consequently, the newly formed gamete contains maternal alleles at one set of loci and
paternal alleles at the complementary set. This process is called \emph{recombination}. Since it leads
to random association between alleles at different loci, recombination has
the potential to combine favorable alleles of different ancestry in one gamete and to break up combinations of deleterious alleles.
These properties are generally considered to confer a substantial evolutionary advantage to sexual species relative to asexuals.
  
The mating pattern may have a substantial influence on the evolution of
gene frequencies. The simplest and most important mode is \emph{random
mating}. This means that matings take place without regard to ancestry or
the genotype under consideration. It seems to occur frequently in nature.
For example, among humans, matings within a population appear to be random with respect to
blood groups and allozyme phenotypes, but are nonrandom with respect to other traits, for example, height.

\emph{Selection} occurs when individuals of different
genotype leave different numbers of progeny because they differ in their
probability to survive to reproductive age (\emph{viability}), 
in their mating success, or in their average number of produced offspring (\emph{fertility}).
Darwin recognized and documented the central importance of
selection as the driving force for adaptation and evolution. Since
selection affects the entire genome, its consequences for the genetic
composition of a population may be complex. Selection is measured in terms
of \emph{fitness} of individuals, i.e., by the number of progeny contributed to the next generation. There are different measures
of fitness, and it consists of several components because selection may act on each stage of the life cycle.

Because many natural populations are geographically structured and selection
varies spatially due to heterogeneity in the environment, it is important to study the consequences of spatial structure for the evolution of populations. 
Dispersal of individuals is usually modeled in one of two alternative ways, either
by diffusion in space or by migration between discrete niches, or demes.
If the population size is sufficiently large, so that random genetic drift can be ignored, then the first kind of model leads to 
partial differential equations (Fisher 1937, Kolmogorov et al.\ 1937). This is
a natural choice if genotype frequencies change continuously along an environmental gradient,
as it occurs in a cline (Haldane 1948). Here we will not be concerned with this wide and fruitful area
and refer instead to Barton (1999), Nagylaki and Lou (2008), Barton and Turelli (2011), Lou et al.\ (2013) for recent developments and references.

Instead, this survey focuses on models of selection and migration between discrete demes. Such models are most appropriate if the dispersal distance is short compared to the scale at which the environment changes, or if the habitat is fragmented. They originated from the work of Haldane (1930) and
Wright (1931). Most of the existing theory has been devoted to study selection on a single locus in
populations with discrete, nonoverlapping generations that mate randomly within demes. However, advances in the theory of multilocus models have been made recently. 
The general goal is to study the influence of population subdivision and of gene flow among subpopulations on
the amount and pattern of genetic variation maintained. The models are formulated in terms of systems of
nonlinear difference or differential equations. The core purpose of this survey is the presentation of the methods of their analysis, of the main results, and of their biological implications. 

For mathematically oriented introductions to the much broader field of population genetics, we refer to the
books of Nagylaki (1992), B\"urger (2000), Ewens (2004), and Wakeley (2008). 
The two latter texts treat stochastic models in detail, an important and topical area ignored in this survey. 

\vspace{.5cm}
\section{Selection on a multiallelic locus}
Darwinian evolution is based on selection and inheritance. In this section, we summarize the essential properties of simple selection models.
It will prepare the ground for the subsequent study of the joint action of spatially varying selection and migration.
Proofs and a detailed treatment may be found in Chapter I of B\"urger (2000). Our focus is on the evolution of the genetic composition 
of the population, but not on its size. Therefore, we always deal with relative frequencies of genes or genotypes within a given population. 

Unless stated otherwise, we consider a population with discrete, nonoverlapping generations, such as annual plants or insects.
We assume two sexes that need not be distinguished because gene or genotype frequencies are the same in both sexes (as is always the case
in monoecious species). Individuals mate at random with respect to the locus under consideration, i.e., in proportion to their frequency.
We also suppose that the population is large enough that gene and genotype frequencies can be 
treated as deterministic, and relative frequency can be identified with probability. Then the evolution of gene or 
genotype frequencies can be described by difference or recurrence equations. These assumptions reflect an idealized situation
which will model evolution at many loci in many populations or species, but which is by no means universal.

\subsection{The Hardy--Weinberg Law}\label{sec:HW}
With the blending theory of inheritance variation in a population declines rapidly, and this was one of the
arguments against Darwin's theory of evolution. With Mendelian 
inheritance there is no such dilution of variation, as was shown independently by the famous British
mathematician Hardy (1908) and, in much greater generality, by the German physician Weinberg (1908, 1909).

We consider a single locus with $I$ possible alleles $\A_i$ and write $\I=\{1,\ldots,I\}$ for the set
of all alleles. We denote the frequency of the ordered genotype $\A_i\A_j$ by $P_{ij}$, so that the frequency of
the unordered genotype $\A_i\A_j$ is $P_{ij}+P_{ji}=2P_{ij}$. Subscripts $i$ and $j$ always refer to alleles. Then 
the frequency of allele $\A_i$ in the population is
$$
	p_i = \sum_{j=1}^I P_{ij}  \;. \footnote{If no summation range is indicated, it is assumed to be over all
			 		admissible values; e.g., $\sum_i=\sum_{i\in\I}$}                                    
$$
After one generation of random mating the zygotic proportions satisfy\footnote{Unless stated otherwise, a prime, ${}'$, always signifies the next generation. Thus,
instead of $P_{ij}(t)$ and $P_{ij}(t+1)$, we write $P_{ij}$ and $P_{ij}'$ (and analogously for other quantities).}
$$
	P_{ij}' = p_ip_j \qquad\text{for every $i$ and $j$}\;.  
$$

A mathematically trivial, but biologically very important, consequence is that (in the absence of other forces)
gene frequencies remain constant across generations, i.e.,
\begin{equation}\label{p_i'=p_i}
	p_i' = p_i \quad\text{for every }i\,. 
\end{equation}
In other words, in a (sufficiently large) randomly mating population reproduction does not change allele 
frequencies. A population is said to be in \emph{Hardy--Weinberg equilibrium} if
\begin{equation}\label{HW}
	P_{ij} = p_ip_j\,.
\end{equation}
In a (sufficiently large) randomly mating population, this relation is always satisfied among zygotes.

Evolutionary mechanisms such as selection, migration, mutation, or random genetic drift distort Hardy-Weinberg
proportions, but Mendelian inheritance restores them if mating is random.

\subsection{Evolutionary dynamics under selection}\label{sec:evol_dyn_sel}
Selection occurs when genotypes in a population differ in their fitnesses, i.e., in
their viability, mating success, or fertility and, therefore, leave
different numbers of progeny. The basic mathematical models of selection
were developed and investigated in the 1920s and early 1930s by Fisher (1930), Wright (1931),
and Haldane (1932). 

We will be concerned with the evolutionary consequences of selection caused by differential
viabilities, which leads to simpler models than (general) fertility selection (e.g., Hofbauer and Sigmund 1988, Nagylaki 1992).
Suppose that at an autosomal locus the alleles $\A_1,\ldots,\A_I$ occur. We count
individuals at the zygote stage and denote the (relative) frequency of 
the ordered genotype $\A_i\A_j$ by $P_{ij}(=P_{ji})$.
Since mating is at random, the genotype frequencies $P_{ij}$ are in
Hardy-Weinberg proportions. We assume 
that the fitness (viability) $w_{ij}$ of an $\A_i\A_j$ individual is constant, i.e., independent of time, population size, or genotype frequencies.
In addition, we suppose $w_{ij}=w_{ji}$, as is usually the case. Then the frequency of
$\A_i\A_j$ genotypes among adults that have survived selection is 
$$
	P_{ij}^\ast = \frac{w_{ij}P_{ij}}{\wb} = \frac{w_{ij}p_ip_j}{\wb}\,,
$$
where we have used \eqref{HW}. Here, 
\begin{equation}\label{Wb}
	\wb = \sum_{i,j} w_{ij}P_{ij} = \sum_{i,j} w_{ij} p_i p_j = \sum_i w_ip_i    
\end{equation}
is the \emph{mean fitness} of the population and
\begin{equation}\label{w_i}
	w_i = \sum_j  w_{ij}p_j 
\end{equation}
is the \emph{marginal fitness} of allele $\A_i$. Both are functions of $p=(p_1,\ldots,p_I)^\Tr$.\footnote{Throughout, the superscript $\vphantom{M}^\Tr$ denotes 
vector or matrix transposition.}

Therefore, the frequency of $\A_i$ after selection is 
\begin{equation}
	p_i^\ast = \sum_j P_{ij}^\ast = p_i\frac{w_i}{\wb}\,.
\end{equation}

Because of random mating, the allele frequency $p_i'$ among zygotes of the next generation is also
$p_i^\ast$ \eqref{p_i'=p_i}, so that allele frequencies evolve according to the 
\emph{selection dynamics}
\begin{equation}\label{sel_dyn}
	p_i' = p_i\,\frac{w_i}{\wb}\,, \quad i\in\I \,.
\end{equation}
This recurrence equation preserves the relation
$$
  \sum_i p_i = 1                                               
$$
and describes the evolution of allele frequencies at a single autosomal locus in a diploid population.
We view the selection dynamics \eqref{sel_dyn} as a (discrete) dynamical system on the simplex
\begin{equation}\label{simplex}
   \S_I = \biggl\{p=(p_1,\ldots,p_I)^\Tr\in\Real^I: p_i\ge0 \text{ for every } i\in\I\,,\,
        \sum_i p_i=1 \biggr\}\,.
\end{equation}
Although selection destroys Hardy-Weinberg proportions, random mating re-establishes them. Therefore, \eqref{sel_dyn} is sufficient to study the
evolutionary dynamics. 

The right-hand side of \eqref{sel_dyn} remains unchanged if every $w_{ij}$ is multiplied by the same constant.
This is very useful because it allows to rescale the fitness parameters according to convenience (also their number is reduced by one). Therefore, we will usually
consider relative fitnesses and not absolute fitnesses.

Fitnesses are said to be \emph{multiplicative} if constants $v_i$ exist such that
\begin{equation}\label{mult_fit}
	w_{ij} = v_i v_j
\end{equation}
for every $i,j$. Then $w_i=v_i\bar v$, where $\bar v=\sum_i v_ip_i$, and $\wb=\bar v^2$. Therefore, 
\eqref{sel_dyn} simplifies to
\begin{equation}\label{sel_dyn_hap}
	p_i' = p_i\,\frac{v_i}{\bar v}\,, \quad i\in\I \,,
\end{equation}
which can be solved explicitly because it is equivalent to the linear system $x_i'=v_i x_i$.
It is easy to show that \eqref{sel_dyn_hap} also describes the dynamics of a haploid population if the fitness $v_i$ is assigned to allele $\A_i$.

Fitnesses are said to be \emph{additive} if constants $v_i$ exist such that
\begin{equation}\label{add_fit}
	w_{ij} = v_i+v_j
\end{equation}
for every $i,j$. Then $w_i=v_i+\bar v$, where $\bar v=\sum_i v_ip_i$, and $\wb=2\bar v$. Although this assumption is important (it means absence of
dominance; see Sect.\ \ref{sec:2alleles_dom}), it does not yield an explicit solution of the selection dynamics.

\begin{example}\rm Selection is very efficient. We assume \eqref{mult_fit}. Then the solution of \eqref{sel_dyn_hap} is
\begin{equation}\label{sol_sel_hap}
	p_i(t) = \frac{p_i(0)v_i^t}{\sum_j p_j(0)v_j^t}\,.
\end{equation}
Suppose that there are only two alleles, $\A_1$ and $\A_2$. If $\A_1$ is the wild type and $\A_2$ is a new
beneficial mutant, we may set (without loss of generality!) $v_1=1$ and $v_2=1+s$. Then we obtain from \eqref{sol_sel_hap}:
\begin{equation}
	\frac{p_2(t)}{p_1(t)} = \frac{p_2(0)}{p_1(0)} \left(\frac{v_2}{v_1}\right)^t
		 = \frac{p_2(0)}{p_1(0)}(1+s)^t\,.
\end{equation}
Thus, $\A_2$ increases exponentially relative to $\A_1$. 

For instance, if $s=0.5$, then after 10 generations the frequency of $\A_2$ has increased by a factor of 
$(1+s)^t=1.5^{10}\approx57.7$ relative to $\A_2$. If $s=0.05$ and $t=100$, this factor is $(1+s)^t=1.05^{100}\approx131.5$. 
Therefore, slight fitness differences may have a big long-term effect, in particular, since 100 generations are short on an evolutionary time scale.
\end{example}

An important property of \eqref{sel_dyn} is that \emph{mean fitness is nondecreasing along trajectories (solutions)}, i.e.,
\begin{equation}\label{Wbar_increase}
   \wb' = \wb(p') \ge \wb(p) = \wb\,,
\end{equation}
and equality holds if and only if $p$ is an equilibrium.\footnote{$p$ is called an equilibrium, or fixed point,
of the recurrence equation $p'=f(p)$ if $f(p)=p$. We use the term equilibrium point to emphasize that we consider
an equilibrium that is a single point. The term equilibrium may also refer to a (connected) manifold of 
equilibrium points.} 

A particularly elegant proof was provided by Kingman (1961). 
As noted by Nagylaki (1977), \eqref{Wbar_increase} follows immediately from an inequality of Baum and Eagon (1967)
by noting that \eqref{sel_dyn} can be written as
\begin{equation*}
	p_i' = p_i\pder{\wb}{p_i}\biggl/\sum_j p_j \pder{\wb}{p_j}
\end{equation*}
because $\partial \wb/\partial p_i = 2w_i$.

The statement \eqref{Wbar_increase} is closely related to \emph{Fisher's Fundamental Theorem of Natural Selection}, which Fisher (1930) formulated as follows:
\begin{center}\small
``The rate of increase in fitness of any organism at any time is equal \linebreak
 to its genetic variance in fitness at that time.''
\end{center} 
For recent discussion, see Ewens (2011) and B\"urger (2011).

In mathematical terms, \eqref{Wbar_increase} shows that
$\wb$ is a Lyapunov function. This has a number of important consequences. For instance, complex dynamical
behavior such as limit cycles or chaos can be excluded. All trajectories approach the set of points
$p\in\S_I$ that are maxima of $\wb$. This is a subset of the set of equilibria. From \eqref{sel_dyn}
it is obvious that the equilibria are precisely the solutions of
\begin{equation}
	p_i(w_i-\wb) = 0 \quad\text{for every $i\in\I$}\,.
\end{equation}
We call an equilibrium internal, or fully polymorphic, if $p_i>0$ for every $i$ (all alleles are present). 
The $I$ equilibria defined by $p_i=1$ are called monomorphic because only one allele is present.

The following result summarizes a number of important properties of the selection dynamics. Proofs and
references to the original literature may be found in B\"urger (2000, Chap.\ I.9); see also Lyubich (1992, Chap.\ 9).
\begin{theorem}{}\label{theorem:sel_dyn}

1. If an isolated internal equilibrium exists, then it is uniquely determined.

2. $\hat p$ is an equilibrium if and only if $\hat p$ is a critical point of the restriction of mean 
			fitness $\wb(p)$ to the minimal subsimplex of $S_I$ that contains the positive components of $\hat p$.

3. If the number of equilibria is finite, then it is bounded above by $2^I-1$.

4. An internal equilibrium is asymptotically stable if and only if it is an isolated local maximum of 
			$\wb$. Moreover, it is isolated if and only if it is hyperbolic (i.e., the Jacobian has 
			no eigenvalues of modulus 1).
			
5. An equilibrium point is stable if and only if it is a local, not necessarily isolated, maximum of $\wb$.

6. If an asymptotically stable internal equilibrium exists, then {\rm every} orbit starting in the 
			interior of $S_I$ converges to that equilibrium.

7. If an internal equilibrium exists, it is stable if and only if, counting multiplicities, the fitness matrix
		 ${\mathsf W}=(w_{ij})$ has exactly one positive eigenvalue.

8. If the matrix $\mathsf W$ has $i$ positive eigenvalues, at least $(i-1)$ alleles will be absent at a 
			stable equilibrium.

9. Every orbit converges to one of the equilibrium points $($even if stable manifolds 
			of equilibria exist$)$.
\end{theorem}

\subsection{Two alleles and the role of dominance}\label{sec:2alleles_dom}
For the purpose of illustration, we work out the special case of two alleles. 
We write $p$ and $1-p$ instead of $p_1$ and $p_2$. Further, we use relative fitnesses and assume 
\begin{equation}
	w_{11}=1\,,\; w_{12}=1-hs\,,\;  w_{22}=1-s\,, 
\end{equation}
where $s$ is called the \emph{selection coefficient} and
$h$ describes the degree of dominance. We assume $s>0$.

The allele $\A_1$ is called \emph{dominant} if $h=0$, partially dominant if $0<h<\tfrac12$, 
\emph{recessive} if $h=1$, and partially recessive if $\tfrac12<h<1$. \emph{No dominance} refers to $h=\tfrac12$.
Absence of dominance is equivalent to additive fitnesses \eqref{add_fit}. 
If $h<0$, there is \emph{overdominance} or \emph{heterozygote advantage}.
If $h>1$, there is \emph{underdominance} or \emph{heterozygote inferiority}.

\begin{figure}
\begin{center}
	\includegraphics[width=8.0cm]{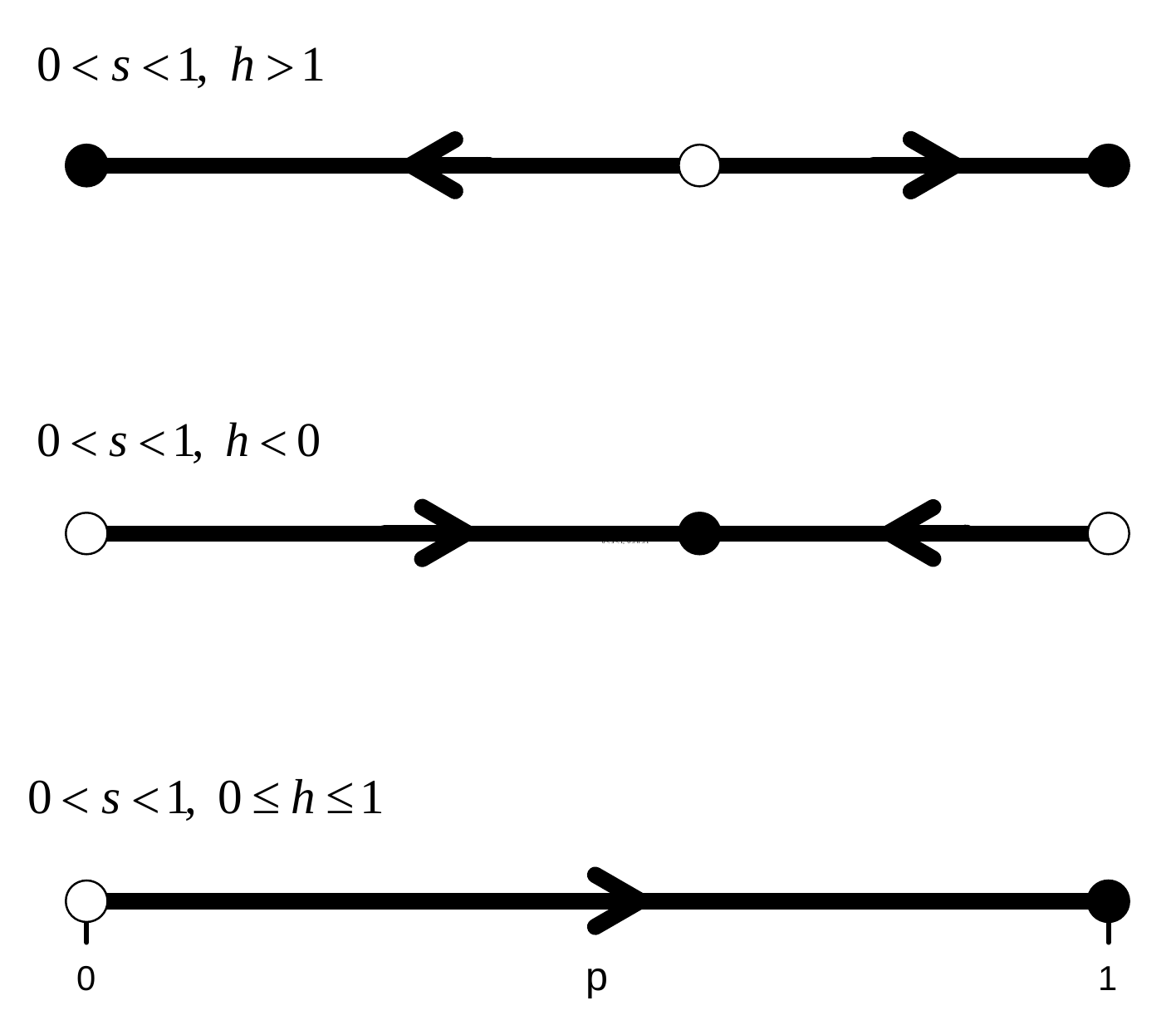}
	\caption{\small Convergence patterns for selection with two alleles.}\label{fig_Conv}
\end{center}
\end{figure}

From \eqref{w_i}, the marginal fitnesses of the two alleles are
$$
	w_1 = 1-hs+hsp \quad\text{and}\quad w_2 = 1-s+s(1-h)p\,         
$$
and, from \eqref{Wb}, the mean fitness is
$$
	\wb = 1-s + 2s(1-h)p - s(1-2h)p^2\,.                          
$$
It is easily verified that the allele-frequency change from one
generation to the next can be written as
\begin{subequations}\label{sel_dyn_2}
\begin{align}
	\De p = p'-p &=\frac{p(1-p)}{2\wb}\,\der{\wb}{p}  \\
  				&= \frac{p(1-p)s}{\wb}\,[1-h-(1-2h)p]\,.             
\end{align}
\end{subequations}

There exists an internal equilibrium if and only if $h<0$ (overdominance) or $h>1$ (underdominance). It
is given by
\begin{equation}
	\hat p = \frac{1-h}{1-2h}\,.
\end{equation}

If dominance is intermediate, i.e., if $0\le h \le 1$, then \eqref{sel_dyn_2} shows that $\De p>0$ if $0<p<1$,
hence $p=1$ is globally asymptotically stable.

If $h<0$ or $h>1$, we write \eqref{sel_dyn_2} in the form
\begin{equation}
	\De p = \frac{sp(1-p)}{\wb}(1-2h)(\hat p -p)\,.
\end{equation}
In the case of overdominance ($h<0$), we have $0<sp(1-p)(1-2h)/\wb<1$ if $0<p<1$, hence
$\hat p$ is globally asymptotically stable and convergence is monotonic. If $h>1$,
then the monomorphic equilibria $p=0$ and $p=1$ each are asymptotically stable and $\hat p$ is unstable.

The three possible convergence patterns are shown in Figure \ref{fig_Conv}. 
Figure \ref{fig_Dominance} demonstrates that the degree of (intermediate) dominance strongly affects
the rate of spread of an advantageous allele.

\begin{figure}
\begin{center}
	\includegraphics[width=12.0cm]{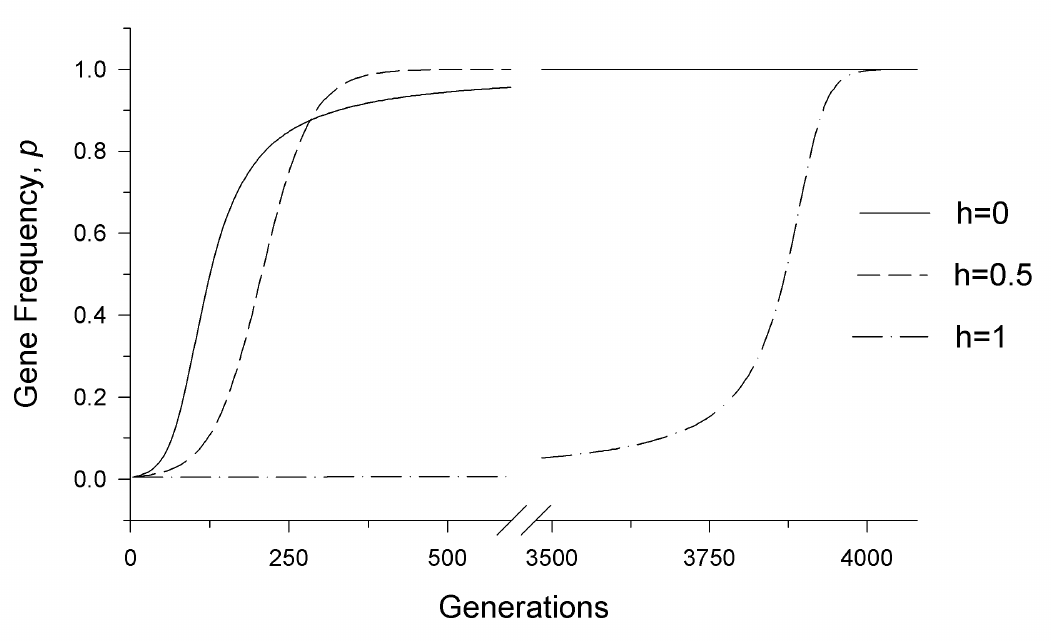}
	\caption{\small Selection of a dominant ($h=0$, solid line), intermediate 
($h=1/2$, dashed), and recessive ($h=1$, dash-dotted) allele. The initial
frequency is $p_0=0.005$ and the selective advantage is $s=0.05$. If the advantageous allele is recessive,
its initial rate of increase is vanishingly small because the frequency $p^2$ of homozygotes is extremely low when $p$ is small. However,
only homozygotes are `visible' to selection.}\label{fig_Dominance}
\end{center}
\end{figure}

\subsection{The continuous-time selection model}\label{sect_cont_time_sel}
Most higher animal species have overlapping generations because birth and death occur
continuously in time. This, however, may lead to substantial
complications if one wishes to derive a continuous-time model from
biological principles. By contrast, discrete-time models can frequently be derived straightforwardly from
simple biological assumptions. If evolutionary forces are weak, a continuous-time version can usually 
be obtained as an approximation to the discrete-time model.

A rigorous derivation of the differential equations describing 
gene-frequency change under selection in a diploid population with overlapping
generations is a formidable task and requires a complex model involving age structure (see Nagylaki 1992,
Chap.\ 4.10). Here, we simply state the system of differential equations and justify it in an alternative way.

In a continuous-time model, the (\emph{Malthusian}) fitness $m_{ij}$ of a genotype $\A_i\A_j$ is defined as its birth rate minus its death rate. 
Then the marginal fitness of allele $\A_i$ is
$$
	m_i = \sum_j m_{ij}p_j\;,                                       
$$
the mean fitness of the population is
$$
	\mb = \sum_i m_ip_i = \sum_{i,j} m_{ij}p_ip_j\;,
$$
and the dynamics of allele frequencies becomes
\begin{equation}\label{sel_cont}
	\dot p_i = \der{p_i}{t} = p_i(m_i-\mb)\,, \quad i\in\I\,. \footnote{Throughout, we use a dot, $\dot{\phantom{x}}$, to indicate derivatives with respect to time.}
\end{equation}
This is the analogue of the discrete-time selection dynamics \eqref{sel_dyn}.
Its state space is again the simplex $S_I$.
The equilibria are obtained from the condition $\dot p_i=0$ for every $i$. We note that \eqref{sel_cont} is a so-called
\emph{replicator equation} (see Hofbauer and Sigmund 1998).

If we set
\begin{equation}\label{w_m}
	w_{ij} = 1 + s m_{ij} \quad\text{for every } i,j\in\I\,,
\end{equation}
where $s>0$ is (sufficiently) small, the difference equation \eqref{sel_dyn} and the differential 
equation \eqref{sel_cont} have the same equilibria. This is obvious upon noting that \eqref{w_m} implies $w_i=1+s m_i$ and $\wb=1+s\mb$.

Following Nagylaki (1992, p.\ 99), we approximate the discrete model \eqref{sel_dyn} by
the continuous model \eqref{sel_cont} under the assumption of weak selection, i.e., small $s$ in \eqref{w_m}.
We rescale time according to $t=\lfloor \ta/s \rfloor$, 
where $\lfloor\,\rfloor$ denotes the closest smaller integer. Then $s$ may be 
interpreted as generation length and, for $p_i(t)$ satisfying the
difference equation \eqref{sel_dyn}, we write $\pi_i(\ta)=p_i(t)$. Then we obtain formally
$$
	\der{\pi_i}{\ta} = \lim_{s\downarrow0} \frac{1}{s}\left[\pi_i(\ta+s)-\pi_i(\ta)\right]
                   = \lim_{s\downarrow0}\frac{1}{s}\left[p_i(t+1)-p_i(t)\right]\;.
$$
From \eqref{sel_dyn} and \eqref{w_m}, we obtain
$p_i(t+1)-p_i(t) = s p_i(t)(m_i-\mb)/(1+s\mb)$. Therefore, $\dot \pi_i = \pi_i(m_i-\mb)$ and
$\De p_i \approx s\dot \pi_i = sp_i(m_i-\mb)$. We note that \eqref{sel_dyn} is essentially
the Euler scheme for \eqref{sel_cont}.

The exact continuous-time model reduces to \eqref{sel_cont} only if the mathematically inconsistent assumption
is imposed that Hardy-Weinberg proportions apply for every $t$ which is generally not true. Under weak
selection, however, deviations from Hardy-Weinberg decay to order $O(s)$ after a short period of time (Nagylaki 1992).

One of the advantages of models in continuous time is that they lead to differential equations, and usually 
these are easier to analyze because the formalism of calculus is available. An example for this
is that, in continuous time, \eqref{Wbar_increase} simplifies to
\begin{equation}
	\dot{\mb} \ge0 \,,
\end{equation}
which is \emph{much} easier to prove than \eqref{Wbar_increase}:
\begin{equation*}
	\dot{\mb} = 2\sum_{i,j}m_{ij}p_j \dot p_i = 2\sum_i m_i \dot p_i = 2\sum_i (m_i^2-\mb^2) p_i = 2\sum_i (m_i-\mb)^2 p_i.
\end{equation*}

\begin{remark}\label{rem:Svir_Shash}\rm
The allele-frequency dynamics \eqref{sel_cont} can be written as a (generalized) gradient system (Svirezhev 1972, Shahshahani 1979):
\begin{equation}\label{Svir_Shash}
   \dot p = G_p \,\grad\bar m
			= G_p \,\left(\frac{\partial\bar m}{\partial p_1},\ldots,
					\frac{\partial\bar m}{\partial p_n}\right)^\Tr\,.
\end{equation}
Here, $G_p =(g^{ij})$ is a quadratic (covariance) matrix, where
\begin{equation}
   g^{ij} = \Cov(f_i,f_j) = \tfrac12 p_i(\de_{ij}-p_j)
\end{equation}
and
\begin{equation}
   f_i(A_kA_l) = \begin{cases} 1          &\quad\text{if } k=l=i\,, \\
			     				\tfrac12   &\quad\text{if $k\ne l$ and $k=i$ or $l=i$}\,, \\
			    				0          &\quad\text{otherwise}\,.
				\end{cases}
\end{equation}

An equivalent formulation is the following covariance form:
\begin{equation}
   \dot p_i = \Cov(f_i,m),
\end{equation}
where $m$ is interpreted as the  random variable $m(A_kA_l)=m_{kl}$ (Li 1967).
It holds under much more general circumstances (Price 1970, Lessard 1997).
\end{remark}

\section{The general migration-selection model}\label{sec:gen_mig-sel}
We assume a population of diploid organisms with discrete, nonoverlapping generations. This population
is subdivided into $\Ga$ demes (niches). Viability selection acts within each deme and is followed by 
adult migration (dispersal). After migration random mating occurs within each deme. We assume that
the genotype frequencies are the same in both sexes (e.g., because the population is monoecious). 
We also assume that, in every deme, the population is so large that gene and genotype frequencies may be 
treated as deterministic, i.e., we ignore random genetic drift. 

\subsection{The recurrence equations}
As before, we consider a single locus with
$I$ alleles $\A_i$ ($i\in\I$). Throughout, we use letters $i,j$ to denote alleles, and greek letters $\a,\be$
to denote demes. We write $\G=\{1,\ldots,\Ga\}$ for the set of all demes. The presentation below is based on
Chapter 6.2 of Nagylaki (1992).

We denote the frequency of allele $\A_i$ in deme $\a$ by $\pia$. Therefore, we have
\begin{equation}
	\sum_i \pia = 1
\end{equation}
for every $\a\in\G$. Because selection may vary among demes, the fitness (viability) $\wija$
of an $\A_i\A_j$ individual in deme $\a$ may depend on $\a$. The marginal fitness of allele $\A_i$ in deme $\a$
and the mean fitness of the population in deme $\a$ are
\begin{equation}\label{marg_mean_fit}
	\wia = \sum_j \wija\pja \quad\text{and}\quad \wba = \sum_{i,j}\wija\pia\pja \,,
\end{equation}
respectively.

Next, we describe migration. Let $\tmab$ denote the probability that an individual in deme $\a$
migrates to deme $\be$, and let $\mab$ denote the probability that an (adult) individual in deme $\a$ immigrated
from deme $\be$. The $\Ga\times\Ga$ matrices
\begin{equation}
	\tilde M =(\tmab) \quad\text{and}\quad M=(\mab)
\end{equation}
are called the \emph{forward} and \emph{backward migration matrices}, respectively. 
Both matrices are \emph{stochastic}, i.e., they are nonnegative and satisfy
\begin{equation}\label{stochasticity}
	\sum_\be \tmab = 1 \quad\text{and}\quad \sum_\be \mab = 1 \quad\text{for every } \a\,.
\end{equation}

Given the backward migration matrix and the fact that random mating within each demes does not change the
allele frequencies, the allele frequencies in the next generation are 
\begin{subequations}\label{mig_sel_dyn}
\begin{equation}\label{mig_dyn}
	\pia' = \sum_\be \mab \pib^\ast\,,
\end{equation}
where
\begin{equation}
	\pia^\ast = \pia \frac{\wia}{\wba}
\end{equation}
\end{subequations}
describes the change due to selection alone; cf.\ \eqref{sel_dyn}. These recurrence equations define
a dynamical system on the $\Ga$-fold Cartesian product $\S_I^\Ga$ of the simplex $\S_I$.
The investigation of this dynamical system, along with biological motivation and interpretation of results, is one of the main purposes of
this survey.

The difference equations \eqref{mig_sel_dyn} require that the backward migration rates are known. 
In the following, we derive their relation to the forward migration rates and discuss conditions when selection or migration do not change
the deme proportions. 

\subsection{The relation between forward and backward migration rates}\label{sec:forw_backw}
To derive this relation, we describe the life cycle explicitly. It starts with zygotes on which selection
acts (possibly including population regulation). After selection adults migrate and usually there is population
regulation after migration (for instance because the number of nesting places is limited). By assumption, population regulation does
not change genotype frequencies. Finally, there is random mating and reproduction, which neither changes gene frequencies (Section \ref{sec:HW}) nor deme proportions. 
The respective proportions of
zygotes, pre-migration adults, post-migration adults, and post-regulation adults in deme $\a$ are
$c_\a$, $c_\a^\ast$, $c_\a^{\ast\ast}$, and $c_\a'$: 

\begin{center}
\setlength{\unitlength}{1cm}
\begin{picture}(14,1)
\thicklines
\put(0,0){Zygote}
\put(1.5,0.1){\vector(1,0){1.5}}
\put(3.2,0){Adult}
\put(4.5,0.1){\vector(1,0){1.5}}
\put(6.2,0){Adult}
\put(7.5,0.1){\vector(1,0){1.5}}
\put(9.2,0){Adult}
\put(10.5,0.1){\vector(1,0){1.5}}
\put(12.3,0){Zygote}
\put(1.65,-0.3){\scriptsize selection}
\put(4.65,-0.3){\scriptsize migration}
\put(7.65,-0.3){\scriptsize regulation}
\put(10.45,-0.3){\scriptsize reproduction}
\put(0,-.8){$c_\a \,,\, \pia$}
\put(3.2,-.8){$c_\a^\ast \,,\, \pia^\ast$}
\put(6.2,-.8){$c_\a^{\ast\ast},\, \pia'$}
\put(9.2,-.8){$c_\a' \,,\, \pia'$}
\put(12.3,-.8){$c_\a' \,,\, \pia'$}
\end{picture}
\end{center}
\vspace{1cm}

Because no individuals are lost during migration, the following must hold:
\begin{subequations}
\begin{align}
	c_\be^{\ast\ast} &= \sum_\a c_\a^\ast \tmab\,, \label{c_a^{astast}}\\
	c_\a^\ast &= \sum_\be c_\be^{\ast\ast} \mba\,.
\end{align}
\end{subequations}
The (joint) probability that an adult is in deme $\a$ and migrates to deme $\be$ can be expressed in terms
of the forward and backward migration rates as follows:
\begin{equation}\label{mfor_mback1}
	c_\a^\ast \tmab = c_\be^{\ast\ast}\mba\,.
\end{equation}
Inserting \eqref{c_a^{astast}} into \eqref{mfor_mback1}, we obtain the desired connection between the
forward and the backward migration rates:
\begin{equation}\label{mfor_mback}
	m_{\be\a} = \frac{c_\a^\ast\tmab}{\sum_\ga c_\ga^\ast \tilde m_{\ga\be}}\,.
\end{equation}
Therefore, if $\tilde M$ is given, an \emph{ansatz} for the vector 
$c^\ast=(c_1^\ast,\ldots,c_\Ga^\ast)^\Tr$ in terms
of $c=(c_1,\ldots,c_\Ga)^\Tr$ is needed to compute $M$ (as well as a hypothesis for the variation, if any, of $c$).

Two frequently used assumptions are the following.

1) \emph{Soft selection}. This assumes that the fraction of adults in every deme is fixed, i.e.,
\begin{equation}\label{soft}
	c_\a^\ast = c_\a \quad\text{for every } \a\in\G\,.
\end{equation}
This may be a good approximation if the population is regulated within each deme, e.g., because individuals
compete for resources locally (Dempster 1955). 

2) \emph{Hard selection}. Following Dempster (1955), the fraction of adults will be proportional to mean
fitness in the deme if the total population size is regulated. This has been called hard selection and is
defined by
\begin{equation}\label{hard}
	c_\a^\ast = c_\a \wba/\wb\,,
\end{equation}
where
\begin{equation}
	\wb = \sum_\a c_\a \wba
\end{equation}
is the mean fitness of the total population.
 
Essentially, these two assumptions are at the extremes of a broad spectrum of possibilities. Soft selection
will apply to plants; for animals many schemes are possible.

Under soft selection, \eqref{mfor_mback} becomes
\begin{equation}
	\mba = \frac{c_\a\tmab}{\sum_\ga c_\ga \tilde m_{\ga\be}}\,.
\end{equation}
As a consequence, if $c$ is constant ($c'=c$), $M$ is constant if and only if $\tilde M$ is
constant. If there is no population regulation after migration, then $c$ will generally depend on time
because \eqref{c_a^{astast}} yields $c' = c^{\ast\ast} = \tilde M^\Tr c$.
Therefore, the assumption of constant deme proportions, $c'=c$,
will usually require that population control occurs after migration.

A migration pattern that does not change deme proportions ($c_\a^{\ast\ast}=c_\a^\ast$) is called
\emph{conservative}. Under this assumption, \eqref{mfor_mback1}
yields 
\begin{equation}
	c_\a^\ast \tmab = c_\be^\ast \mba
\end{equation}
 and, by stochasticity of $M$ and $\tilde M$, we obtain
\begin{equation}
	c_\be^\ast = \sum_\a c_\a^\ast \tmab \quad\text{and}\quad c_\a^\ast = \sum_\be c_\be^\ast \mba \,.
\end{equation}
If there is soft selection and the deme sizes are equal ($c_\a^\ast=c_\a\equiv\text{constant}$), 
then $\mab=\tmba$.

\begin{remark}\rm
Conservative migration has two interesting special cases.

1) Dispersal is called \emph{reciprocal} if the \emph{number} of individuals that migrate from deme $\a$
to deme $\be$ equals the number that migrate from $\be$ to $\a$: 
\begin{equation}\label{reciprocal}
	c_\a^\ast \tmab = c_\be^\ast \tmba\,.
\end{equation}
If this holds for all pairs of demes, then \eqref{c_a^{astast}} and \eqref{stochasticity} 
immediately yield $c_\be^{\ast\ast}=c_\be^\ast$. From \eqref{mfor_mback1}, we infer $\mab=\tmab$, i.e., the forward
and backward migration matrices are identical. 

2) A migration scheme is called \emph{doubly stochastic} if
\begin{equation}
	\sum_\a \tmab = 1 \quad\text{for every } \a\,.
\end{equation}
If demes are of equal size, then \eqref{c_a^{astast}} shows that $c_\a^{\ast\ast}=c_\a^\ast$. Hence,
with equal deme sizes a doubly stochastic migration pattern is conservative. Under soft selection,
deme sizes remain constant without further population regulation. Hence, $\mab=\tmba$ and $M$ is also 
doubly stochastic.

Doubly stochastic migration patterns arise naturally if there is a periodicity, e.g., because the demes
are arranged in a circular way. If we posit equal deme sizes and homogeneous migration, i.e.,
$\tmab = \tilde m_{\be-\a}$ so that migration rates depend only on distance, then the backward migration 
pattern is also homogeneous because $\mab = \tmba = \tilde m_{\a-\be}$ and, hence, depends only on $\be-\a$. 
If migration is symmetric, $\tmab=\tmba$, and the deme sizes are equal, then dispersion is both reciprocal
and doubly stochastic.
\end{remark}

\subsection{Important special migration patterns}\label{sec:special mig patterns I}
We introduce three migration patterns that play an important role in the population genetics and ecological literature.

\begin{example}\label{Deakin_model}\rm
\emph{Random outbreeding and site homing, {\rm or the} Deakin {\rm(1966)} model.} This model assumes that a proportion $\mu\in[0,1]$ of individuals in each deme leaves their deme and is dispersed randomly across all demes. Thus, they perform outbreeding whereas a proportion $1-\mu$ remains at their home site.
If $c_\a^{\ast\ast}$ is the proportion (of the total population) of post-migration adults in deme $\a$, then the forward migration rates are defined by 
$\tmab = \mu c_\be^{\ast\ast}$ if $\a\ne\be$, and $\tilde m_{\a\a} = 1 - \mu + \mu c_\a^{\ast\ast}$.
If $\mu=0$, migration is absent; if $\mu=1$, the \emph{Levene} model is obtained (see below). Because this migration pattern is reciprocal, $M=\tilde M$ holds.

To prove that migration in the Deakin model satisfies \eqref{reciprocal}, we employ \eqref{mfor_mback1} and find
\begin{equation}\label{m_deak1}
	\mba = \frac{c_\a^\ast}{c_\be^{\ast\ast}}\tmab 
		= \begin{cases}
			\mu c_\a^\ast &\text{ if } \a\ne\be\\
			\displaystyle
			\frac{c_\be^\ast}{c_\be^{\ast\ast}}(1-\mu) + \mu c_\be^\ast &\text{ if } \a=\be\,.
			\end{cases}
\end{equation}
From this we deduce
\begin{equation}
	1 = \sum_\a \mba 
	= \sum_{\a\ne\be} \mu c_\a^\ast + \frac{c_\be^\ast}{c_\be^{\ast\ast}}(1-\mu) + \mu c_\be^\ast
	= \mu\cdot1 + (1-\mu)\frac{c_\be^\ast}{c_\be^{\ast\ast}}\,,
\end{equation}
which immediately yields $c_\be^{\ast\ast}={c_\be^\ast}$ for every $\be$ provided $\mu<1$. Therefore, we obtain 
$\tmba = \mu c_\a^{\ast}$ if $\a\ne\be$ and 
$c_\be^\ast\tmba=c_\be^\ast\mu c_\a^{\ast}= c_\a^{\ast}\tmab$, i.e., reciprocity.

We will always assume soft selection in the Deakin model, i.e., $c_\a^{\ast}=c_\a$. Thus, for a given (probability) vector $c=(c_1,\ldots,c_\Ga)^\Tr$, the single parameter $\mu$ is sufficient to describe the migration pattern:
\begin{equation}\label{Deakin}
	\mba = \tmba = \begin{cases}
				\mu c_\a &\text{ if } \a\ne\be\\
				1-\mu + \mu c_\be &\text{ if } \a=\be\,.
				\end{cases}
\end{equation}
\end{example}

\begin{example}\label{Levene_model}\rm
\emph{The Levene {\rm(1953)} model} assumes soft selection and
\begin{equation}\label{def_Levene}
	m_{\a\be} = c_\be\,.
\end{equation}
Thus, dispersing individuals are distributed randomly across all demes in proportion to the deme sizes. In particular, migration is independent of the deme of origin
and $M=\tilde M$.

Alternatively, the Levene model could be defined by $\tilde m_{\a\be} = \mu_\be$, where $\mu_\be>0$ are constants satisfying $\sum_\be\mu_\be=1$. Then 
\eqref{mfor_mback} yields $m_{\a\be} = c_\be^\ast$ for every $\a,\be\in\G$. With soft selection, we get $m_{\a\be} = c_\be$. This is all we need if demes are regulated to constant proportions. But the proportions remain constant even without regulation, for \eqref{c_a^{astast}} gives
$c'_\a=c_\a^{\ast\ast}=\mu_\a$. This yields the usual interpretation $\mu_\a=c_\a$ (Nagylaki 1992, Sect.\ 6.3).
\end{example}

\begin{example}\label{Stepping-stone model}\rm
In the \emph{linear stepping-stone model} the demes are arranged in a linear order and individuals can reach only one of the neighboring
demes. It is an extreme case among migration patterns exhibiting \emph{isolation by distance},
i.e., patterns in which migration diminishes with the distance from the parental deme. In the classical homogeneous version, the forward migration matrix is 
\begin{equation}\label{step_stone_class}
	\tilde M = \begin{pmatrix}
			1-m & m & 0 & \ldots & 0 \\
			m & 1-2m & m & & 0 \\
			\vdots & & \ddots & &\vdots \\
			0 & &  m& 1-2m & m\\
			0 & \ldots & 0 & m & 1-m
		\end{pmatrix}\,.
\end{equation}
We leave it to the reader to derive the backward migration matrix using \eqref{mfor_mback}. It is a special case of the following
general tridiagonal form:
\begin{equation}\label{tridi}
	M = \begin{pmatrix} 
			n_1 & r_1 & 0 & \ldots & 0 \\
			q_2 & n_2 & r_2 & & 0 \\
			\vdots & & \ddots & &\vdots \\
			0 & &  q_{\Ga-1}& n_{\Ga-1} & r_{\Ga-1}\\
			0 & \ldots & 0 & q_\Ga & n_\Ga
		\end{pmatrix}\,,
\end{equation}
where $n_\a\ge0$ and $q_\a+n_\a+r_\a=1$ for every $\a$, $q_\a>0$ for $\a\ge2$, $r_\a>0$ for $\a\le\Ga-1$, and $q_1=r_\Ga=0$.
This matrix admits variable migration rates between neighboring demes. 

If all deme sizes are equal, the homogeneous matrix \eqref{step_stone_class} satisfies $M=\tilde M$, and 
each deme exchanges a fraction $m$ of the population with each of its neighboring demes. The stepping-stone model
has been used as a starting point to derive the partial differential equations for selection and dispersal in continuous space (Nagylaki 1989).
Also circular and infinite variants have been investigated.
\end{example}

\emph{Juvenile migration} is of importance for many marine organisms and plants, where seeds disperse. It can
be treated in a similar way as adult migration. Also models with both juvenile and adult migration
have been studied. Some authors investigated migration and selection in dioecious populations,
as well as selection on X-linked loci (e.g., Nagylaki 1992, pp.\ 143, 144).

Unless stated otherwise, throughout this survey we assume that the backward migration matrix $M$
is constant, as is the case for soft selection if deme proportions and the forward
migration matrix are constant. Then the recurrence equations \eqref{mig_sel_dyn} provide a self-contained 
description of the migration-selection dynamics. Hence, they are sufficient to study evolution for an arbitrary
number of generations.

\section{Two alleles}
Of central interest is the identification of conditions that guarantee the maintenance of genetic diversity.
Often it is impossible to determine the equilibrium structure in detail because establishing existence and,
even more so, stability or location of polymorphic equilibria is unfeasible. Below we introduce an important concept
that is particularly useful to establish maintenance of genetic variation at diallelic loci.
Throughout this section we consider a single locus with two alleles. The number of demes, $\Ga$, can be
arbitrary. 

\subsection{Protected polymorphism}\label{sec:PP}
There is a \emph{protected polymorphism} (Prout 1968) if, independently of the initial
conditions are, a polymorphic population cannot become monomorphic. Essentially, this requires that if an allele
becomes very rare, its frequency must increase. In general, a protected polymorphism is neither
necessary nor sufficient for the existence of a stable polymorphic equilibrium. For instance, on the one
hand, if there is an internal limit cycle that attracts all solutions, then there is a protected polymorphism.
On the other hand, if there are two internal equilibria, one asymptotically stable, the other unstable, 
then selection may remove one of the alleles if sufficiently rare. A generalization of this concept to
multiple alleles would correspond to the concept of permanence often used in ecological models (e.g.,
Hofbauer and Sigmund 1998).

Because we consider only two alleles, we can simplify the notation. We write $\pa=p_{1,\a}$ for the frequency
of allele $\A_1$ in deme $\a$ (and $1-\pa$ for that of $\A_2$ in deme $\a$). Let $p=(p_1,\ldots,p_\Ga)^\Tr$
denote the vector of allele frequencies.
Instead of using the fitness assignments $w_{11,\a}$, $w_{12,\a}$, and $w_{22,\a}$, it 
will be convenient to scale the fitness of the three genotypes in deme $\a$ as follows
\begin{equation}\label{xa_ya}
	\begin{matrix} \A_1\A_1 \quad& \A_1\A_2 \quad& \A_2\A_2 &\\
	                \xa & 1 & \ya &
	\end{matrix}
\end{equation}
($x_\a,y_\a\ge0$).
This can be achieved by setting $\xa=w_{11,\a}/w_{12,\a}$ and $\ya=w_{22,\a}/w_{12,\a}$, provided $w_{12,\a}>0$.

With these fitness assignments, one obtains
\begin{equation}
	w_{1,\a} = 1-\pa +\xa\pa \quad\text{and}\quad \wba = \xa\pa^2 + 2\pa(1-\pa) + \ya(1-\pa)^2\,,
\end{equation}
and the migration-selection dynamics \eqref{mig_sel_dyn} becomes
\begin{subequations}\label{mig_sel_dyn2}
\begin{align}
	\pa^\ast &= \pa w_{1,\a}/\wba \label{mig_sel_dyn21}\\
	\pa' &= \sum_\be \mab \pb^\ast\,. \label{mig_sel_dyn22}
\end{align}	
\end{subequations}
We consider this as a (discrete) dynamical system on $[0,1]^\Ga$.

We call allele $\A_1$ protected if it cannot be lost. Thus, it has to increase in frequency if rare. In
mathematical terms this means that the monomorphic equilibrium $p=0$ must be unstable. To derive a sufficient
condition for instability of $p=0$, we linearize \eqref{mig_sel_dyn2} at $p=0$. If $\ya>0$ for every $\a$ (which means that
$\A_2\A_2$ is nowhere lethal), a simple calculation shows
that the Jacobian of \eqref{mig_sel_dyn21},
\begin{equation}\label{define_D}
	D = \left(\pder{\pa^\ast}{\pb} \right)\Biggl|_{p=0}\,,
\end{equation}
is a diagonal matrix with (nonzero) entries $d_{\a\a}=\ya^{-1}$. Because \eqref{mig_sel_dyn22} is
linear, the linearization of \eqref{mig_sel_dyn2} is
\begin{equation}\label{define_Q}
	p' = Qp\,, \quad\text{where}\quad Q=MD\,,
\end{equation}
i.e., $q_{\a\be}=\mab/\yb$.

To obtain a simple criterion for protection, we assume that the descendants of individuals in every deme be
able eventually to reach every other deme. Mathematically, the appropriate assumption is that $M$ is
irreducible. Then $Q$ is also irreducible and it is nonnegative. Therefore, the Theorem of Perron and Frobenius (e.g., Seneta 1981) 
implies the existence of a uniquely determined eigenvalue $\la_0>0$ of $Q$ such that $\abs{\la}\le\la_0$ holds
for all eigenvalues of $Q$. In addition, there exists a strictly positive eigenvector pertaining to $\la_0$
which, up to multiplicity, is uniquely determined. As a consequence,
\begin{equation}\label{A1_protected}
	\text{$\A_1$ \emph{is protected if } $\la_0 > 1$ \;and\; $\A_1$ \emph{is not protected if } $\la_0 < 1$ }
\end{equation}
(if $\la_0=1$, then stability cannot be decided upon linearization).
This maximal eigenvalue satisfies
\begin{equation}\label{la_0_est}
	\min_\a \sum_\be q_{\a\be} \le \la_0 \le \max_\a \sum_\be q_{\a\be} \,,
\end{equation}
with equality if and only if all the row sums are the same.

\begin{example}\label{ex:A1not_prot}\rm
Suppose that $\A_2\A_2$ is at least as fit as $\A_1\A_2$ in every deme and more fit in at least one deme, i.e.,
$y_\a\ge1$ for every $\a$ and $y_\be>1$ for some $\be$. Then $q_{\a\be}=\mab/\yb\le\mab$ for every $\be$.
Because $M$ is irreducible, there is no $\be$ such that $\mab=0$ for every $\a$. Therefore, the row sums
$\sum_\be q_{\a\be}=\sum_\be \mab/\yb$ in \eqref{la_0_est} are not all equal to one, and we obtain
\begin{equation}
	\la_0 < \max_\a \sum_\be q_{\a\be} \le \max_\a \sum_\be \mab = 1\,.
\end{equation}
Thus, $\A_1$ is not protected, and this holds independently of the choice of the $\xa$, or $w_{11,\a}$.

It can be shown similarly that $\A_1$ is protected if $\A_1\A_2$ is favored over $\A_2\A_2$ in at least
one deme and is nowhere less fit than $\A_2\A_2$.
\end{example}

One obtains the condition for protection of $\A_2$ if, in \eqref{A1_protected}, $\A_1$ is replaced by $\A_2$
and $\la_0$ is the maximal eigenvalue of the matrix with entries $\mab/\xb$.
Clearly, there is a protected polymorphism if both alleles are protected. 

In the case of complete dominance the eigenvalue condition \eqref{A1_protected} cannot be satisfied. Consider, for instance,
protection of $\A_1$ if $\A_2$ is dominant, i.e., $\ya=1$ for every $\a$. Then $q_{\a\be}=\mab$,
$\sum_\be q_{\a\be}=\sum_\be \mab = 1$, and $\la_0=1$. This case is treated in Section 6.2 of Nagylaki (1992).

\subsection{Two demes}\label{Two diallelic demes}
It will be convenient to set
\begin{equation}\label{xy_fit_diall}
	\xa = 1-r_\a \quad\text{and}\quad \ya=1-s_\a \,,
\end{equation}
where $r_\a\le1$ and $s_\a\le1$ for every $\a\in\{1,2\}$. 
We write the backward migration matrix as
\begin{equation}
	M = \begin{pmatrix} 1-m_1 & m_1 \\ m_2 & 1-m_2 \end{pmatrix}\,,
\end{equation}
where $0<m_\a<1$ for every $\a\in\{1,2\}$. 

Now we derive the condition for protection of $\A_1$. The characteristic polynomial of $Q$ is given by
\begin{equation}\label{charpol}
	\ph(x) = (1-s_1)(1-s_2)x^2-(2-m_1-m_2-s_1-s_2+s_1m_2+s_2m_1)x + 1-m_1-m_2\,.
\end{equation}
It is convex and satisfies
\begin{subequations}\label{charpol_vals}
\begin{align}
		\ph(0) &= 1-m_1-m_2 > 0\,, \quad \ph(1)= s_1s_2(1-\ka)\,, \\
		\ph'(1) &= (1-s_1)(m_2-s_2)+(1-s_1)(m_1-s_1)\,,
\end{align}
\end{subequations}
where
\begin{equation}\label{kappa}
     \ka = \frac{m_1}{s_1} + \frac{m_2}{s_2}\,.
\end{equation}

By Example \ref{ex:A1not_prot}, $\A_1$ is not protected if $\A_1\A_2$ is less fit than
$\A_2\A_2$ in both demes (more generally, if $s_1\le0$, $s_2\le0$, and $s_1+s_2<0$). Of course, $\A_1$ will be protected if 
$\A_1\A_2$ is fitter than $\A_2\A_2$ in both demes (more generally, if $s_1\ge0$, $s_2\ge0$, and $s_1+s_2>0$). Hence, we restrict attention to the
most interesting case when $\A_1\A_2$ is fitter than $\A_2\A_2$ in one deme and less fit in the other, i.e., $s_1s_2<0$. 

The Perron-Frobenius Theorem informs us that $\ph(x)$ has two real roots. We have to determine when the
larger ($\la_0$) satisfies $\la_0>1$. From \eqref{charpol} and \eqref{charpol_vals}, we infer that this is the
case if and only if $\ph(1)<0$ or $\ph(1)>0$ and $\ph'(1)<0$. It is straightforward to show that $\ph(1)>0$ and $\ph'(1)<0$ is
never satisfied if $s_1s_2<0$. Therefore, we conclude that, if $s_1s_2<0$, allele $\A_1$ is protected if
\begin{equation}\label{A1_protected_2demes}
     \ka < 1\,;
\end{equation}
cf.\ Bulmer (1972). It is not protected if $\ka>1$.
Figure \ref{AK_Bio_NL_Fig1} displays the region of protection of $\A_1$ for given $m_1$ and $m_2$.

\begin{figure}
\begin{center}
	\includegraphics[width=8.0cm]{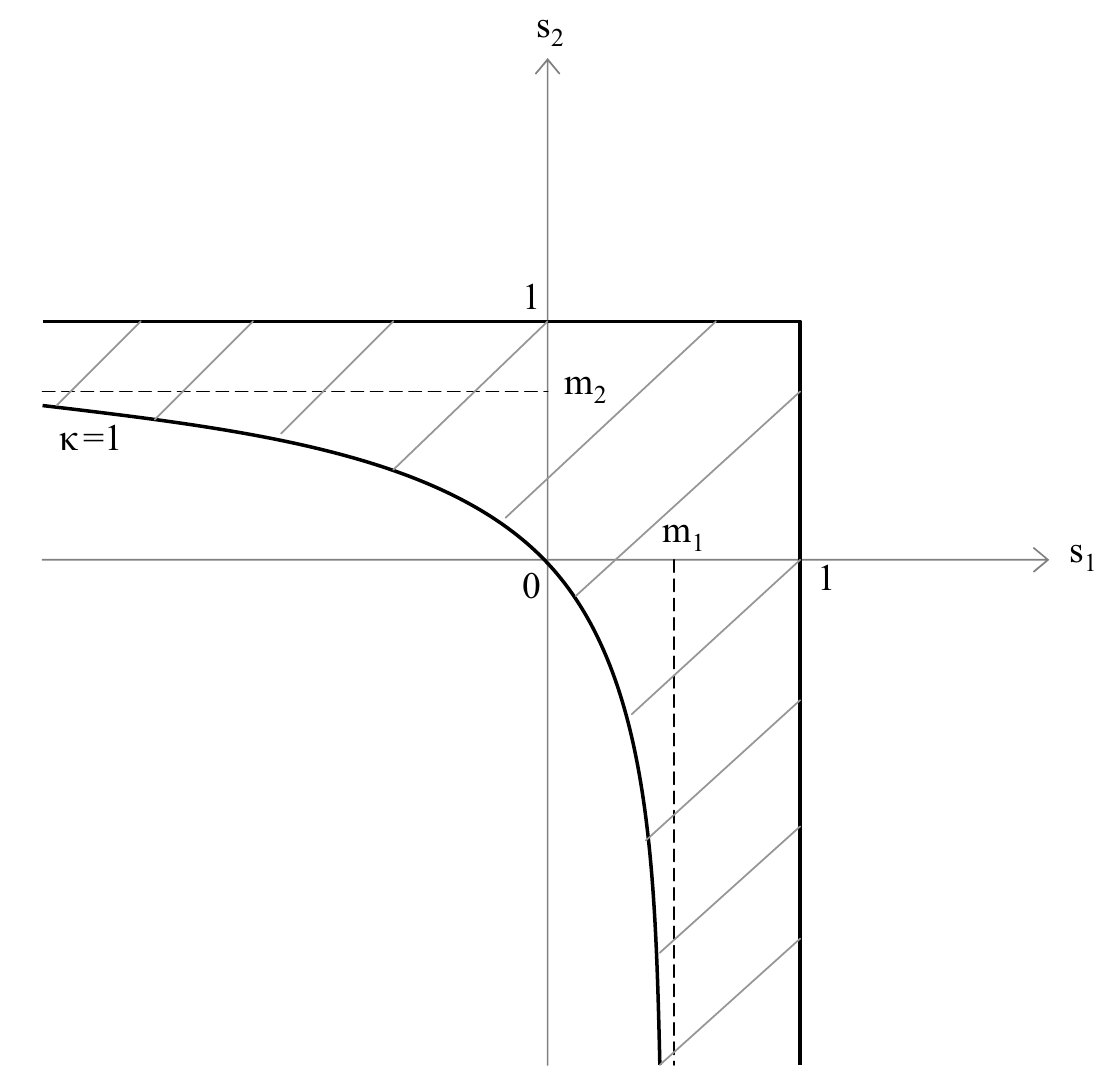}
	\caption{\small The region of protection of $\A_1$ (hatched). From Nagylaki and Lou (2008).}
	 \label{AK_Bio_NL_Fig1}
\end{center}
\end{figure}
\begin{figure}
\begin{center}
	\includegraphics[width=8.0cm]{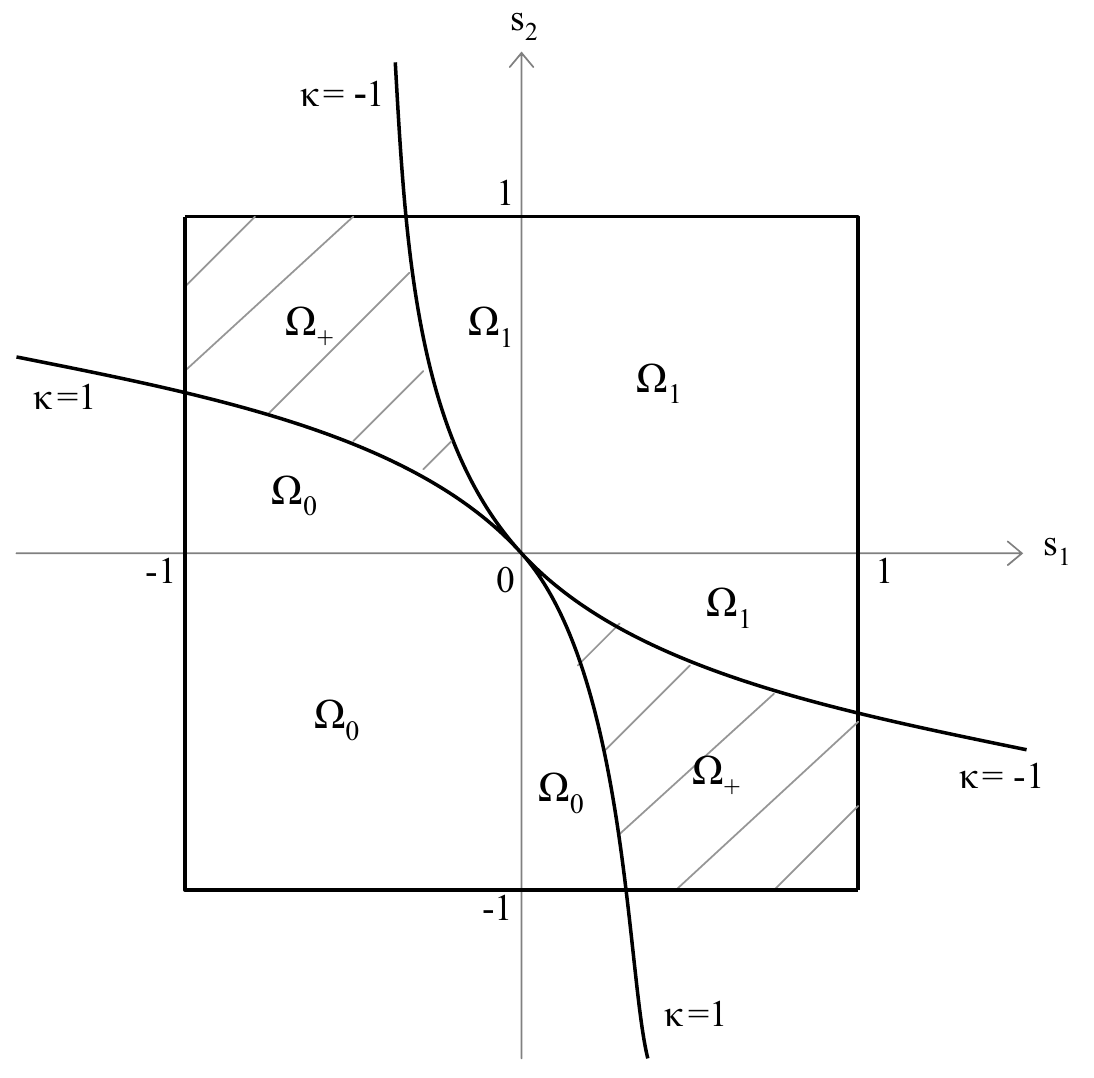}
	\caption{\small The regions of protection of $\A_2$ only ($\Om_0)$, $\A_1$ only ($\Om_1$), and both
	$\A_1$ and $\A_2$ ($\Om_+$) in the absence of dominance. From Nagylaki and Lou (2008).}
	 \label{AK_Bio_NL_Fig2}
\end{center}
\end{figure}

If there is no dominance ($r_\a=-s_\a$ and $0<\abs{s_\a}<1$ for $\a=1$, 2), then further simplification can be 
achieved. From the preceding paragraph the results depicted in Figure \ref{AK_Bio_NL_Fig2}
are obtained. The region of a protected polymorphism is
\begin{equation}\label{PP_nodom}
    \Om_+ = \left\{(s_1,s_2): s_1s_2<0 \text{ and } \abs{\ka}<1 \right\}\,. 
\end{equation}

In a panmictic population, a stable polymorphism can not occur in the absence of overdominance.
Protection of both alleles in a subdivided population requires that selection in the two
demes is in opposite direction and sufficiently strong relative to migration.
Therefore, the study of the maintenance of polymorphism is of most interest if selection acts in opposite direction and dominance is intermediate, i.e., 
\begin{equation}
	r_\a s_\a<0 \;\text{ for } \a=1,2 \;\text{ and } s_1s_2<0.
\end{equation}

\begin{figure}
\vglue-1cm
\begin{center}
	\includegraphics[width=8.0cm]{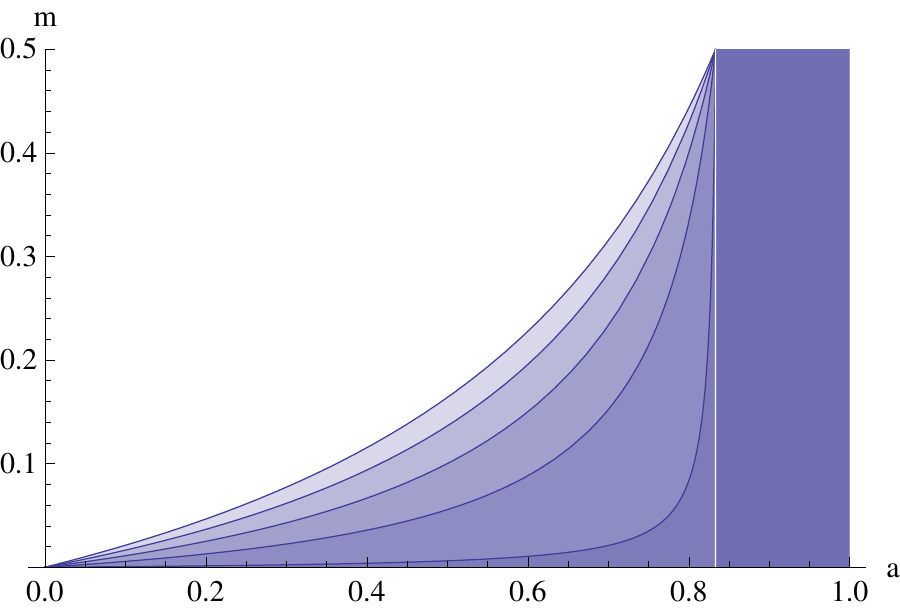}
	\caption{\small Influence of the degree of dominance on the region of protected polymorphism (shaded). The fitness
	scheme is \eqref{Prot_dom1} with $s=0.1$, and migration is symmetric, i.e., $m_1=m_2=m$. The values of $h$ are 
	-0.95, -0.5, 0, 0.5, 0.95 and correspond to the curves from left to right (light shading to dark shading). 
	A protected polymorphism is maintained in the shaded area to the right of the respective curve. 	
	To the right of the white vertical line, a protected polymorphism exists for every $m$.}
	 \label{fig_prot_dom1}
\end{center}
\bigskip
\begin{center} 
	\includegraphics[width=8.0cm]{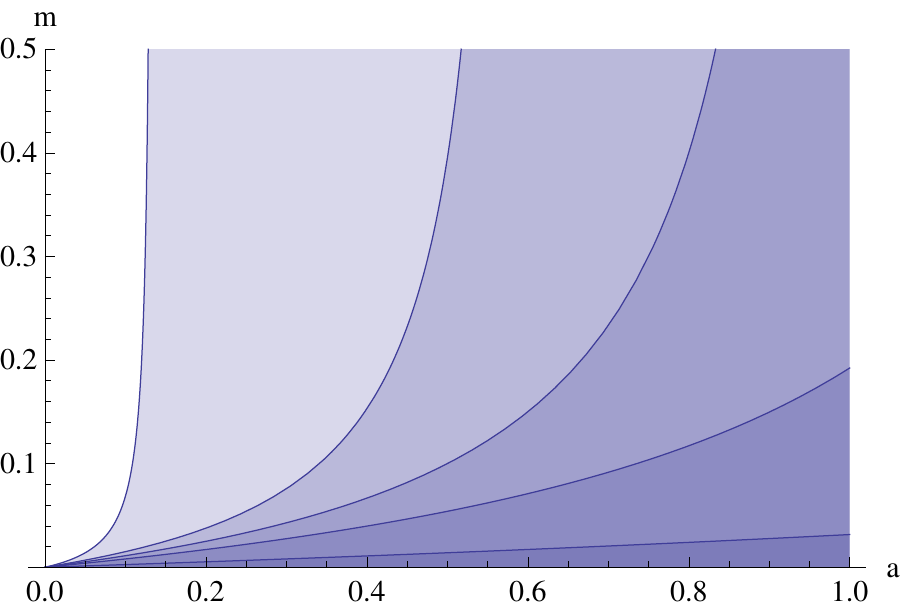}
	\caption{\small Influence of the degree of dominance on the region of protected polymorphism (shaded). The fitness scheme
	is \eqref{Prot_dom2} with $s=0.1$, and migration is symmetric, i.e., $m_1=m_2=m$. The values of $h$ are 
	-0.75, -0.25, 0, 0.25, 0.75 and correspond to the curves from right to left (shading from dark to light!). 
	A protected polymorphism is maintained in the shaded area to the right of the respective curve.}
	 \label{fig_prot_dom2}
\end{center}
\end{figure}

\begin{example}\rm
It is illuminating to study how the
parameter region in which a protected polymorphism exists depends on the degree of dominance in the two
demes. Figures \ref{fig_prot_dom1} and \ref{fig_prot_dom2} display the regions of
protected polymorphism for two qualitatively different scenarios of dominance. In the first, the fitnesses are
given by
\begin{equation}\label{Prot_dom1}
	\begin{matrix} \A_1\A_1 \quad& \A_1\A_2 \quad& \A_2\A_2 &\\
                1+s & 1+hs & 1-s & \\
                1-as & 1-has & 1+as & \,,
	\end{matrix}
\end{equation}
thus, there is deme independent degree of dominance (DIDID). In the second scenario, the fitnesses are given by
\begin{equation}\label{Prot_dom2}
	\begin{matrix} \A_1\A_1 \quad& \A_1\A_2 \quad& \A_2\A_2 &\\
                1+s & 1+hs & 1-s & \\
                1-as & 1+has & 1+as & \,.
	\end{matrix}
\end{equation}
In both cases, we assume $0<s<1$, $0<a<1$ ($a$ is a measure of the selection intensity in deme 2 relative to deme 1), and $-1\le h \le 1$ (intermediate dominance).

If, in \eqref{Prot_dom1}, selection is sufficiently symmetric, i.e., $a>1/(1+2s)$, there exists a protected polymorphism for every
$h\le1$ and every $m\le1$. Otherwise, for given $a$, the critical migration rate $m$ admitting a protected polymorphism decreases with increasing $h$ because this increases the average (invasion) fitness of $\A_1$. This can be shown analytically by studying the principal eigenvalue.

For \eqref{Prot_dom2}, increasing $h$ greatly facilitates protected polymorphism because it leads to an increase of the average fitness of heterozygotes relative to the homozygotes. The precise argument is as follows. If we rescale fitnesses in \eqref{Prot_dom2} according to \eqref{xa_ya}, the matrix
$Q$ in \eqref{define_Q} has the entries $q_{\a1}=m_{\a1}(1+hs)/(1-s)$ and $q_{\a2}=m_{\a1}(1+has)/(1-as)$, which are increasing in $h$. Therefore, the principal eigenvalues $\la_0$ of $Q$ increases in $h$ (e.g., Berman and Plemmons 1994, Chapter 1.3), and \eqref{A1_protected} implies that protection of $\A_1$ is facilitated by increasing $h$. Because an analogous
reasoning applies to $\A_2$, the result is proved.
\end{example}

Indeed, the above finding on the role of dominance for the fitness scheme \eqref{Prot_dom2} is a special case the following result of Nagylaki (personal communication).
Assume fitnesses of $\A_1\A_1$, $\A_1\A_2$, and $\A_2\A_2$ in deme $\a$ are $1+s_\a$, $1+h_\a s_\a$, and 1, respectively, where $s_\a\ge 1$ and $0\le h_\a\le 1$.
If the homozygote fitnesses are fixed, increasing the heterozygote fitness in each deme aids the existence of a protected polymorphism. 
The proof follows by essentially the same argument as above.

\begin{example}\rm 
In the \emph{Deakin model}, the condition \eqref{A1_protected_2demes} for protection of allele $\A_1$
becomes
\begin{equation}\label{PP_Deakin}
	\ka = \mu\left(\frac{c_2}{s_1} + \frac{c_1}{s_2}\right) < 1,
\end{equation}
where $s_1s_2<0$.
Therefore, for given $s_1$, $s_2$, and $c_1$, there is a critical value $\mu_0$ such that allele $\A_1$ is
protected if and only if $\mu<\mu_0$. This implies that for two diallelic demes 
a protected polymorphism is favored by a smaller migration rate. 
\end{example}

\begin{example}\rm 
In the \emph{Levene model}, the condition for a protected polymorphism is
\begin{equation}
	\frac{c_2}{s_1} + \frac{c_1}{s_2} < 1 \quad\text{and}\quad 
	 \frac{c_2}{r_1} + \frac{c_1}{r_2} < 1 \,.
\end{equation}
\end{example}

We close this subsection with an example showing that already with two alleles and two demes the equilibrium structure can be quite complicated.

\begin{example}\label{ex:underdom}\rm
In the absence of migration, the recurrence equations for the allele frequencies $p_1$, $p_2$ in the two demes are two decoupled one-locus selection dynamics
of the form \eqref{sel_dyn_2}. Therefore, if there is underdominance in each deme, the top convergence pattern in Figure \ref{fig_Conv} applies to each deme. As a consequence, in the absence of migration, the complete two-deme system has nine equilibria, four of which are asymptotically stable and the others are unstable.  Under sufficiently weak migration all nine equilibria are admissible and the four stable ones remain stable, whereas the other five are unstable. Two of the stable equilibria are internal. For increasing migration rate, several of these equilibria are extinguished in a sequence of bifurcations (Karlin and McGregor 1972a).
\end{example}

\subsection{Arbitrary number of demes}\label{sec: Ga_diallelic_demes}
The following result is a useful tool to study protection of an allele. Let $I^{(n)}$ and $Q^{(n)}$
respectively designate the $n\times n$ unit matrix and the square matrix formed from the first $n$
rows and columns of $Q$.

\begin{theorem}[Christiansen 1974]\label{Theorem_prot_poly}
If there exists a permutation of demes such that 
\begin{equation}\label{cond_Christ}
	\det(I^{(n)}-Q^{(n)}) < 0
\end{equation}
for some $n$, where $1\le n\le\Ga$, then $\A_1$ is protected. 
\end{theorem}

This theorem is sharp in the sense that if the inequality in \eqref{cond_Christ} is reversed for every 
$n\le\Ga$, then $\A_1$ is not protected.

The simplest condition for protection is obtained from Theorem \ref{Theorem_prot_poly} by setting $n=1$. 
Hence, $\A_1$ is protected if an $\a$ exists such that (Deakin 1972)
\begin{equation}\label{prot_simple}
	m_{\a\a}/\ya > 1\,.
\end{equation}
This condition ensures that, when rare, the allele $\A_1$ increases in frequency in deme $\a$ even if this
subpopulation is the only one containing the allele. The general condition in the above theorem ensures that
$\A_1$ increases in the $n$ subpopulations if they are the only ones that contain it. Therefore,
we get the following simple sufficient condition for a protected polymorphism:
\begin{equation}\label{PP_simple}
	\text{There exist $\a$ and $\be$ such that $m_{\a\a}/\ya > 1$ and $m_{\be\be}/\xb > 1$.}
\end{equation}

If we apply these results to the Deakin model with an arbitrary number of demes, condition \eqref{PP_simple} becomes
\begin{equation}\label{PP_simple_Deakin}
	\text{There exist $\a$ and $\be$ such that $\frac{1-\mu(1-c_\a)}{\ya} > 1$ and 
	$\frac{1-\mu(1-c_\be)}{\xb} > 1$.}
\end{equation}

A more elaborate and less stringent condition follows from Theorem \ref{Theorem_prot_poly} by nice matrix algebra:

\begin{corollary}[Christiansen 1974]\label{corollary_Christ_Deak}
For the Deakin model with $\Ga\ge2$ demes, the following is a sufficient condition for protection of $\A_1$.
There exists a deme $\a$ such that
\begin{equation}\label{KC1}
	1-\ya \ge \mu
\end{equation}
or, if \eqref{KC1} is violated in every deme $\a$, then
\begin{equation}\label{KC2}
	\mu \sum_\a \frac{c_\a}{\mu+\ya-1}  > 1 \,. 
\end{equation}

If \eqref{KC1} is violated for every $\a$ and the inequality in \eqref{KC2} is reversed, then $\A_1$ is not protected.
\end{corollary}

Corollary \ref{corollary_Christ_Deak} can be extended to the inhomogeneous Deakin model, which allows
for different homing rates (Christiansen 1974; Karlin 1982, p.\ 182). 
Using Corollary \ref{corollary_Christ_Deak}, we can generalize the finding from two diallelic demes that a lower degree of outbreeding is favorable 
for protection of one or both alleles. More precisely, we show:

\begin{corollary}\label{protect_monotone}
If $0<\mu_1\le\mu_2\le1$, then allele $\A_1$ is protected for $\mu_1$ if $\mu_2$ satisfies the conditions in
Corollary \ref{corollary_Christ_Deak}.
\end{corollary}

\begin{proof}
If condition \eqref{KC1} holds for $\mu_2$, it clearly holds for $\mu_1$. Now suppose that $1-\ya<\mu_1$ 
for every $\a$ (hence $1-\ya<\mu_2$) and that $\mu_2$ satisfies \eqref{KC2}. This is equivalent to
\begin{equation}
	\sum_\a \frac{c_\a(1-y_\a)}{\mu_1+\ya-1}\, \frac{\mu_1+\ya-1}{\mu_2+\ya-1}
	= \sum_\a \frac{c_\a(1-y_\a)}{\mu_2+\ya-1} > 0\,.
\end{equation}
Because $(\mu_1+\ya-1)/(\mu_2+\ya-1)\ge \mu_1/\mu_2$ if and only if $\ya>1$, we obtain
\begin{equation}
	\frac{\mu_1}{\mu_2}	\sum_\a \frac{c_\a(1-y_\a)}{\mu_1+\ya-1}
	     \ge \sum_\a \frac{c_\a(1-y_\a)}{\mu_2+\ya-1} > 0\,,
\end{equation}
which proves our assertion.
\end{proof}

In general, it is not true that less migration favors the
maintenance of polymorphism. For instance, if the amount of homing ($1-\mu_\a$) varies among demes, a protected polymorphism may be destroyed by decreasing one $\mu_\a$ (Karlin 1982, p.\ 128).

\begin{example}\rm
We apply Corollary \ref{corollary_Christ_Deak} to the Levene model. 
Because \eqref{KC1} can never be satisfied if $\mu=1$, $\A_1$ is protected from loss if  
the harmonic mean of the $y_\a$ is less than one, i.e., if (Levene 1953)
\begin{subequations}\label{PP_gen}
\begin{equation}\label{PP_gen_a}
  y^\ast = \left(\sum_\a \frac{c_\a}{y_\a}\right)^{-1} < 1\,.
\end{equation}
Analogously, allele $\A_2$ is protected if
\begin{equation}\label{PP_gen_b}
  x^\ast = \left(\sum_\a \frac{c_\a}{x_\a}\right)^{-1} < 1\,.
\end{equation}
\end{subequations}
Jointly, \eqref{PP_gen_a} and \eqref{PP_gen_b} provide a sufficient 
condition for a protected polymorphism. If $y^\ast>1$ or $x^\ast>1$, then 
$\A_1$ or $\A_2$, respectively, is lost if initially rare. 

If $\A_1$ is recessive everywhere ($y_\a=1$ for every $\a$), then it is 
protected if (Prout 1968)
\begin{equation}\label{PP_rec_A}
   \bar x = \sum_\a c_\a x_\a > 1\,.
\end{equation}
Therefore, a sufficient condition for a protected polymorphism is
\begin{equation}\label{PP_rec}
    x^\ast < 1 < \bar x\,.
\end{equation}
\end{example}

Whereas in the Deakin model and in its special case, the Levene model, dispersal does not depend on the geographic distribution of niches, in the stepping-stone model it occurs between neighboring niches. How does this affect the maintenance of a protected polymorphism? Here is the answer:

\begin{example}\rm
For the general stepping-stone model \eqref{tridi} with fitnesses given by \eqref{xa_ya}, Karlin (1982)
provided the following explicit criterion for protection of $\A_1$:
\begin{equation}\label{PP_step_stone}
	\sum_\a \frac{\pi_\a}{\ya}\biggl/ \sum_\a \pi_\a\ > 1 \,,
\end{equation}
where $\pi_1=1$ and
\begin{equation}
	\pi_\a = \frac{r_{\a-1}r_{\a-2}\cdots r_1}{q_\a q_{\a-1}\cdots q_2} 
\end{equation}
if $\a\ge2$. This result is a simple consequence of Lemma \ref{Friedland_Karlin}.

For the homogeneous stepping-stone model with equal demes sizes, i.e., $M$ given by \eqref{step_stone_class},
\eqref{PP_step_stone} simplifies to
\begin{equation}
	\frac{1}{\Ga} \sum_\a \frac{1}{y_\a} > 1\,,
\end{equation}
which is the same as condition \eqref{PP_gen} in the Levene model with equal deme sizes. 
\end{example}

Thus, we obtain the surprising result that the conditions for protection are the same in the Levene model and in the homogeneous stepping stone model provided all demes have equal size. 

\subsection{The continent-island model}\label{CI-model}
We consider a population living on an island. At first, we admit an arbitrary number of alleles.
Each generation, a proportion $a$ of adults is removed by
mortality or emigration and a fraction $b$ of migrants with constant allele frequencies $\hat q_i$ is added.
A simple interpretation is that of one-way migration from a continent to an island. The continental population is in equilibrium
with allele frequencies $\hat q_i>0$. Sometimes, $\hat q_i$ is interpreted as the average frequency
over (infinitely) many islands and the model is simply called island model. 

Assuming random mating on the island, the dynamics of allele frequencies $p_i$ on the
island becomes
\begin{equation}\label{island_dyn}
	p_i' = (1-m)p_i \frac{w_i}{\wb} + m \hat q_i\,,
\end{equation}
where $w_i$ is the fitness of allele $\A_i$ on the island and $m=b/(1-a+b)$. Thus, $m$ is the fraction of
zygotes with immigrant parents. Throughout we assume $0<m<1$.

If we define $u_{ij}=m \hat q_j$ and consider $u_{ij}$ as the mutation rate from $\A_i$ to $\A_j$, a special case of the (multiallelic)
mutation-selection model is obtained (the so-called house-of-cards model). Therefore (B\"urger 2000, pp.\ 102-103), \eqref{island_dyn} has the Lyapunov
function
\begin{equation}
	V(p) = \wb(p)^{1-m}\prod_i p_i^{2m \hat q_i}\,.
\end{equation}
It follows that all trajectories are attracted by the set of equilibria. 

In the sequel we determine the conditions under which an `island allele' persists in the population
despite immigration of other alleles from the continent. Clearly, no allele carried to the
island recurrently by immigrants can be lost. 

We investigate the diallelic case and consider alleles $\A_1$ and $\A_2$ with frequencies $p$ and $1-p$ on the island, and $\hat q_2=1$ on the continent. 
Thus, all immigrants are $\A_2\A_2$. For the fitnesses on the island, we assume 
\begin{equation}\label{island_fit}
	w_{11}=1+s\,,\; w_{12}=1+hs\,, \; w_{22}=1-s\,,
\end{equation}
where $0<s\le1$ and $-1\le h\le 1$. Therefore, $\A_1$ evolves according to
\begin{equation}\label{island_p'}
	p' = f(p) = (1-m)p \frac{w_1}{\wb}\,.
\end{equation}

We outline the analysis of \eqref{island_p'}. Since $f(1)=1-m<1$, this confirms that $\A_2$ cannot be lost ($p=1$ is not an equilibrium). Because
\begin{equation}
	f(p) = p (1-m) \frac{w_{12}}{w_{22}} + O(p^2)
\end{equation}
as $p\to0$, the allele $\A_1$ is protected if
\begin{equation}
	m < 1-\frac{w_{22}}{w_{12}}\,.
\end{equation}
Obviously, $p=0$ is an equilibrium. If there is no other equilibrium, then it must be globally asymptotically
stable (as is also obvious from $f(1)<1$, which implies $p'<p$).

The equilibria with $p\ne0$ satisfy
\begin{equation}
	\wb = (1-m)w_1\,,
\end{equation}
which is quadratic in $p$. With the fitnesses \eqref{island_fit}, the solutions are
\begin{equation}\label{equi_island}
	\hat p_\pm = \frac{1}{4h}\left[ 1+3h+m(1-h) \pm \sqrt{(1-h)^2(1+m)^2 + 8hm(1+1/s)} \right]\,,
\end{equation}
These solutions give rise to feasible equilibria if $0\le p_\pm\le1$.
As $h\to0$, \eqref{equi_island} gives the correct limit,
\begin{equation}
	\hat p_- = \frac{1-m/s}{1+m} \quad\text{if } h=0.\
\end{equation}

We define
\begin{subequations}
\begin{align}
	h_0 &= -\frac{1+m}{3-m}\,, \\
	\mu_1 &= 1+h(1-m)\,, \\
	\mu_2 &= -m-\frac{(1-h)^2(1+m)^2}{8h}\,.
\end{align}	
\end{subequations}
Then $\mu_2>\mu_1$ if and only if $h<h_0$. If $\mu_2<\mu_1$, then $\hat p_+$ is not admissible.
As $m$ increases from 0 to $1$, $h_0$ decreases from $-\tfrac13$ to $-1$. Hence, if there is no dominance or
the fitter allele $\A_1$ is (partially) dominant ($0\le h \le 1$), then $h>h_0$. If $\A_1$ is recessive,
then $h<h_0$. For fixed $h$ and $s$ it is straightforward, but tedious, to study the dependence of $\hat p_+$ and $\hat p_-$ on $m$.
The results of this analysis can be summarized as follows (for an illustration, see Figure \ref{Fig_CI-bifplot_1loc}):

\begin{figure}
\begin{center}
	\includegraphics[width=12.0cm]{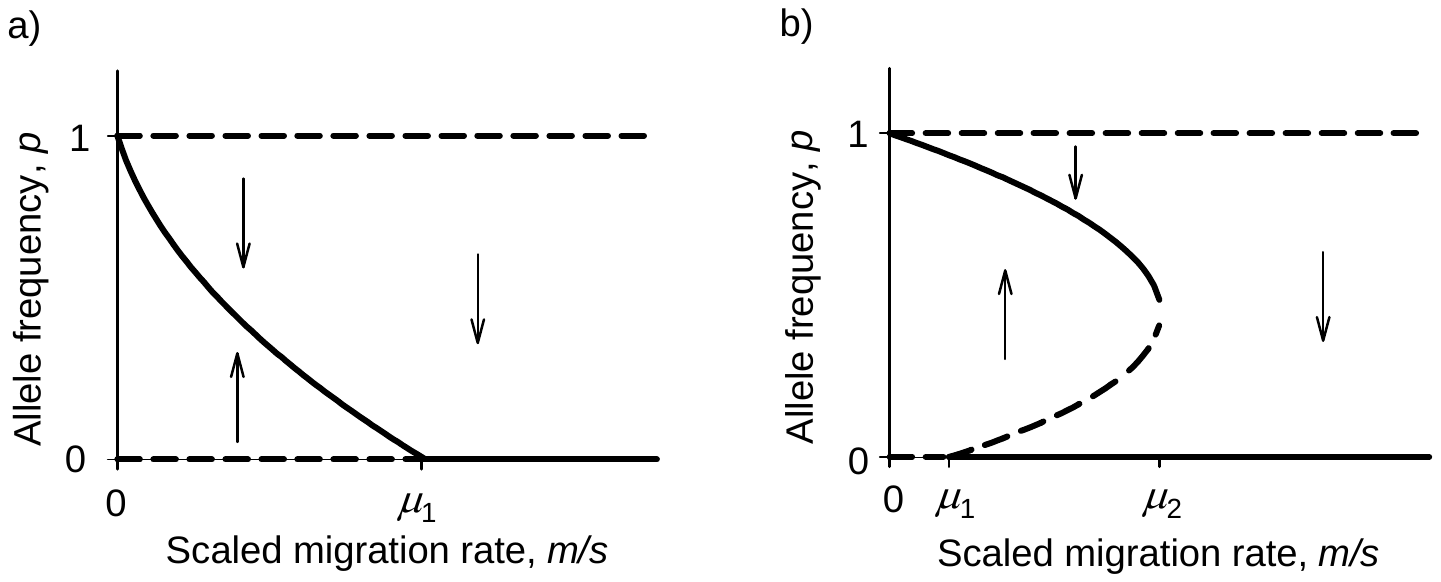}
	\caption{\small Bifurcation patterns for the one-locus continent-island model. Bold lines indicate an asymptotically stable equilibrium, dashed lines an unstable equilibrium. Figure a displays the case $h_0\le h \le1$, and Figure b the case
	$-1\le h < h_0$. The parameters are $s=0.1$, and $h=0.5$ and $h=-0.9$ in cases a and b, respectively. In a, we have $\mu_1\approx 1.41$; in b, $\mu_1\approx0.11$, $\mu_2\approx0.5$}
	 \label{Fig_CI-bifplot_1loc}
\end{center}
\end{figure}

\begin{theorem}\label{th:oneloc_island_model}
(a) Let $m/s<\mu_1$. Then 0 and $p_-$ are the only equilibria and $p(t)\to \hat p_-$ as $t\to\infty$. Thus,
for sufficiently weak migration, a unique polymorphism is established independently of the degree of dominance
and the initial condition.

(b) Let $m/s>\mu_2$ and $-1\le h < h_0$, or $m/s\ge\mu_1$ and $h_0\le h \le1$, i.e., migration is strong relative to selection.
Then there exists no internal equilibrium and $p(t)\to0$ as $t\to0$, i.e., the continental allele $\A_2$ becomes fixed.

(c) Let $\mu_1<m/s\le\mu_2$ and $-1\le h < h_0$, i.e., migration is moderately strong relative to selection and the fitter allele is (almost) recessive. 
Then there are the three equilibria 0, $\hat p_+$, and $\hat p_-$, where $0<\hat p_+ < \hat p_-$, and
\begin{subequations}
\begin{alignat}{2}
	&p(t) \to 0 &\quad\text{if }\; p(0) < \hat p_+ \,, \\
	&p(t) \to \hat p_- &\quad\text{if } \; p(0) > \hat p_+ \,.
\end{alignat}
\end{subequations}
\end{theorem}

For a detailed proof, see Nagylaki (1992, Chapter 6.1). The global dynamics of the
diallelic continent-island model admitting evolution on the continent is derived in Nagylaki (2009a).

Sometimes immigration from two continents is considered (Christiansen 1999). Some authors call our general
migration model with $n$ demes the $n$-island model (e.g., Christiansen 1999). 
The following (symmetric) island model is a standard model in investigations of the consequences of 
random drift and mutation in finite populations:
\begin{subequations}\label{island}
\begin{align}
	m_{\a\a} &= 1 - m \,, \\
	\mab &= \frac{m}{\Ga-1} \;\text{ if } \a\ne\be \,.	
\end{align}
\end{subequations}
Clearly, migration is doubly stochastic in this model and there is no isolation by distance. In the special
case when all deme sizes are equal, i.e., $c_\a=1/\Ga$, \eqref{island} reduces to a special case of the
Deakin model, \eqref{Deakin}, with $m=\mu(1-\Ga^{-1})$. In this, and only this case, the island model
is conservative and the stationary distribution of $M$ is $c=e/\Ga$.

\subsection{Submultiplicative fitnesses}\label{sec:submult}
Here, we admit arbitrary migration patterns but assume that fitnesses are submultiplicative, i.e., the fitnesses in
\eqref{xa_ya} satisfy
\begin{equation}\label{submult}
	x_\a y_\a \le 1
\end{equation}
for every $\a$ (Karlin and Campbell 1980). We recall from Section \ref{sec:evol_dyn_sel} that with multiplicative fitnesses ($x_\a y_\a=1$), the diploid selection dynamics reduces to the haploid dynamics.
Throughout, we denote left principal eigenvector of $M$ by $\xi\in\S_\Ga$; cf.\ \eqref{xi}.

\begin{theorem}[Karlin and Campbell 1980, Result I]\label{KarlinCampbell_Result1}
If $\Ga\ge2$ and fitnesses are submultiplicative in both demes, then at most one monomorphic equilibrium
can be asymptotically stable and have a geometric rate of convergence.
\end{theorem}

The proof of this theorem applies the following very useful result to the conditions for protection of $\A_1$ or $\A_2$ (Section \ref{sec:PP}).

\begin{lemma}[Friedland and Karlin 1975]\label{Friedland_Karlin}
If $M$ is a stochastic $n\times n$ matrix and $D$ is a diagonal matrix with entries $d_i>0$ along the
diagonal, then
\begin{equation}
	\rh(MD) \ge \prod_{i=1}^n d_i^{\xi_i}
\end{equation}
holds, where $\rh(MD)$ is the spectral radius of $MD$.
\end{lemma}

Without restrictions on the migration matrix, the sufficient condition 
\begin{equation}
	\prod_\a \left(1/\ya\right)^{\xi_\a}>1
\end{equation}
for protection of $A_1$ is sharp. This can be verified for a circulant stepping stone model, 
when $\xi_\a = 1/\Ga$.

The conclusion of Theorem \ref{KarlinCampbell_Result1} remains valid under some other conditions (see
Karlin and Campbell 1980). For instance, if migration is doubly stochastic, then $\xi_\a=1/\Ga$ for every $\a$
and $\rh(MD)\ge\prod_\a d_\a^{1/\Ga}$. Therefore, the compound condition
\begin{equation}
	\prod_\a \xa \ya <1
\end{equation}
implies the conclusion of the theorem. 

Here is another interesting result:

\begin{theorem}[Karlin and Campbell 1980, Result III]\label{KarlinCampbell_Result3}
If $\Ga=2$ and fitnesses are multiplicative, then there exists a unique asymptotically stable equilibrium which is globally attracting.
\end{theorem}

Because the proof given by Karlin and Campbell (1980) is erroneous, we present a corrected proof (utilizing
their main idea).

\begin{proof}[Proof of Theorem \ref{KarlinCampbell_Result3}]
Because fitnesses are multiplicative, we can treat the haploid model. Without loss of generality we assume that the fitnesses of alleles $\A_1$ and $\A_2$ in deme 1 are
$s_1\ge1$ and 1, respectively, and in deme 2 they are $s_2$ ($0<s_2<1$) and 1. (We exclude the case when one allele is favored in both demes, hence goes to fixation.) Let the frequency of $\A_1$ in demes 1 and 2 be $p_1$ and $p_2$, respectively. Then the recursion becomes
\begin{subequations}\label{dynamics_haploid_2A2D}
\begin{align}
	p_1' &= (1-m_1) \frac{s_1 p_1}{s_1 p_1+1-p_1} + m_1 \frac{s_2 p_2}{s_2 p_2+1-p_2}\,, \\
    p_2' &= m_2 \frac{s_1 p_1}{s_1 p_1+1-p_1} + (1-m_2)\frac{s_2 p_2}{s_2 p_2+1-p_2} \,. 
\end{align}	
\end{subequations}
Solving $p_1'=p_1$ for $p_2$, we obtain
\begin{equation}
	p_2 = \frac{p_1[(1-p_1)(1-s_1) + s_1m_1]}{p_1(1-p_1)(1-s_1)(1-s_2)+m_1 (s_2+s_1 p_1-s_2 p_1)}\,.
\end{equation}
Substituting this into $p_2'=p_2$, we find that every equilibrium must satisfy
\begin{equation}
	p_1(1-p_1)(1-p_1+s_1 p_1) A(p_1)=0\,,
\end{equation}
where
\begin{subequations}
\begin{align}
	A (p_1) &= m_1 [(1-s_1) (1-s_2+m_2 s_2)+m_1 s_1 (1-s_2)] \notag\\
            &\quad -(1-s_1) \left\{(1-m_2)(1-s_1)(1-s_2) + m_1[1- s_2(1-m_2)+s_1(1-m_2-s_2)]\right\} p_1  \notag\\
            &\quad +(1-m_2)(1-s_1)^2(1-s_2)p_1^2\,.
\end{align}	
\end{subequations}
We want to show that $A(p_1)$ has always two zeros, because then zeros in $(0,1)$ can emerge only 
by bifurcations through either $p_1=0$ or $p_1=1$. Therefore, we calculate the discriminant
\begin{subequations}
\begin{align}
		D &= (1-m_2)^2 \si^2\ta^2 + 2m_1(1-m_2) \si \ta [m_2 (2+\si-\ta)-\si \ta] \notag\\
          &\qquad  +m_1^2 [m_2^2 (\si+\ta)^2-2m_2 \si (2+\si-\ta) \ta+\si^2 \ta^2]\,,
\end{align}
\end{subequations}
where we set $s_1=1+\si$ and $s_2=1-\ta$ with $\si\ge 0$ and $0<\ta<1$.

Now we consider $D$ as a (quadratic) function in $m_1$. To show that $D(m_1)>0$ if $0<m_1<1$, 
we compute
\begin{equation}
	D(0)=(1-m_2)^2\si^2\ta^2 >0 \,, \; D (1)=m_2^2[\ta-\si (1-\ta)]^2 >0\,,
\end{equation}
and
\begin{equation}
	D'(0) = 2 (1-m_2) \si \ta [m_2 (2+\si-\ta)-\si \ta]\,.
\end{equation}
We distinguish two cases.

1. If $m_2 (2+\si-\ta)\ge\si \ta$, then $D'(0)\ge 0$ and there can be no zero in $(0,1)$. 
(If $D$ is concave, then $D(m_1)>0$ on $[0,1]$ because $D(1)>0$; if $D$ is convex, then
$D'(m_1)>0$ holds for every $m_1>0$, hence $D(m_1)>0$ on $[0,1]$.)

2. If $m_2 (2+\si-\ta)<\si\ta$, whence $D'(0)< 0$ and $m_2<\frac{\si \ta}{2+\si-\ta}$, we obtain
\begin{subequations}
\begin{align}
	D'(1) &= 2m_2[-\si \ta (2+\si -\ta-\si \ta)+m_2 (\si^2 (1-\ta)+\ta^2+\si \ta^2)] \notag \\
          &< -2m_2 \frac{4 \si (1+\si) (1-\ta) \ta}{2+\si-\ta} < 0\,.
\end{align}
\end{subequations}
Therefore, $D$ has no zero in $(0,1)$. Thus, we have shown that $A(p_1)=0$
has always two real solutions. 

For the rest of the proof we can follow Karlin and Campbell:
If $s_1=1$ and $s_2 < 1$, there are no polymorphic equilibria, $p_1=0$ is stable and $p_1=1$ is unstable. 
As $s_1$ increases from 1, a bifurcation event at $p_1=0$ occurs before $p_1=1$ becomes stable. At this bifurcation
event, $p_1=0$ becomes unstable and a stable polymorphic equilibrium bifurcates off $p_1=0$. 
Because the constant term $m_1 [(1-s_1) (1-s_2+m_2 s_2)+m_1 s_1 (1-s_2)]$ in $A(p_1)$ is linear in $s_1$, no further
polymorphic equilibria can appear as $s_1$ increases until $p=1$ becomes stable, which then remains stable for
larger $s_1$.

Karlin and Campbell (1980) pointed out that the dynamics \eqref{dynamics_haploid_2A2D} has the following monotonicity property:
If $p_1<q_1$ and $p_2<q_2$, then $p_1'<q_1'$ and $p_2'<q_2'$. By considering $1-p_2$ and $1-q_2$, $p_1<q_1$ and $p_2>q_2$ implies $p_1'<q_1'$ and $p_2'>q_2'$. In particular, if $p_1<p_1'$ and $p_2<p_2'$, then $p_1'<p_1''$ and $p_2'<p_2''$.
Therefore, this monotonicity property implies global monotone convergence to the asymptotically stable equilibrium.
\end{proof}

\begin{remark}\label{rem_converge_2A2D}\rm
Because the monotonicity property used in the above proof holds for an arbitrary number of demes and for selection on diploids,
Karlin and Campbell (1980) stated that ``periodic or oscillating trajectories seem not to occur'' (in diallelic models). However, the argument in the above proof cannot be extended to more than two demes and Akin (personal communication) has shown that for three diallelic demes unstable periodic orbits can occur in the corresponding continuous-time model (cf.\ Section \ref{sec:Mig_sel_cont}). For two diallelic demes, however, every trajectory converges to an equilibrium.
\end{remark}

Karlin and Campbell (1980) showed that for some migration patterns, Theorem \ref{KarlinCampbell_Result3} 
remains true for an arbitrary number of demes. The Levene model is one such example. In fact,
Theorem \ref{RB_condition} is  a much stronger result. It is an open problem
whether Theorem \ref{KarlinCampbell_Result3} holds for arbitrary migration patterns if there
are more than two demes. An interesting and related result is the following.

\begin{theorem}[Karlin and Campbell 1980, Result IV]\label{KarlinCampbell_Result4}
Let $\Ga\ge2$ and let $M$ be arbitrary but fixed.  If there exists a unique, globally attracting equilibrium
under multiplicative fitnesses, then there exists a unique, globally attracting equilibrium for arbitrary
submultiplicative fitnesses.
\end{theorem}

\section{The Levene model}\label{sec:Levene model}
This is a particularly simple model to examine the consequences of spatially varying selection for the maintenance of genetic variability.
It can also be interpreted as a model of frequency-dependent selection (e.g., Wiehe and Slatkin 1998).
From the definition \eqref{def_Levene} of the Levene model and equation \eqref{mig_dyn}, we infer 
\begin{equation}
	p_{i,\a}' = \sum_\be c_\be p_{i,\be}^\ast,
\end{equation}
which is independent of $\a$. Therefore, after one round of migration allele frequencies are the same in all demes, whence
it is sufficient to study the dynamics of the vector of allele frequencies
\begin{equation}
	p = (p_1,\ldots,p_I)^\Tr \in \S_I\,.
\end{equation}
The migration-selection dynamics \eqref{mig_sel_dyn} simplifies drastically and yields the recurrence equation
\begin{equation}\label{dyn_Levene}
	p_i' = p_i \sum_\a c_\a \frac{w_{i,\a}}{\wba}  \quad (i\in\I)
\end{equation}
for the evolution of allele frequencies in the Levene model.
Therefore, there is no population structure in the Levene model although the population experiences spatially varying selection. In particular, distance does not play any role. The dynamics \eqref{dyn_Levene} is also obtained if, instead of the life cycle in Section \ref{sec:forw_backw}, it is assumed that adults form a common mating pool and zygotes are distributed randomly to the demes.

\begin{remark}\label{rem:weak_sel_etc}\rm
1. In the limit of weak selection, the Levene model becomes equivalent to panmixia.
Indeed, with $w_{ij,\a} = 1 + s r_{ij,\a}$, an argument analogous to that below \eqref{w_m} shows that the resulting dynamics is of the form
\begin{equation}
	\dot p_i = p_i \sum_\a c_\a \left(w_{i,\a} - \bar w_\a\right) = p_i(z_i-\bar z)\,,
\end{equation}
where $z_i=\sum_j z_{ij}p_j = \sum_j\sum_\a c_\a r_{ij,\a}p_j$ is linear in $p$ and $\bar z= \sum_i z_i p_i$. 
Here, $z_{ij}=\sum_\a c_\a r_{ij,\a}$ is the spatially averaged selection coefficient of $\A_i\A_j$ (Nagylaki and Lou 2001). 
Therefore, the Levene model is of genuine interest only if selection is strong. 
For arbitrary migration, the weak-selection limit generally does not lead to a panmictic dynamics (Section \ref{sec:Weak_mig}).

2. The weak-selection limit for juvenile migration is panmixia (Lou and Nagylaki 2008).

3. For hard selection \eqref{hard}, the dynamics in the Levene model becomes
\begin{equation}
	p_i' = p_i\, \frac{w_i}{\wb} \quad (i\in\I)\,, \text{ where } w_i = \sum_\a c_\a w_{i,\a}\,,
\end{equation}
which is again equivalent to the panmictic dynamics with fitnesses averaged over all demes. 
Therefore, no conceptually new behavior occurs. 
With two alleles, there exists a protected polymorphism if and only if the averaged fitnesses display
heterozygous advantage. Interestingly, it is not true that the condition for protection of an allele
under hard selection is always more stringent than under soft selection, although under a range of
assumptions this is the case (see Nagylaki 1992, Sect.\ 6.3).
\end{remark}

\subsection{General results about equilibria and dynamics}\label{sec:Levene_general}
We define 
\begin{equation}\label{tw}
	\tw = \tw(p)= \prod_\a {\wba(p)}^{c_\a}\,,
\end{equation}
which is the geometric mean of the mean fitnesses in single demes. Furthermore, we define
\begin{equation}\label{F}
	F(p) = \ln\tw(p)= \sum_\a c_\a \ln\wba(p)\,.
\end{equation}
Both functions will play a crucial role in our study of the Levene model.

From
\begin{equation}
	\pder{\wba}{p_i} = \pder{}{p_i}\sum_{i,j} p_ip_jw_{ij,\a} = 2\sum_j p_j w_{ij,\a} = 2w_{i,\a}\,,
\end{equation}
we obtain
\begin{equation}\label{p_i_pder1}
	p_i' = p_i \sum_\a \frac{c_\a}{\wba}\,\frac{1}{2}\,\pder{\wba}{p_i} 
	 	= \frac{1}{2} p_i \sum_\a c_\a \pder{\ln\wba}{p_i}  
	 	= \frac{1}{2} p_i \pder{F(p)}{p_i}\,.
\end{equation}
Because $\sum_i p_i'=1$, we can write \eqref{p_i_pder1} in the form
\begin{equation}\label{p_i_pder}
	p_i' = p_i \pder{\tw(p)}{p_i}\biggl/\sum_jp_j\pder{\tw(p)}{p_j} 
		= p_i \pder{F(p)}{p_i}\biggl/\sum_jp_j\pder{F(p)}{p_j} \,.
\end{equation}

A simple exercise using Lagrange multipliers shows that the equilibria of \eqref{p_i_pder}, hence of
\eqref{dyn_Levene}, are exactly the critical points of $\tw$, or $F$.
Our aim is to prove the following theorem.

\begin{theorem}[Li 1955; Cannings 1971; Nagylaki 1992, Sect.\ 6.3]\label{Levene_Lyap}
(a) Geometric-mean fitness satisfies $\De\tw(p)\ge0$ for every $p\in\S_I$, 
and $\De\tw(p)=0$ if and only if $p$ is an equilibrium of \eqref{dyn_Levene}.
The same conclusion holds for $F$.

(b) The set $\La$ of equilibria is globally attracting, i.e., $p(t)\to\La$ as $t\to\infty$. If every point
in $\La$ is isolated, as is generic\footnote{We call a property generic if it holds for almost all parameter combinations and, if applicable, for almost all initial conditions}, then $p(t)$ converges to some $\hat p\in\La$. Generically,
$p(t)$ converges to a local maximum of $\tw$.
\end{theorem}

This is an important theoretical result because it states that in the Levene model no complex dynamical behavior,
such as limit cycles or chaos, can occur. One immediate consequence of this theorem is

\begin{corollary}
Assume the Levene model with two alleles. If there exists a protected polymorphism, then all trajectories
converge to an internal equilibrium point.
\end{corollary}

\begin{proof}[Outline of the proof of Theorem \ref{Levene_Lyap}]
(a) We assume $c_\a\in\Rat$ for every $\a$. 
Then there exists $n\in\Nat$ such that $W=\tw^n$ is a homogeneous polynomial in $p$
with nonnegative coefficients. From
\begin{equation*}
	\pder{W}{p_i} = n \tw^{n-1}\pder{\tw}{p_i}\,,
\end{equation*}
we infer
\begin{equation*}
	\pder{W}{p_i}\biggl/\sum_j p_j \pder{W}{p_j} 
		= \pder{\tw}{p_i}\biggl/\sum_j p_j \pder{\tw}{p_j}  \,.
\end{equation*}
From the inequality of Baum and Eagon (1967), 
we infer immediately $W(p')>W(p)$ unless $p'=p$. Therefore, $W$ and $\tw$ are (strict) Lyapunov functions.
By a straightforward but tedious approximation argument, which uses compactness of $\S_I$, it follows that $\tw$ is also a Lyapunov function if $c_\a\in\Real$ (see Nagylaki 1992, Sect. 6.3). This proves the first assertion in (a). The second follows because the logarithm is
strictly monotone increasing.

(b) The first statement is an immediate consequence of LaSalle's invariance principle (LaSalle 1976, p.\ 10). The second
is obvious, and the third follows because $\tw$ is nondecreasing. (Note, however, that convergence to
a maximum holds only generically because some trajectories may converge to saddle points. Moreover, because
$\tw$ may also have minima and saddle points, the globally attracting set $\La$ may not be stable; only a 
subset of $\La$ is stable.)
\end{proof}

Next, we derive a useful criterion for the existence of a unique stable equilibrium. 

\begin{lemma}\label{wba_concave}
If $\wba$ is concave for every $\a$, then $F$ is concave.
\end{lemma}

\begin{proof}
Because the logarithm is strictly monotone increasing and concave, a simple
estimate shows that $\ln\wba$ is concave. Hence, $F$ is concave.
\end{proof}

\begin{theorem}[Nagylaki and Lou 2001]\label{F_concave}
If $F$ is concave, there exists exactly one stable equilibrium (point or manifold) and it is globally
attracting. If there exists an internal equilibrium, it is the global attractor. 
\end{theorem}

\begin{proof}
Concavity and the fact that $\De F\ge0$ imply that if an internal equilibrium exists, it must be the global
attractor.

Suppose now there exist two stable equilibrium points on the boundary of the simplex $\S_I$. Because $F$
is concave, $F$ is constant on the line that joins them. If that line is on the boundary of $\S_I$, then
these two equilibrium points are elements of the same manifold of equilibria (on the boundary of $\S_I$). 
If the line connecting them
is an internal equilibrium, then the case treated above applies.
\end{proof}

If dominance is absent in every deme, there exist constants $v_{i,\a}$ such that
\begin{equation}\label{nodom}
	w_{ij,\a} = v_{i,\a} + v_{j,\a}
\end{equation}
for every $i,j\in\I$ and $\a\in\G$, whence \eqref{marg_mean_fit} yields
\begin{equation}\label{nodom_w}
	w_{i,\a} = v_{i,\a} + \bar v_\a \quad\text{and}\quad \wba = 2\vba\,,
\end{equation}
where
\begin{equation}\label{vba}
	\vba = \sum_i v_{i,\a}p_i\,.
\end{equation}
A simple calculation shows that without dominance the dynamics \eqref{dyn_Levene} simplifies to
\begin{equation}\label{dyn_Levene_nodom}
	p_i' = \tfrac12 p_i \biggl(1 + \sum_\a c_\a \frac{v_{i,\a}}{\vba} \biggr) \quad (i\in\I)\,.
\end{equation}

Fitnesses are called \emph{multiplicative} if there exist constants $v_{i,\a}$ such that
\begin{equation}
	w_{ij,\a} = v_{i,\a} v_{j,\a}
\end{equation}
for every $i,j\in\I$ and $\a\in\G$. Then \eqref{marg_mean_fit} yields
\begin{equation}
	w_{i,\a} = v_{i,\a} \vba  \quad\text{and}\quad \wba = \vba^2 \,.
\end{equation}
With multiplicative fitnesses, one obtains the haploid Levene model,
i.e.,
\begin{equation}
	p_i'=p_i\sum_\a c_\a v_{i,\a}/\vba \quad (i\in\I)\,.
\end{equation}

\begin{theorem}[Nagylaki and Lou 2001]\label{no_dom_no_mig}
The function $F$ is concave in the following cases:

(a) In every deme there is no dominance.

(b) In every deme fitnesses are multiplicative. 

(c) In every deme a globally attracting internal equilibrium exists without migration (i.e., with the dynamics
$p_{i,\a}'=p_{i,\a}w_{i,\a}/\wba$).

Therefore, the conclusions of Theorem \ref{F_concave} apply in each of the cases.
\end{theorem}

\begin{proof}
(a) Without dominance, $\wba=2\bar v_\a$ is linear in every deme. 

(b) With multiplicative fitnesses, $\ln\wba=2\ln\bar v_\a$, which is concave because $\bar v_\a$ is linear.

(c) If there is a globally attracting internal equilibrium in deme $\a$ when it is isolated, then the
fact that $\De\wba\ge0$ implies that the quadratic $\wba$ must be concave. 

In all three cases, the assertion follows from Lemma \ref{wba_concave}.
\end{proof}

\begin{remark}\rm
Theorems \ref{F_concave} and \ref{no_dom_no_mig} hold for hard selection if we replace $F$ by $\bar w$.
For haploids, Strobeck (1979) proved uniqueness and asymptotic stability of an internal equilibrium.
\end{remark}

Nagylaki and Lou (2006b) provide sufficient conditions for the nonexistence of an internal equilibrium.

The following theorem shows that a globally stable internal equilibrium can be maintained on an open set of parameters, even in the absence of overdominance.
Clearly, this requires that there are at least two demes; cf.\ Section \ref{sec:2alleles_dom}.
We say there is partial dominance in every deme if
\begin{equation}\label{partial_dom}
	w_{ii,\a} < w_{ij,\a} < w_{jj,\a} \quad\text{or}\quad 	w_{jj,\a} < w_{ij,\a} < w_{ii,\a}	
\end{equation}
holds for every pair $i,j\in\I$ with $i\ne j$ and for every $\a\in\Ga$.

\begin{theorem}\label{int_dom_fully}
Assume an arbitrary number of alleles and $\Ga\ge2$ demes. Then there exists a nonempty open set of fitness parameters exhibiting partial dominance 
such that for every parameter combination in this set, there is a unique, internal, asymptotically stable equilibrium
point of the dynamics \eqref{dyn_Levene}. This equilibrium is globally asymptotically stable.
\end{theorem}

This theorem arises for the special case of a single locus in Result 5.1 in B\"urger (2010). The latter is an immediate consequence of Theorem 3.1 in B\"urger (2010) and of Theorem 2.2, Remark 2.3(iii), and Remark 2.4 in B\"urger (2009b); cf.\ Theorem \ref{poly_weak_sel_th} below. The proof shows that the maintenance of a stable internal
equilibrium requires a certain form of spatially averaged overdominance that, with intermediate dominance, can be achieved only if the direction of selection and the degree of dominance vary among demes.

\subsection{No dominance}\label{sec:Levene_nodom}
Throughout this subsection, we assume no dominance, i.e., \eqref{nodom}. For this important
special case more detailed results can be proved than for general (intermediate) dominance.
Nevertheless, internal equilibrium solutions can be determined explicitly only under specific assumptions (e.g., Theorem 4.2 in Nagylaki and Lou 2006b).
An interesting question to ask is: What determines the number of alleles that can be maintained at an 
equilibrium? If there is no dominance, there is a simple answer.

\begin{theorem}[Nagylaki and Lou 2001]\label{no_dom_bound}
Without dominance, the number of demes is a generic upper bound on the number of alleles present
at equilibrium. Any neutral deme should not be counted in this bound. 
\end{theorem}

\begin{proof}
From \eqref{dyn_Levene_nodom} we conclude that, at an internal equilibrium $\hat p$, 
\begin{equation}\label{nodom_internal}
	\sum_\a c_\a \frac{v_{i,\a}}{\hat{\bar v}_\a} = 1 \quad (i\in\I)
\end{equation}
holds. The substitution 
\begin{equation}\label{subst_x}
	x_\a = c_\a/\hat{\bar v}_\a\,, \;\text{ where } \hat{\bar v}_\a=\sum_i v_{i,\a}\hat p_i\,,
\end{equation}
linearizes \eqref{nodom_internal}:
\begin{equation}\label{nodom_linearized}
	\sum_\a v_{i,\a}x_\a =1 \quad (i\in\I)\,.
\end{equation}
This is a system of $I$ inhomogeneous linear equations for the $\Ga$ unknowns $x_\a$. Therefore, a
solution exists generically only if $I\le\Ga$. (Note
that even if there exists a solution $(x_\a)$, it does not necessarily give rise to a solution $\hat p\in\S_I$.)

The statement about neutral demes follows because in a neutral deme $v_{i,\a}=\vba=1$ for every $i$,
hence $x_\a=c_\a$ holds.
\end{proof}

\begin{remark}\rm
If in Theorem \ref{no_dom_bound} general intermediate
dominance is admitted, then there exists an open set of parameters for which any number of alleles can
be maintained at an asymptotically stable equilibrium. Theorem \ref{weak_sel_th} provides a much more general result.
\end{remark}

\begin{remark}\rm
Theorem \ref{no_dom_bound} holds for hard selection. 
A slight modification of the proof shows that it holds also for multiplicative fitnesses.
In fact, Strobeck (1979) proved an analog of Theorem \ref{no_dom_bound} for haploid species.
\end{remark}

\begin{example}[Nagylaki and Lou 2001] \rm 
The upper bound established in Theorem \ref{no_dom_bound} can be
assumed if $\Ga=I$. Let $v_{i,\a}=u_i\de_{i\a}$, where $u_i>0$ and $\de_{i\a}$ is the Kronecker delta. 
This means that allele $\A_i$ has fitness $u_i$ in deme $i$ and fitness 0 elsewhere. Hence, every allele
is the best in one deme. Then \eqref{dyn_Levene_nodom} simplifies to
\begin{equation}
	p_i' = \tfrac12 p_i \biggl(1 + \frac{u_i}{\bar v_i} \biggr) 
			= \tfrac12 p_i (1+c_i/p_i) = \tfrac12 (p_i+c_i)  \,.
\end{equation}
The solution is $p_i(t)= c_i + (p_i(0)-c_i)(\tfrac12)^t \to c_i$ as $t\to\infty$.
\end{example}

In the formulation of Theorem \ref{no_dom_bound}, the word `generic' is essential. 
If $\Ga$ and $I$ are arbitrary and one assumes
\begin{equation}
	w_{ij,\a} = 1+r_{ij} g_\a
\end{equation}
for (sufficiently small) constants $r_{ij}$ and $g_\a$, it can be shown that an internal (hence globally
attracting) manifold of equilibria exists for an open set of parameter combinations. This
holds also for the additive case, when $r_{ij}=s_i+s_j$ (Nagylaki and Lou 2001).

Another interesting problem is to determine conditions when a specific allele will go to fixation. A simple and intuitive
result is the following:

\begin{theorem}[Nagylaki and Lou 2006b]\label{nodom_loss}
Suppose there exists $i\in\I$ such that
\begin{equation}
	\sum_\a c_\a \frac{v_{j,\a}}{v_{i,\a}} < 1
\end{equation}
for every $j\ne i$. Then $p_i(t)\to1$ as $t\to\infty$.
\end{theorem}

A proof as well as additional results on the loss or fixation of alleles may be found in Nagylaki and Lou (2006b).
The following example is based on a nice application of Theorem \ref{no_dom_no_mig} and shows that
migration may eliminate genetic variability.

\begin{figure}
\begin{center}
	\includegraphics[width=8.0cm]{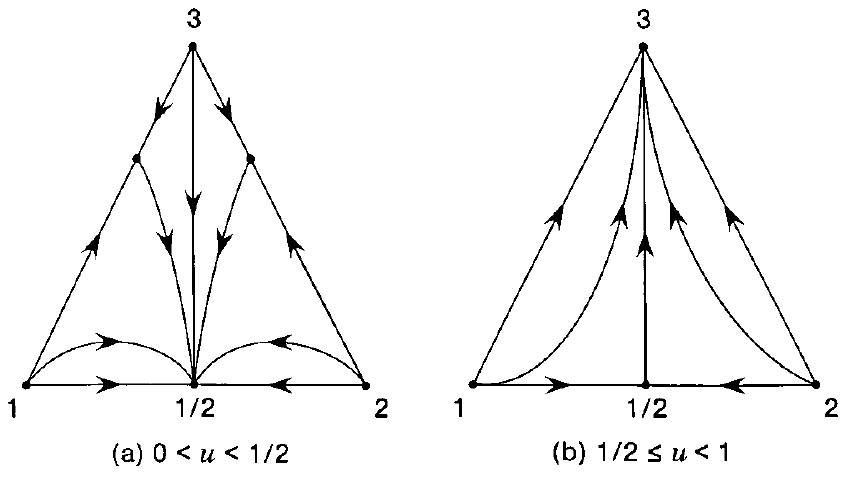}
	\caption{\small Phase portrait for the Levene model in Example \ref{EX_NL01} (Figure from Nagylaki and Lou 2001).}
	 \label{Fig_NL01}
\end{center}
\end{figure}

\begin{example}[Nagylaki and Lou 2001]\label{EX_NL01}\rm
We suppose two demes of equal size ($c_1=c_2=\tfrac12$) and three alleles without dominance. The alleles
$\A_1$ and $\A_2$ have extreme fitnesses and $\A_3$ is intermediate in both demes. More precisely,
we assume $v_{1,1}=1$, $v_{2,1}=0$, $v_{3,1}=u$, and $v_{1,2}=0$, $v_{2,2}=1$, $v_{3,2}=u$, where $0<u<1$.
Without migration, $\A_1$ is ultimately fixed in deme 1 and $\A_2$ is fixed in deme 2. We now establish
that the Levene model evolves as sketched in Figure \ref{Fig_NL01}, i.e.,
\begin{equation}
	p(t) \to \begin{cases}  
				(\tfrac12,\tfrac12,0) \quad\text{if } 0<u<\tfrac12\,, \\
				 (0,0,1) \quad\text{if } \tfrac12 \le u<1\,.
	\end{cases}
\end{equation}
Because there is no dominance, Theorem \ref{no_dom_no_mig} informs us that there can be only 
one stable equilibrium. Therefore, it is sufficient to establish asymptotic stability of the limiting
equilibrium, which we leave as an easy exercise.
\end{example}

Following Nagylaki (2009a), we say there is \emph{deme-independent degree of intermediate dominance} (DIDID) if 
\begin{subequations}\label{DIDID}
\begin{equation}
  w_{ij,\a} = \vth_{ij} w_{ii,\a}
       + \vth_{ji} w_{jj,\a}
\end{equation}
holds for constants $\vth_{ij}$ such that
\begin{equation}\label{DID_sym}
  0 \le \vth_{ij} \le 1 \quad\text{and}\quad \vth_{ji} = 1 - \vth_{ij}    
\end{equation}
\end{subequations}
for every $\a$ and every pair $i,j$. In particular, $\vth_{ii}=\tfrac12$.

Obviously, DIDID covers complete dominance or recessiveness ($\vth_{ij}=0$ or $\vth_{ij}=1$ if $i\ne j$), 
and no dominance ($\vth_{ij}=\tfrac12$), but not multiplicativity.  We also note that
DIDID includes the biologically important case of absence of genotype-by-environment interaction. In general, 
$F$ is not concave under DIDID. However, Nagylaki (2009a, Theorem 3.2) proved that with DIDID the evolution
of $p(t)$ is qualitatively the same as 
without dominance. Therefore, the conclusion of Theorem \ref{F_concave} holds,
i.e., under DIDID there exists exactly one stable equilibrium (point or manifold) and it is globally
attracting. If there exists an internal equilibrium, it is the global attractor. Moreover, the number
of demes is a generic upper bound on the number of alleles that can segregate at an equilibrium 
(generalization of Theorem \ref{no_dom_bound}). This contrasts sharply with Theorem \ref{int_dom_fully}.
Finally, the condition for loss of an allele in Theorem \ref{nodom_loss}
has a simple generalization to DIDID. For proofs and further results about DIDID consult Nagylaki (2009a).

\begin{remark}\label{rem:Nodom_DIDID}\rm
For general migration, the dynamics under DIDID may differ qualitatively from that under no dominance. For
instance, in a diallelic model with one way migration, two asymptotically
stable equilibria may coexist if there is DIDID (Nagylaki 2009a). Moreover, in two triallelic demes,
a globally asymptotically stable internal equilibrium exists for an open set of parameters (Peischl 2010). Therefore, with arbitrary migration and DIDID,
the number of alleles maintained at a stable equilibrium can be higher than the number of demes. This is not the case in the absence of dominance
when the number of demes is a generic upper bound (Theorem \ref{th:NL01_nodom}).
\end{remark}

\subsection{Two alleles with dominance}
Here, we specialize to two alleles but admit arbitrary dominance.
We assume that fitnesses are given by \eqref{xa_ya}.
Our first result yields a further class of examples when a unique stable equilibrium exists.

\begin{theorem}[B\"urger 2009c]\label{RB_condition}
For every $\a$, let the fitnesses satisfy
\begin{subequations}\label{RB_cond}
\begin{align}
   &x_\a y_\a \le 1 + (1-y_\a)^2 \text{ and } y_\a \le1 \label{RB_cond1} \\
\noalign{\vskip-5truemm}
\intertext{or}
\noalign{\vskip-5truemm}
   &x_\a y_\a \le 1 + (1-x_\a)^2 \text{ and } x_\a \le1 \,. \label{RB_cond2}
\end{align}
\end{subequations}
Then $F$ is concave on $[0,1]$. Hence, there exists at most one 
internal equilibrium. If an internal equilibrium exists, it is globally asymptotically stable. 
If a monomorphic equilibrium is stable, then it is globally asymptotically stable. 
\end{theorem}

In the proof it is shown that $\ln\wba$ is concave if and only if \eqref{RB_cond} is fulfilled. 
Then the conclusion follows from Theorem \ref{F_concave}.

Theorem \ref{RB_condition} shows that the protection conditions \eqref{PP_gen} imply the existence of a globally asymptotically stable internal equilibrium if \eqref{RB_cond} holds. It generalizes Result I in Karlin (1977), who proved uniqueness of an internal equilibrium and global convergence 
under the assumption of submultiplicative fitnesses \eqref{submult}. His proof is based on a different method.
Submultiplicative fitnesses, hence fitnesses satisfying \eqref{RB_cond}, include a number of important cases:
multiplicative fitnesses (hence, selection on haploids), no dominance,
partial or complete dominance of the fitter allele, and overdominance.

\begin{example}\label{ex:PP_nodom}\rm
If there are two niches, no dominance, and fitnesses are given by $1+s_\a$, 1, and $1-s_\a$ ($s_1s_2<0$), 
then \eqref{PP_nodom}
informs us that there is protected polymorphism if
\begin{equation}
	\abs{\ka}<1\,, \quad\text{where } \ka = \frac{c_1}{s_2}+\frac{c_2}{s_1}\,.
\end{equation}
By Theorem \ref{no_dom_no_mig} or \ref{RB_condition} there is a unique, asymptotically stable equilibrium $\hat p$.
A simple calculation yields $\hat p = (1-\ka)/2$.
\end{example}

\begin{example}\label{ex:PP_multfit}\rm
If there are two niches and multiplicative fitnesses given by $1/(1-s_\a)$, 1, and $1-s_\a$ ($s_1s_2<0$), 
then there is a protected polymorphism if
\begin{equation}
	0<1-\ka<1\,, 
\end{equation}
where $\ka$ is as above. The unique, asymptotically stable equilibrium is given by $\hat p = 1-\ka$.
\end{example} 

\begin{figure}
\vglue-.5cm
\begin{center}
	\includegraphics[width=8.0cm]{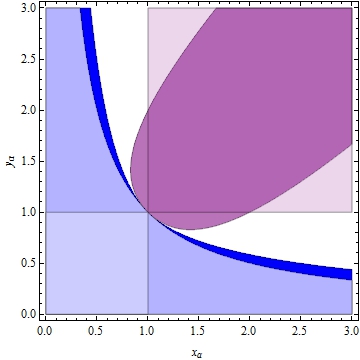}
	\caption{\small In the light blue regions fitnesses are submultiplicative, in the dark blue and light blue regions they
	satisfy \eqref{RB_cond}. In the blue region(s), there exists a single stable equilibrium. In the purple
	region, fitnesses satisfy \eqref{Karlin_cond}. There, both boundary equilibria are stable provided an
	internal equilibrium exists. In the light blue square there is overdominance, in the light purple square there is underdominance.}
	 \label{rb_cond_fig}
\end{center}
\end{figure}

A partial converse to Theorem \ref{RB_condition} is the following:

\begin{theorem}[Karlin 1977, Result IA]\label{Karlin 1977, Result IA}
Suppose
\begin{equation}\label{Karlin_cond}
	x_\a y_\a > 1 + \max\{(x_\a-1)^2, (y_\a-1)^2\} \quad\text{for every } \a\,,
\end{equation}
and there exists at least one internal equilibrium. Then both monomorphic equilibria are asymptotically
stable and the internal equilibrium is unique.
\end{theorem}

Figure \ref{rb_cond_fig} displays the regions in which submultiplicativity, \eqref{RB_cond}, or \eqref{Karlin_cond} hold. Table \ref{tab:Levene_1loc} presents numerical results which demonstrate that stronger selection, submultiplicativity, and a higher number of demes all facilitate protected polymorphism in the Levene model. These findings agree with intuition. For instance, if an allele has sufficiently low fitness in just one deme, i.e., $c_\a/y_\a>1$, the other allele is protected.

\begin{table}[ht]
	\centering\small
		\begin{tabular}{ccccccc}
		\hline
			$\Ga$ & \multicolumn{2}{c}{weak selection} & \multicolumn{2}{c}{moderate selection} &\multicolumn{2}{c}{strong selection} \\
			      & intdom & submult & intdom & submult & intdom & submult\\
			\hline
			2  &0.094 & 0.185 & 0.115 & 0.271 & 0.220 & 0.332 \\
			3  &0.124 & 0.262 & 0.162 & 0.416 & 0.357 & 0.516 \\
			4  &0.138 & 0.310 & 0.190 & 0.512 & 0.460 & 0.636 \\
			10 &0.158 & 0.468 & 0.270 & 0.788 & 0.787 & 0.917 \\		
		\hline
		\end{tabular}
	\caption{\small Proportion of protected polymorphism in the Levene model with intermediate dominance. Fitnesses of $\A_1\A_1$, $\A_1\A_2$, $\A_2\A_2$ are parameterized as $1+s_\a$, $1+h_\a s_\a$, $1-s_\a$. The data for weak selection are generated for the weak-selection limit (Remark \ref{rem:weak_sel_etc}). Then there is a protected polymorphism if and only if there is average overdominance, i.e., if 
$\sum_\a c_\a s_\a h_\a > \abs{\sum_\a c_\a s_\a}$ holds. For moderate or strong selection, the conditions \eqref{PP_gen} were evaluated.
In all cases, $10^6$ parameter combinations satisfying $s_\a\in(-s,s)$, $h_\a\in(-1,1)$, $c_\a\in(-1,1)$ were randomly (uniformly) chosen, and the values $c_\a$ were normalized. For weak or strong selection, $s=1$; for moderate selection, $s=0.2$. For submultiplicative fitnesses (columns `submult'), the $10^6$ parameter combinations satisfy the additional constraint \eqref{submult}, which is reformulated as a condition for $h_\a$.
The columns `intdom' contain the data for general intermediate dominance. The proportion of submultiplicative parameter combinations among all parameter combinations with intermediate dominance is $(1-\tfrac12\ln2)^\Ga\approx0.653^\Ga$.}
\label{tab:Levene_1loc}
\end{table}

So far, we derived sufficient conditions for a unique (internal) equilibrium, but we have not
yet considered the question of how many (stable) internal equilibria can coexist. This turns out to be
a difficult question which was solved only recently for diallelic loci.

With fitnesses given by \eqref{xa_ya}, a simple calculation shows that the dynamics \eqref{dyn_Levene}
can be written as
\begin{equation}
	\De p = p(1-p)\sum_\a c_\a \frac{1-y_\a + p(x_\a+y_\a-2)}{x_\a p^2+2p(1-p)+y_\a(1-p)^2}\,.
\end{equation}
Expressing this with a common denominator, we see that the internal equilibria are the solutions of a polynomial in $p$
of degree $2\Ga-1$. Thus, in principle, there can be up to $2\Ga-1$ internal equilibria. For two demes,
Karlin (1977) provided a construction principle for obtaining three internal equilibria. It requires quite extreme
fitness differences and that in both demes the less fit allele is clsoe to dominant. 

By a clever procedure, Novak (2011) proved the following result:
\begin{theorem}\label{Novak}
The diallelic Levene model allows for any number $j\in\{1,2,\ldots,2\Ga-1\}$ of hyperbolic internal equilibria with any
feasible stability configuration.
\end{theorem}

His numerical results, which admit arbitrary dominance, show that the proportion of parameter space supporting more than $\Ga$ equilibria becomes extremely small if $\Ga>2$.

\section{Multiple alleles and arbitrary migration}\label{sec:mult_allele_arb_mig}
Because for multiple alleles and arbitrary migration few general results are available,
we focus on three limiting cases that are biologically important and amenable to mathematical analysis. These are weak selection and
weak migration, weak migration (relative to selection), and weak selection (relative to migration).
The first leads to the continuous-time migration model, the second to the so-called
weak-migration limit, and the third to the strong-migration limit. The latter two are based on a separation of time scales.
With the help of perturbation theory, results about existence and stability of equilibria, but also about global convergence, can be derived for an open subset of parameters of the full model. In addition, we report results about the case of no dominance and about uniform selection, i.e., when selection is the same in every deme.
We start with no dominance.

\subsection{No dominance}
The following generalizes Theorem \ref{no_dom_bound} for the Levene model. It holds for an arbitrary (constant) backward migration matrix.

\begin{theorem}[Nagylaki and Lou 2001, Theorem 2.4]\label{th:NL01_nodom}
Without dominance, the number of demes is a generic upper bound on the number of alleles that can be
maintained at any equilibrium.
\end{theorem}

This theorem also holds for hard selection. As shown by Peischl (2010), it can not be extended to DIDID; cf.\ Remark \ref{rem:Nodom_DIDID}.
An example showing that the upper bound can be achieved if
$I=\Ga$ is obtained by setting $v_{i,\a}=u_i\de_{i\a}$, where $u_i>0$. The internal equilibrium can be
written as $\hat p =  \tfrac12(I-\tfrac12 M)^{-1}M$, and convergence is geometric (at a rate $\le\tfrac12$). 

As we shall see below, with dominance the number of alleles that can be maintained at equilibrium may depend on the strength and pattern of migration.

\subsection{Migration and selection in continuous time}\label{sec:Mig_sel_cont}
Following Nagylaki and Lou (2007), we assume that both selection and migration are weak and approximate the discrete migration-selection dynamics \eqref{mig_sel_dyn} by a differential equation which is easier accessible. Accordingly, let
\begin{equation}\label{weak_sel_and_mig}
	w_{ij,\a} = 1 + \ep r_{ij,\a} \quad\text{and}\quad 
		\tilde m_{\a\be} = \de_{\a\be} + \ep \tilde\mu_{\a\be}\,,
\end{equation}
where $r_{ij,\a}$ and $\tilde\mu_{\a\be}$ are fixed for every $i,j\in\I$ and every $\a,\be\in\G$, and
$\ep>0$ is sufficiently small. From \eqref{marg_mean_fit} we deduce
\begin{subequations}\label{marg_fit_weak}
\begin{equation}\label{marg_mean_fit_weak}
	\wia = 1+ \ep r_{i,\a} \quad\text{and}\quad \wba = 1 + \ep \bar r_\a \,,
\end{equation}
where
\begin{equation}\label{marg_mean_fit_r}
	r_{i,\a} = \sum_j r_{ij,\a}\pja \quad\text{and}\quad \bar r_\a = \sum_{i,j}r_{ij,\a}\pia\pja \,.
\end{equation}
\end{subequations}

To approximate the backward migration matrix $M$, note that \eqref{hard} and \eqref{marg_mean_fit_weak}
imply that, for both soft and hard selection,
\begin{equation}\label{c_ast_ep}
	c_\a^\ast=c_\a+O(\ep)
\end{equation}
as $\ep\to0$. 
Substituting \eqref{weak_sel_and_mig} and \eqref{c_ast_ep} into \eqref{mfor_mback} leads to
\begin{equation}
	m_{\a\be} = \de_{\a\be} + \ep\mu_{\a\be} + O(\ep^2)
\end{equation}
as $\ep\to0$, where
\begin{equation}
	\mu_{\a\be} = \frac{1}{c_\a}\left(c_\be \tilde\mu_{\be\a} 
			- \de_{\a\be}\sum_\ga c_\ga\tilde\mu_{\ga\a}\right)\,.
\end{equation}
Because $\tilde M$ is stochastic, we obtain for every $\a\in\Ga$,
\begin{equation}
	\tilde\mu_{\a\be}\ge0 \text{ for every } \be\ne\a \quad\text{and}\quad \sum_\be \tilde\mu_{\a\be} = 0\,.
\end{equation}
As a simple consequence, $\mu_{\a\be}$ shares the same properties.

The final step in our derivation is to rescale time as in Sect.\ \ref{sect_cont_time_sel} by setting 
$t=\lfloor \ta/\ep \rfloor$ and $\pi_{i,\a}(\ta) = p_{i,\a}(t)$. Inserting all this into the difference equations
\eqref{mig_sel_dyn} and expanding yields
\begin{equation}
	\pi_{i,\a}(\ta+\ep) = \pi_{i,\a}\left\{1 + \ep[r_{i,\a}(\pi_{\cdot,\a}) - \bar r_\a (\pi_{\cdot,\a})]\right\}
		+ \ep \sum_\be \mu_{\a\be} \pi_{i,\be} + O(\ep^2) 
\end{equation}
as $\ep\to0$, where $\pi_{\cdot,\a}=(\pi_{1,\a},\ldots,\pi_{I,\a})^\Tr\in\S_I$.
Rearranging and letting $\ep\to0$, we arrive at
\begin{equation}
	\der{\pi_{i,\a}}{\ta} = \sum_\be \mu_{\a\be} \pi_{i,\be} 
			+ \pi_{i,\a} [r_{i,\a}(\pi_{\cdot,\a}) - \bar r_\a (\pi_{\cdot,\a})]\,.
\end{equation}
Absorbing $\ep$ into the migration rates and selection coefficients and returning to $p(t)$, we obtain the
slow-evolution approximation of \eqref{mig_sel_dyn},
\begin{equation}\label{weak_mig_sel}
	\dot p_{i,\a} = \sum_\be \mu_{\a\be} p_{i,\be} + p_{i,\a} [r_{i,\a}(p_{\cdot,\a}) - \bar r_\a (p_{\cdot,\a})]\,.
\end{equation}
In contrast to the discrete-time dynamics \eqref{mig_sel_dyn}, here the migration and selection terms are decoupled. This is a general
feature of many other slow-evolution limits (such as mutation and selection, or selection, recombination and migration).
Because of the decoupling of the selection and migration terms, the analysis of explicit models is often facilitated.

With multiple alleles, there are no general results on the dynamics of \eqref{weak_mig_sel}. For
two alleles, we set $p_\a=p_{1,\a}$ and write \eqref{weak_mig_sel} in the form
\begin{equation}\label{weak_mig_2all}
	\dot p_\a = \sum_\be \mu_{\a\be} p_\be + \ph_\a(p_\a)\,.
\end{equation}
Since $\mu_{\a\be}\ge0$ whenever $\a\ne\be$, the system \eqref{weak_mig_2all} is quasimonotone
or cooperative, i.e., $\partial\dot p_\a/\partial p_\be\ge0$ if $\a\ne\be$. As a consequence, \eqref{weak_mig_2all} 
cannot have an exponentially stable limit cycle. However, Akin (personal communication) has proved for three diallelic demes that a Hopf bifurcation
can produce unstable limit cycles. This precludes global convergence, though not generic convergence.
If $\Ga=2$, then every trajectory converges (Hirsch 1982; Hadeler and Glas 1983; see also Hofbauer and Sigmund 1998, p.\ 28), 
as is the case in the discrete-time model (Remark \ref{rem_converge_2A2D}).

\begin{example}\label{Example_Eyland}\rm
Eyland (1971) provided a global analysis of \eqref{weak_mig_sel} for the special case of two diallelic 
demes without dominance. As in Sect.\ \ref{Two diallelic demes}, 
we assume that the fitnesses of $\A_1\A_1$, $\A_1\A_2$, and
$\A_2\A_2$ in deme $\a$ are $1+s_\a$, 1, and $1-s_\a$, respectively, where $s_\a\ne0$ ($\a=1,2$).
Moreover, we set $\mu_1=\mu_{12}>0$, $\mu_2=\mu_{21}>0$, and write $p_\a$ for the frequency of $\A_1$ in deme $\a$.
Then \eqref{weak_mig_sel} becomes
\begin{subequations}
\begin{align}
	\dot p_1 &= \mu_1(p_2-p_1) + s_1p_1(1-p_1)\,, \\
	\dot p_2 &= \mu_2(p_1-p_2) + s_2p_2(1-p_2)\,.
\end{align}
\end{subequations}
The equilibria can be calculated explicitly. At equilibrium, $p_1=0$ if and only if $p_2=0$,
and $p_1=1$ if and only if $p_2=1$. In addition, there may be an internal equilibrium point. We set
\begin{equation}
	\si_\a=\frac{\mu_\a}{s_\a} \,,\quad \ka = \si_1+\si_2\,,
\end{equation}
and
\begin{equation}
	B = (1-4\si_1\si_2)^{1/2}\,.
\end{equation}
The internal equilibrium exists if and only if $s_1s_2<0$ and $\abs{\ka}<1$; cf.\ \eqref{PP_nodom}.
If $s_2<0<s_1$, it is given by
\begin{equation}\label{Eyland_internal}
	\hat p_1 = \tfrac12(1+B)-\si_1 \quad\text{and}\quad \hat p_2 = \tfrac12(1-B)-\si_2 \,.
\end{equation}

It is straightforward to determine the local stability properties of the three possible equilibria. Gobal
asymptotic stability follows from the results cited above about quasimonotone systems. Let
$p=(p_1,p_2)^\Tr$. Then allele $\A_1$ is eliminated in the region $\Om_0$ in Figure \ref{AK_Bio_NL_Fig2},
i.e., $p(t)\to(0,0)$ as $t\to\infty$, whereas $\A_1$ is ultimately fixed in the region $\Om_1$. In
$\Om_+$, $p(t)$ converges globally to the internal equilibrium point $\hat p$ given by \eqref{Eyland_internal}.
\end{example}

For the discrete-time dynamics \eqref{mig_sel_dyn} such a detailed analysis is not available. However, for
important special cases, results about existence, uniqueness, and stability of equilibria were derived by Karlin and Campbell (1980). 
Some of them are treated in Section \ref{sec:submult}.

Finally, we present sufficient conditions for global loss of an allele. In discrete time, such conditions are
available only for the Levene model. For \eqref{weak_mig_sel}, however, general conditions were derived
by Nagylaki and Lou (2007). Suppose that there exists $i\in\I$ and constants $\ga_{ij}$ such that
\begin{subequations}\label{loss}
\begin{equation}
	\ga_{ij}\ge0\,, \quad \ga_{ii}=0\,, \quad \sum_j \ga_{ij} =1\,,
\end{equation}
and
\begin{equation}\label{lossb}
	\sum_j \ga_{ij} r_{jk,\a} > r_{ik,\a}
\end{equation}	
\end{subequations}
for every $\a\in\G$ and every $k\in\I$. 

\begin{theorem}[Nagylaki and Lou 2007, Theorem 3.5 and Remark 3.7]\label{theorem_loss}
If the matrix $(\mu_{\a\be})$ is irreducible and the conditions \eqref{loss} are satisfied, then $p_i(t)\to0$ as $t\to\infty$
whenever $p_i(0)>0$ and $p_j(0)>0$ for every $j$ such that $\ga_{ij}>0$.
\end{theorem}

If there is no dominance, constants $s_{i,\a}$ exist such that $r_{ij,\a}=s_{i,\a}+s_{j,\a}$ for every
$i,j,\a$. Then condition \eqref{lossb} simplifies to
\begin{equation}\label{lossc}
	\sum_j \ga_{ij} s_{j,\a} > s_{i,\a}\,,
\end{equation}
and Theorem \ref{theorem_loss} applies. 

To highlight one of the biological implications, we follow Remark 3.12 in Nagylaki and Lou (2007) and
assume $\ga_{i1}>0$, $\ga_{iI}>0$, and $\ga_{ij}=0$ for
$j=2,\ldots,I-1$. Then \eqref{lossc} becomes
\begin{equation}\label{lossd}
	\ga_{i1} s_{1,\a} +  \ga_{iI} s_{I,\a} > s_{i,\a}
\end{equation}
for every $\a$, and the theorem shows that $p_i(t)\to0$ as $t\to\infty$ for $i=2,\ldots,I-1$. If,
in addition to \eqref{lossd}, we assume that
\begin{equation}
	\min(s_{1,\a},s_{I,\a}) < s_{i,\a} < \max(s_{1,\a},s_{I,\a})
\end{equation}
for $i=2,\ldots,I-1$ and every $\a$, then every allele $\A_i$ with $1<i<I$ is intermediate in every
deme, and $\A_1$ and $\A_I$ are extreme. Thus, Theorem \ref{theorem_loss} implies that all intermediate
alleles are eliminated. This conclusion can be interpreted as the elimination of generalists by
specialists, and it can yield the increasing phenotypic differentiation required for parapatric speciation
(cf.\ Lou and Nagylaki 2002 for an analogous result and discussion in the context of diffusion models).
If $s_{1,\a}-s_{I,\a}$ changes sign among demes, so that every allele is the fittest in some deme(s) and
the least fit in the other(s), then both alleles may be maintained. 

An other application of Theorem \eqref{theorem_loss} is the following. If in every deme the homozygotes have the
same fitness order and there is strict heterozygote intermediacy, then the allele with the greatest homozygous
fitness is ultimately fixed (Remark 3.20 in Nagylaki and Lou 2007).

\begin{remark}\rm
The slow-evolution approximation of the exact 
juvenile-migration model is also \eqref{weak_mig_sel} (Nagylaki and Lou 2008, Nagylaki 1992, pp.\ 143-144).
\end{remark}

\subsection{Weak migration}\label{sec:Weak_mig}
If migration is sufficiently weak relative to selection, properties of the dynamics can be inferred
by perturbation techniques from the well-understood case of a finite number of isolated demes in which
there is selection and random mating.

To study weak migration, we assume
\begin{equation}\label{weak_mig}
	m_{\a\be} = \de_{\a\be} + \ep \mu_{\a\be}\,,
\end{equation}
where $\mu_{\a\be}$ is fixed for every $\a,\be$, and $\ep>0$ is sufficiently small.  Because $M$ is stochastic, 
we obtain for every $\a\in\Ga$ (cf.\ Section \ref{sec:Mig_sel_cont})
\begin{equation}\label{ass_muab}
	\mu_{\a\be} \ge 0 \;\text{ for every } \be\ne\a \quad\text{and}\quad   \sum_\be \mu_{\a\be} = 0 \,.
\end{equation}
If there is no migration ($\ep=0$), the dynamics in each deme reduces to the pure selection dynamics
\begin{equation}\label{no_mig_dyn}
	\pia' = \pia \frac{\wia}{\wba}\,.
\end{equation}
Because \eqref{no_mig_dyn} is defined on the Cartesian product $\S_I^\Ga$, it may exhibit a 
richer equilibrium and stability structure than the panmictic selection dynamics \eqref{sel_dyn}. This was already illustrated by Example \ref{ex:underdom}, in which two asymptotically stable internal equilibria may coexist. By Theorem \ref{theorem:sel_dyn}, such an equilibrium configuration does not occur for \eqref{sel_dyn}.
 
The following is a central perturbation result that has a number of important consequences. Among others, it excludes complex dynamics in the full system \eqref{mig_sel_dyn} provided migration is sufficiently weak. 
We note that hyperbolicity is a generic property for \eqref{no_mig_dyn} (Appendix A in Nagylaki et al.\ 1999), and an internal equilibrium is hyperbolic if and only if it is isolated (Lemma 3.2 in Nagylaki and Lou 2006a).

\begin{theorem}[Nagylaki and Lou 2007, Theorem 4.1]\label{theorem_weak_mig}
Suppose that every equilibrium of \eqref{no_mig_dyn} is hyperbolic, that
\eqref{weak_mig} holds, and that $\ep>0$ is sufficiently  small.

(a) The set of equilibria $\Si_0\subset\S_I^\Ga$ of \eqref{no_mig_dyn}
contains only isolated points, as does the set of equilibria
$\Si_\ep\subset\S_I^\Ga$ of \eqref{mig_sel_dyn}. As $\ep\to0$, 
each equilibrium in $\Si_\ep$ converges to the corresponding equilibrium in $\Si_0$.

(b) In the neighborhood of each asymptotically stable equilibrium in
$\Si_0$, there exists exactly one equilibrium in $\Si_\ep$, and it 
is asymptotically stable. In the neighborhood of each unstable internal
equilibrium in $\Si_0$, there exists exactly one equilibrium
in $\Si_\ep$, and it is unstable. In the neighborhood of each unstable 
boundary equilibrium in $\Si_0$, there exists at most one equilibrium 
in $\Si_\ep$, and if it exists, it is unstable. 

(c) Every solution $p(t)$ of \eqref{mig_sel_dyn} converges to one of the equilibrium
points in $\Si_\ep$.
\end{theorem}

The perturbation results in (a) and (b) are essentially due to Karlin and McGregor (1972a,b). 
The proof of (c), which also yields (a) and (b), is a simplification of that of Theorem 2.3 in 
Nagylaki et al.\ (1999) and is based on quite deep results. To outline the proof, we need some 
preparation.

We consider a family of maps (difference equations) that depends on a parameter $\ep$,
\begin{equation}\label{difference}
     x' = f(x,\ep)\,,
\end{equation}
where $x\in\X\subseteq\Real^n$ ($\X$ compact and convex) and $\ep\in\E\subseteq\Real^k$ ($\E$ open).
We assume that $\hat x$ is a fixed point of \eqref{difference} if $\ep=\ep_0$ and that the Jacobian
$f'(\hat x,\ep_0)$ of $f$ evaluated at $(\hat x,\ep_0)$ exists and is continuous. Furthermore, we posit that
$\hat x$ is a hyperbolic equilibrium. If we define the function
\begin{equation}
     F(x,\ep) = f(x,\ep)-x\,,
\end{equation}
then $F(\hat x,\ep_0)=0$ and the Jacobian $F'(\hat x,\ep_0)$ is continuous and nonsingular (by
hyperbolicity of $\hat x$). Therefore, the
implicit function theorem shows that there exists an open neighborhood $\U$ of $\ep_0$ and a function
$\phi:\U\to\X$ such that $\phi(\ep_0) = \hat x$ and $F(\phi(\ep),\ep)=0$ for $\ep\in\U$, hence
\begin{equation}
	f(\phi(\ep),\ep) = \phi(\ep)\,.
\end{equation}
Hence, for every $\ep\in\U$, i.e., for $\ep$ close to $\ep_0$, \eqref{difference} has a uniquely 
determined fixed point $\hat x_\ep=\phi(\ep)$ close to $\hat x$. 

With the help of the Hartman-Grobman theorem, it can be shown that the stability properties of the perturbed 
fixed points $\hat x_\ep$ are the same as those of the unperturbed, $\hat x$. The reason is that if
an equilibrium is hyperbolic, this
property persists under small perturbations (the Jacobian changes continuously if parameters change
continuously); see also Theorem 4.4 in Karlin and McGegor (1972b). 
The extension of this argument to finitely many hyperbolic equilibria is evident. 

Although hyperbolic equilibria change continuously under small perturbations, limit sets of trajectories
do not: they can explode. Thus, perturbations could introduce `new' limit sets away from the hyperbolic
equilibria. What has good properties under perturbations is the set of \emph{chain-recurrent points}
introduced by Conley (1978).

Let $X$ be a compact set with metric $d$ and let $f:\X\to \X$ be a continuous map. A point $x\in \X$ is called
{\it chain recurrent} (with respect to $f$) if, for every $\de>0$, there exists a
finite sequence $x_0=x$, $x_1,\ldots,x_{r-1},x_r=x$ (often called a
$\de$-pseudo-orbit) such that $d(f(x_m),x_{m+1})<\de$ for $m=0,1,\ldots,r-1$.
The set of chain-recurrent points contains the limit sets of all orbits and  behaves well under perturbations (Akin 1993, p.\ 244).

\begin{proof}[Outline of the proof of Theorem \ref{theorem_weak_mig}]
(a) and (b) follow from the implicit function theorem and the Hartman-Grobman theorem. That asymptotically
stable boundary equilibria remain in the state space follows from Brouwer's fixed point theorem 
(for details, see Karlin and McGregor 1972, especially Theorem 4.4).

(c) Let
\begin{equation}
   \F = \{p\in\S_I^\Ga: p_{i,\a}(w_{i,\a}-\wba) = 0 \,
                  \; \forall i\in\I, \forall\a\in\Ga \}
\end{equation}
denote the set of equilibria of \eqref{no_mig_dyn}. Then, within in each deme, $\De\wba\ge0$ with
equality only at equilibrium; cf.\ \eqref{Wbar_increase}. Hence,
\begin{equation}
	\wb(p) = \sum_\a \wba(p_{\cdot,\a})
\end{equation}
satisfies $\De\wb(p)\ge0$  with $\De\wb(p)=0$ if and only if $p\in\F$. Because, on $\F$, $\wb$ takes only finitely
many values, Theorem 3.14 in Akin (1993) implies that the chain-recurrent points of \eqref{no_mig_dyn} are
exactly the equilibria.  

Now we can follow the proof of Theorem 2.3 in Nagylaki et al.\ (1999) almost verbally.
Because the set of chain-recurrent points consists only of hyperbolic equilibria, this is
also true for small $C^1$ perturbations of the dynamics (Akin 1993, p.\ 244).
Indeed, as an immediate consequence of the definition of chain recurrence,
it follows that the chain-recurrent set of \eqref{no_mig_dyn} changes in an upper
semicontinuous way with $\ep$. In particular, the chain-recurrent set for
$\ep>0$ is contained in the union of the $\de$-neighborhoods of the equilibria
for $\ep=0$, with $\de\to0$ for $\ep\to0$. By the implicit function theorem and the
openness of hyperbolicity (Hartman-Grobman theorem), if $\ep>0$ is
small, then the maximal invariant sets in those neighborhoods are
hyperbolic equilibria. Hence, for small $\ep$, the chain-recurrent set consists
only of finitely many equilibria, which implies convergence of all trajectories.
\end{proof}

\begin{remark}\rm
Boundary equilibria that are unstable in the absence of migration, can disappear under weak migration because
they may leave $\S_I^\Ga$. A simple example is overdominance in two diallelic demes: the unstable
zero-migration equilibria $(p_{1,1},p_{1,2})=(1,0)$ and $(p_{1,1},p_{1,2})=(0,1)$ do not survive perturbation.
Indeed, if $\ep>0$, the perturbation of $(1,0)$ must have the form $(1-\ep z_1, \ep z_2)$. If the fitnesses
of the genotypes are as in \eqref{xa_ya} and the migration rates are $m_{12}=\ep m_1$
and $m_{21}=\ep m_2$, then straightforward calculations show that this equilibrium
is given by
\begin{equation*}
     z_1 = -\frac{x_1 m_1}{1-x_1}  \quad\text{and}\quad  z_2 = -\frac{y_2 m_2}{1-y_2}\,,
\end{equation*}
which is not in $[0,1]\times[0,1]$ if there is overdominance, i.e., if $\xa<1$ and $\ya<1$.
\end{remark}

\begin{remark}\rm In the absence of migration, mean fitness is monotone increasing in each deme (except at the equilibria); see \eqref{Wbar_increase}. 
Therefore, by continuity, $\De\wb(p)=\sum_\a \De\wba(p_{\cdot,\a})>0$ for sufficiently small $\ep$ if
$p$ is bounded away from the set $\F$ of equilibria. If $p$ is close to $\F$, mean fitness may decrease. As an example assume two
diallelic demes with overdominance and (stable) equilibria $\hat p_{\cdot,1}$ and $\hat p_{\cdot,2}$
with $\hat p_{\cdot,1}\ne\hat p_{\cdot,2}$. Suppose that in some generation $p_{\cdot,\a}=\hat p_{\cdot,\a}$ 
holds for $\a=1,2$. Since $\wb(p)$ is maximized at $\hat p$, which, with migration, is an equilibrium only if 
$\hat p_{\cdot,1}=\hat p_{\cdot,2}$, we see that $\wb(p')<\wb(p)$.
\end{remark}

As another application of Theorem \ref{theorem_weak_mig}, we study the number of alleles that can be maintained 
at an asymptotically stable equilibrium under weak migration. First we prove a simple result for a single deme.
For more general results, see Sect.\ 2 in Nagylaki and Lou (2006a).

\begin{proposition}\label{int_dom_no_mig}
Let $\Ga=1$ and assume that the alleles are ordered such that $w_{ii}\ge w_{i+1,i+1}$ for $i=1,\ldots,I-1$.
In addition, assume that $w_{11} > w_{22}$ and there is intermediate dominance, i.e., 
\begin{equation}
	w_{ii} \ge w_{ij} \ge w_{jj} 
\end{equation}
for every $i$ and every $j>i$.
Then allele $\A_1$ is fixed as $t\to\infty$. The corresponding equilibrium is globally asymptotically stable.
\end{proposition}

\begin{proof}
The assumptions imply $w_{1i}\ge w_{ii}\ge w_{Ii}$ for every $i\in\I$ and $w_{11}>w_{I1}$ or $w_{I1}>w_{II}$. It follows that
$w_1= \sum_i w_{1i}p_i \ge \sum_i w_{Ii}p_i =w_I$ and, if $p_1\ne0$ and $p_I\ne0$,
\begin{equation}
	w_1>w_I\,.
\end{equation}
Therefore, we obtain
\begin{equation}
	\left(\frac{p_I}{p_1}\right)' = \frac{p_I'}{p_1'} = \frac{p_I w_I}{p_1 w_1} < \frac{p_I}{p_1}
\end{equation}
if $p_1>0$ and $p_I>0$. Hence, $p_I(t)\to0$ as $t\to\infty$ provided $p_I<1$ and $p_1>0$. Now we can
repeat this argument with $p_1$ and $p_{I-1}$. Proceeding inductively, we obtain $p_1(t)\to1$ as $t\to\infty$.
\end{proof}

\begin{theorem}\label{weak_mig_gen_max}
Suppose that migration is sufficiently weak and there is partial dominance in every deme \eqref{partial_dom}.

(a) Generically, there is global convergence to an asymptotically stable equilibrium at which at most $\Ga$
alleles are present. Thus, the number of demes is a generic upper bound for the number of alleles that 
can be maintained at a  stable equilibrium. 

(b) If $\Ga\le I$, then there is an open set of parameters such that $\Ga$ alleles are segregating at a globally
asymptotically stable equilibrium.
\end{theorem}

\begin{proof}
Corollary \ref{int_dom_no_mig} shows that, in the absence of migration, one allele is fixed in every deme.
The parameter combinations satisfying the assumptions of the theorem clearly form an open set of all 
possible parameter combinations. Therefore, without migration, at most $\Ga$ alleles can be maintained 
at an asymptotically stable equilibrium, and this can be achieved on an open set in the parameters space. 
Moreover, this equilibrium is globally asymptotically stable. 
Therefore, Theorem \ref{theorem_weak_mig} yields statement (a) for weak migration. Clearly, the same set
of alleles as without migration occurs at this equilibrium.

If $\Ga\le I$, we still obtain an open set of parameters if we choose fitnesses such that
in each deme a different allele has the highest homozygous fitness. 
Hence, the upper bound $\Ga$ can be achieved on an open set,
and Theorem \ref{theorem_weak_mig} yields statement (b).
\end{proof}

Therefore, in contrast to the Levene model (Theorem \ref{int_dom_fully}), in which migration is strong, for weak migration the number of alleles that
can be maintained at a stable equilibrium cannot exceed the number of demes.

\subsection{Strong migration}
If migration is much stronger than selection, we expect that rapid convergence to spatial quasi-homogeneity
occurs. After this short initial phase, evolution should be approximately panmictic with suitably averaged
allele frequencies. Because in the absence of selection, there exists a globally attracting manifold of
equilibria, so that the dynamics is not gradient like,
the derivation of perturbation results is much more delicate than for weak migration. Since the fundamental ideas in the proofs of the most relevant results
are essentially the same as if selection acts on many loci, we defer the analysis of strong migration to Section \ref{sec:mult_strong_mig}, where the multilocus case is treated.

\subsection{Uniform selection}
Selection is called uniform if
\begin{equation}
	w_{ij,\a} = w_{ij} \quad\text{for every $i$, $j$, and every $\a$.} 
\end{equation}
Under spatially uniform selection, one might
expect that population structure leaves no genetic traces. However, this is not always true as shown
by Example \ref{ex:underdom} with underdominance and weak migration. Indeed, if we have
two diallelic demes with the same underdominant selection in both, then under weak migration there are
nine equilibria, four of which are asymptotically stable. These are the two monomorphic equilibria and
the two equilibria where each of the alleles is close to fixation in one deme and rare in the other. Only three
of the equilibria are uniform, i.e., have the same allele frequencies in both demes. These are the two
monomorphic equilibria and the `central' equilibrium.

In the following we state sufficient conditions under which there is no genetic indication of population
structure. We call $\hat p\in\S_I^\Ga$ a uniform selection equilibrium if every $\hat p_{\cdot,\a}$
is an equilibrium of the pure selection dynamics \eqref{no_mig_dyn} and $\hat p_{\cdot,\a}=\hat p_{\cdot,\be}$
for every $\a$, $\be$.

\begin{theorem}[Nagylaki and Lou 2007, Theorem 5.1]
If $\hat p\in\S_I^\Ga$ is a uniform selection equilibrium, then $\hat p$ is an equilibrium of the
(full) migration-selection dynamics \eqref{mig_sel_dyn}, and $\hat p$ is either asymptotically stable
for both \eqref{no_mig_dyn} and \eqref{mig_sel_dyn}, or unstable for both systems.
\end{theorem}

It can also be shown that the ultimate rate of convergence to equilibrium is determined entirely by 
selection and is independent of migration.

Next, one may ask for the conditions when the solutions of \eqref{mig_sel_dyn} converge globally to a
uniform selection equilibrium. One may expect global convergence under migration if the uniform
selection equilibrium is globally asymptotically stable without migration. So far, only weaker results
could be proved.

\begin{theorem}[Nagylaki and Lou 2007, Sections 5.2 and 5.3]
Suppose there is a uniform selection equilibrium that is globally asymptotically stable in the 
absence of migration. Each of the following conditions implies global convergence of solutions to $\hat p$
under migration and selection.

(a) Migration is weak and, without migration, every equilibrium is hyperbolic.

(b) Migration is strong, $M$ is ergodic, and every equilibrium of the strong-migration limit is hyperbolic.

(c) The continuous-time model \eqref{weak_mig_sel} applies, $\mu_{\a\be}>0$ for every pair $\a,\be$ with
$\a\ne\be$, and $\hat p$ is internal in the absence of migration (hence, also with migration).

(d) The continuous-time model \eqref{weak_mig_sel} applies, $\mu_{\a\be}=\mu_{\be\a}$ for every pair $\a,\be$, and 
$\hat p$ is internal.
\end{theorem}

\section{Multilocus models}
Since many phenotypic traits are determined by many genes, many loci are subject to selection. If loci are on the same chromosome, especially if they are within a short physical distance, they can not be treated independently. In order to understand the evolutionary effects of selection on multiple loci, models have to be developed and studied that take linkage and recombination into account. Because of the complexity of the general model, useful analytical results can be obtained essentially only for limiting cases or under special assumptions. Whereas the former can be sometimes extended using perturbation theory to obtain general insight, the latter are mainly useful to study specific biological questions or to demonstrate the kind of complexity that can arise. 

Before developing the general model, we introduce the basic model describing the interaction of selection and recombination and point out some of its fundamental properties. We begin by illustrating the effects of recombination in the simplest meaningful setting.

\subsection{Recombination between two loci}
We consider two loci, $\A$ and $\B$, each with two alleles, $\A_1$, $\A_2$, and $\B_1$, $\B_2$. Therefore, there are four possible gametes, $\A_1\B_1$, $\A_1\B_2$, $\A_2\B_1$, $\A_2\B_2$, and 16 diploid genotypes. If, as usual, $\A_i\B_j/\A_k\B_\ell$ and $\A_k\B_\ell/\A_i\B_j$ are indistinguishable, then only 10 different unordered genotypes remain. 

Genes on different chromosomes are separated during meiosis with probability one half (Mendel's Principle of Independent Assortment). If the loci are on the same chromosome, they may become separated by a recombination event (a crossover) between them. We denote this \emph{recombination probability} by $r$.
The value of $r$ usually depends on the distance between the two loci along the chromosome. 
Loci with $r=0$ are called completely linked (and may be treated as a single locus), and loci with
$r=\tfrac12$ are called unlinked. The maximum value of $r=\tfrac12$ occurs for loci on different chromosomes, because then all
four gametes are produced with equal frequency $\tfrac14$. Thus, the recombination rate satisfies $0\le r\le\tfrac12$.

If, for instance, in the initial generation only the genotypes $\A_1\B_1/\A_1\B_1$
and $\A_2\B_2/\A_2\B_2$ are present, then in the next generation only these double homozygotes, as well as the two double heterozygotes
$\A_1\B_1/\A_2\B_2$ and $\A_1\B_2/\A_2\B_1$ will be present. 
After further generations of random mating, all other genotypes will occur, but not immediately at their
equilibrium frequencies. The formation of gametic types other than $\A_1\B_1$ or $\A_2\B_2$ requires that recombination occurs between the two loci.

We denote the frequency of gamete $\A_i\B_j$ by $P_{ij}$ and, at first, admit an arbitrary number of alleles at each locus.
Let the frequencies of the alleles $\A_i$ at the first locus be denoted by $p_i$ and those of the alleles
$\B_j$ at the second locus by $q_j$. Then
\begin{equation}
	p_i =\sum_j P_{ij} \;\text{and}\; q_j=\sum_i P_{ij}\,.
\end{equation}
The allele frequencies are no longer sufficient to describe the genetic composition of the population because, in general,
they do not evolve independently. 
\emph{Linkage equilibrium} (LE) is defined as the state in which
\begin{equation}
	P_{ij} = p_i q_j                                            \label{LEij}
\end{equation}
holds for every $i$ and $j$. Otherwise the population is said to be in 
\emph{linkage disequilibrium} (LD). LD is equivalent to probabilistic dependence of allele frequencies between loci.

Given $P_{ij}$, we want to find the gametic frequencies $P_{ij}'$ in the
next generation after random mating. The derivation of the recursion equation is based on the
following basic fact of Mendelian genetics: an individual with genotype 
$\A_i\B_j/\A_k\B_l$ produces gametes of parental type if no recombination occurs
(with probability $1-r$), and recombinant gametes if recombination 
between occurs (with probability $r$).
Therefore, the fraction of gametes $\A_i\B_j$ and $\A_k\B_l$ is $\tfrac12(1-r)$ each,
and that of $\A_i\B_l$  and $\A_k\B_j$ is $\tfrac12r$ each. 
From these considerations, we see that the frequency of gametes of type
$\A_i\B_j$ in generation $t+1$ produced without recombination is
$(1-r)P_{ij}$, and that produced with recombination is $rp_iq_j$ because
of random mating. Thus,
\begin{equation}
	P_{ij}' = (1-r)P_{ij} + rp_iq_j\,.                              \label{P_{ij}'}
\end{equation}

This shows that the gene frequencies are conserved, but the gamete
frequencies are not, unless the population is in LE. Commonly, LD between alleles $\A_i$ and $B_j$ is measured by the parameter
\begin{equation}
	D_{ij} = P_{ij} - p_iq_j\,.                                   \label{D_{ij}}
\end{equation}
The $D_{ij}$ are called linkage disequilibria.
From \eqref{P_{ij}'} and \eqref{D_{ij}} we infer
\begin{equation}
	D_{ij}' = (1-r)D_{ij}                                        \label{6.4}
\end{equation}
and
\begin{equation}
	D_{ij}(t) = (1-r)^t D_{ij}(0)\,.                               \label{D_{ij}(t)}
\end{equation}
Therefore, unless $r=0$, linkage disequilibria decay at the geometric rate $1-r$
and LE is approached gradually without oscillation.

For two alleles at each locus, it is more convenient to label the frequencies of the 
gametes $\A_1\B_1$, $\A_1\B_2$, $\A_2\B_1$, and
$\A_2\B_2$ by $x_1$, $x_2$, $x_3$, and $x_4$, respectively. A simple calculation reveals that 
\begin{equation}
	D = x_1x_4 - x_2x_3                                            \label{D}
\end{equation}
satisfies
\begin{equation}
	D = D_{11} = -D_{12} = -D_{21} = D_{22}\,.                  
\end{equation}
Thus, the recurrence equations \eqref{P_{ij}'} for the gamete frequencies can be rewritten as
\begin{equation}\label{gam_dyn_2L2A}
	x_i' = x_i - \et_i rD\,,   \quad i\in\{1,2,3,4\}\,,
\end{equation}  
where $\et_1=\et_4=-\et_2=-\et_3$.                                              

The two-locus gametic frequencies are the elements of the simplex $\S_4$ and may be represented geometrically by the
points in a tetrahedron. The subset where $D=0$ forms a two-dimensional manifold and is called the linkage equilibrium, or
Wright, manifold. It is displayed in Figure~\ref{fig_Simplex}. 

\begin{figure}
\begin{center}
	\includegraphics[height=7.0cm]{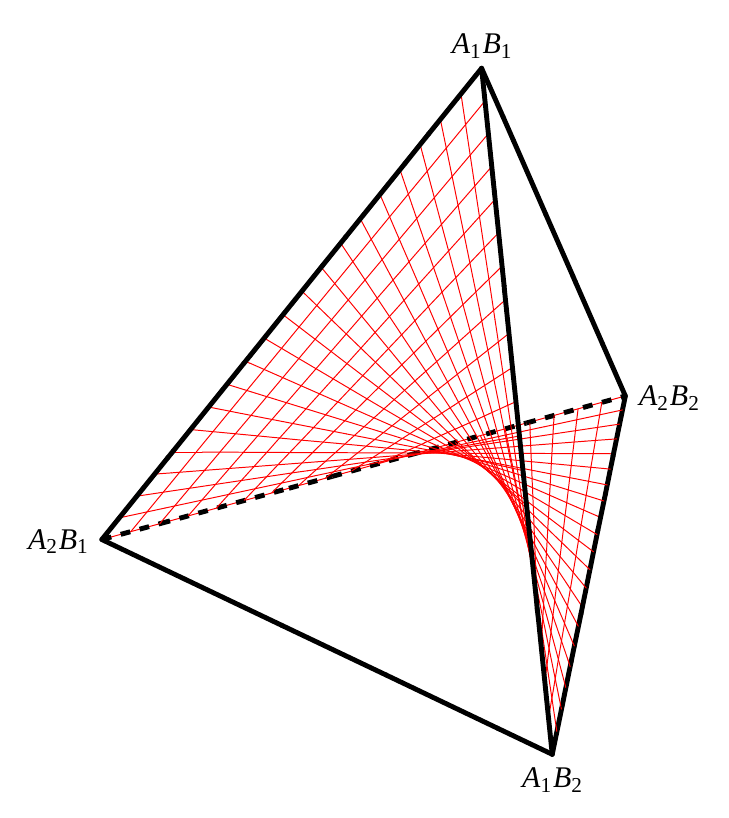}
	\caption{The tetrahedron represents the state space $\S_4$ of the
two-locus two-allele model. The vertices correspond to fixation of the
labeled gamete, and frequencies are measured by the (orthogonal) distance
from the opposite boundary face. At the center of the simplex all gametes
have frequency $\tfrac14$. The two-dimensional (red) surface is the LE
manifold, $D=0$. The states of maximum LD, $D=\pm\tfrac14$, are the centers
of the edges connecting $\A_1\B_2$ to $\A_2\B_1$ and $\A_1\B_1$ to $\A_2\B_2$.}\label{fig_Simplex}
\end{center}
\end{figure}

If $r>0$, \eqref{D_{ij}(t)} implies that all
solutions of \eqref{gam_dyn_2L2A} converge to the LE manifold along
straight lines, because the allele frequencies $x_1+x_2$ and $x_1+x_3$ 
remain constant, where sets such as $x_1+x_2=\text{const.}$ represent 
planes in this geometric picture. The LE manifold is invariant under the dynamics \eqref{gam_dyn_2L2A}.

\subsection{Two diallelic loci under selection}\label{sec:2L2A}
To introduce selection, we assume that viability selection acts on juveniles. Then recombination and random mating occurs.
Since selection acts on diploid individuals, we assign fitnesses to two-locus genotypes. We denote the fitness
of genotype $ij$ ($i,j\in\{1,2,3,4\}$) by $w_{ij}$, where we assume $w_{ij}=w_{ji}$, because usually it does not matter which gamete is paternally or maternally inherited. The marginal fitness of gamete $i$ is defined by $w_i=\sum_{i=1}^4 w_{ij}x_j$, and
the mean fitness of the population is $\wb= \sum_{i,j=1}^4 w_{ij}x_ix_j$. If we assume, as is frequently the case, that there is no
position effect, i.e., $w_{14}=w_{23}$, simple considerations yield the selection-recombination dynamics (Lewontin and Kojima 1960):
\begin{equation}\label{recosel_2L2A}
	x_i' = x_i\frac{w_i}{\wb} - \et_i \frac{w_{14}}{\wb} rD\,,   \quad i\in\{1,2,3,4\}\,.
\end{equation}
This is a much more complicated dynamical system than either the pure selection dynamics \eqref{sel_dyn} or the pure recombination
dynamics \eqref{gam_dyn_2L2A}, and has been studied extensively (for a review, see B\"urger 2000, Sects.\ II.2 and VI.2).
In general, mean fitness may decrease and is no longer maximized at an equilibrium. 

In addition, the existence of stable limit cycles has been established for this discrete-time model (Hastings 1981b, Hofbauer and Iooss 1984)
as well as for the corresponding continuous-time model (Akin 1979, 1982). 
Essentially, the demonstration of limit cycles requires that selection coefficients and recombination rates are of similar magnitude.

There is a particularly important special case in which the dynamics is simple. This is the case of \emph{no epistasis}, or \emph{additive fitnesses}. 
Then there are constants $u_i^{(n)}$ such that
\begin{equation}\label{additive_fitness}
	w_{ij} = u_i^{(1)}+u_j^{(2)} \quad\text{for every } i,j\in\{1,2,3,4\}\,.
\end{equation}
In the absence of epistasis, i.e., if \eqref{additive_fitness} holds, mean fitness $\wb$ is a (strict) Lyapunov function (Ewens 1969). In addition, a point $p$ is an equilibrium point of \eqref{recosel_2L2A} if and only if it is both a selection equilibrium for each locus and it
is in LE (Lyubich 1992, Nagylaki et al.\ 1999). In particular, the equilibria are the critical points of mean fitness.

The reason for this increased complexity of two-locus (or multilocus) systems lies not so much in the increased dimensionality but arises mainly
from the fact that epistatic selection generates nonrandom
associations (LD) among the alleles at different loci. Recombination breaks up
these associations to a certain extent but changes gamete frequencies in a complex way. Thus, there are different kinds
of interacting nonlinearities arising in the dynamical equations under selection and recombination.

\subsection{The general model}\label{sec:mult_gen_model}
We extend the migration-selection model of Section \ref{sec:gen_mig-sel} by assuming that selection acts on a finite number of recombining loci.
The treatment follows B\"urger (2009a), which was inspired by Nagylaki (2009b).
We consider a diploid population with discrete, nonoverlapping
generations, in which the two sexes need not be distinguished. The population
is subdivided into $\Ga\ge1$ panmictic colonies (demes) that exchange adult
migrants independently of genotype. In each of the demes, selection acts 
through differential viabilities, which are time and frequency independent. 
Mutation and random genetic drift are ignored. 

The genetic system consists of $L\ge1$ loci and $I_n\ge2$ alleles, $\A_{i_n}^{(n)}$ ($i_n=1,\ldots,I_n$),
at locus $n$. We use the multi-index $i=(i_1,\ldots,i_L)$ as an
abbreviation for the gamete $\A_{i_1}^{(1)}\ldots \A_{i_L}^{(L)}$.
We designate the set of all loci by $\Loc = \{1,\ldots,L\}$, 
the set of all alleles at locus $n$ by $\I_n = \{1,\ldots,I_n\}$, 
and the set of all gametes by $\I$. The number of gametes is
$I=\abs{\I}=\prod_n I_n$, the total number of genes (alleles at
all loci) is $I_1+\cdots+I_L$. We use letters $i,j,\ell\in\I$ for
gametes, $k,n\in\Loc$ for loci, and $\a,\be\in\G$ for demes. Sums or products without ranges 
indicate summation over all admissible indices, e.g., $\sum_n=\sum_{n\in\Loc}$. 

Let $p_{i,\a}=p_{i,\a}(t)$ represent the frequency of gamete $i$ among zygotes in
deme $\a$ in generation $t$. We define the following column vectors:
\begin{subequations}\label{vector_notation}
\begin{align}
    p_i     &= (p_{i,1},\ldots,p_{i,\Ga})^\Tr \in\Real^\Ga\,, \label{p_i}\\
    p_\dota&= (p_{1,\a},\ldots,p_{I,\a})^\Tr \in \S_I\,, \label{p_{(a)}} \\
    p       &= \left(p_{\cdot,1}^\Tr,\ldots,p_{\cdot,\Ga}^\Tr\right)^\Tr\in \S_I^\Ga\,.
\end{align}
\end{subequations}
Here, $p_i$, $p_\dota$, and $p$ signify the frequency of gamete $i$ in each deme, the gamete
frequencies in deme $\a$, and all gamete frequencies, respectively. 
We will use analogous notation for other quantities, e.g.\ for $D_{i,\a}$.

The frequency of allele $\A_{i_k}^{(k)}$ among gametes in deme $\a$ is
\begin{equation}\label{pikak}
   \pikak = \sum_{i|i_k} p_{i,\a}\,, 
\end{equation}
where the sum runs over all multi-indices $i$ with the $k$th component
fixed as $i_k$. We write 
\begin{equation}\label{pikk}
   \pikk = \left(p_{i_k,1}^{(k)},\ldots,p_{i_k,\Ga}^{(k)}\right)^\Tr\in\Real^\Ga
\end{equation}
for the vector of frequencies of allele $\A_{i_k}^{(k)}$ in each deme.

Let $x_{ij,\a}$ and $w_{ij,\a}$ denote the frequency and fitness of 
genotype $ij$ in deme $\a$, respectively. We designate the marginal fitness 
of gamete $i$ in deme $\a$ and the mean fitness of the population in deme 
$\a$ by
\begin{subequations}\label{marg_mean}
\begin{align}
   w_{i,\a} &= w_{i,\a}(p_\dota) = \sum_j w_{ij,\a} p_{j,\a} \\
\intertext{and}
  \wbara &= \wbara(p_\dota) = \sum_{i,j} w_{ij,\a} p_{i,\a} p_{j,\a}\,.
\end{align}
\end{subequations}

The life cycle starts with zygotes in Hardy-Weinberg proportions.
Selection acts in each deme on the newly born offspring. Then recombination occurs followed by 
adult migration and random mating within in each deme.
This life cycle extends that in Section \ref{sec:forw_backw}.
To deduce the general multilocus migration-selection dynamics (Nagylaki 2009b), let
\begin{subequations}\label{dyn}
\begin{equation}\label{dyn_sel}
   x_{ij,\a}^\ast = p_{i,\a} p_{j,\a} w_{ij,\a}/\wbara
\end{equation}
be the frequency of genotype $ij$ in deme $\a$ after selection, and
\begin{equation}\label{dyn_rec}
   p_{i,\a}^\# = \sum_{j,\ell} R_{i,j\ell} x_{j\ell,\a}^\ast
\end{equation}
its frequency after recombination. Here, $R_{i,j\ell}$ is the probability that during gametogenesis,
paternal haplotypes $j$ and $\ell$ produce a gamete $i$ by recombination. 
Finally, migration occurs and yields the gamete frequencies in the next generation in each deme:
\begin{equation}\label{dyn_mig2}
   p_{i,\a}' = \sum_\be m_{\a\be} p_{i,\be}^\# \,.
\end{equation}
\end{subequations}
The recurrence equations \eqref{dyn} describe the evolution of gamete frequencies under selection on multiple recombining loci and migration.
We view \eqref{dyn} as a dynamical system on $\S_I^\Ga$. We leave it to the reader to check the obvious fact that the processes of 
migration and recombination commute. 

The complications introduced by recombination are disguised by the terms $R_{i,jl}$ which depend on the
recombination frequencies among all subsets of loci. To obtain an analytically useful representation of \eqref{dyn_rec}, 
more effort is required.

Let $\{\K,\N\}$ be a nontrivial decomposition of $\Loc$, i.e.,
$\K$ and its complement $\N=\Loc\setminus \K$ are each proper subsets of $\Loc$
and, therefore, contain at least one locus. (The decompositions
$\{\K,\N\}$ and $\{\N,\K\}$ are identified.) We
designate by $c_\K$ the probability of reassociation of the genes at the
loci in $\K$, inherited from one parent, with the genes at the loci in $\N$,
inherited from the other. Let
\begin{equation}\label{ctot}
   \ctot = \sum_\K c_\K\,, 
\end{equation}
where $\sum_\K$ runs over all (different) decompositions $\{\K,\N\}$ of $\Loc$,
denote the total recombination frequency.
We designate the recombination frequency
between loci $k$ and $n$, such that $k<n$, by $c_{kn}$. It is given by
\begin{equation}
    c_{kn} = \sum_{\K\in \Loc_{kn}} c_\K\,,
\end{equation}
where $\Loc_{kn} = \{\K: k\in \K \text{ and } n\in \N\}$ (B\"urger 2000, p.\ 55).
Unless stated otherwise, we assume that all pairwise recombination rates 
$c_{kn}$ are positive. Hence,
\begin{equation}\label{cmin}
   \cmin = \min_{k<n} c_{kn} > 0\,,
\end{equation}

We define
\begin{equation}\label{D_ia}
   D_{i,\a} = \frac{1}{\bar w_\a} \sum_j \sum_\K c_\K 
         \left(w_{ij,\a}\, p_{i,\a}p_{j,\a} 
          - w_{i_\K j_{\N},j_\K i_{\N},\a}\, 
            p_{i_\K j_{\N},\a}\,p_{j_\K i_{\N},\a}\right)\,.
\end{equation}
This is a measure of LD in gamete $i$ in deme $\a$.
A considerable generalization of the arguments leading to \eqref{P_{ij}'} shows that \eqref{dyn_rec} can be expressed in the following form (Nagylaki 1993):
\begin{equation}\label{dyn_rec2}
   p_{i,\a}^\# = p_{i,\a}\frac{w_{i,\a}}{\bar w_\a} - D_{i,\a}\,.   \tag{\ref{dyn_rec}'}
\end{equation}
If there is only one deme, then \eqref{dyn_rec2} provides the recurrence equation describing evolution under selection on
multiple recombining loci. It seems worth noting that the full dynamics, \eqref{dyn}, 
does not depend on linkage disequilibria between demes. 

Let
\begin{equation}\label{LE_a}
   \La_{0,\a} = \left\{p_\dota: p_{i,\a} 
       = p_{i_1,\a}^{(1)}\cdot\ldots\cdot p_{i_L,\a}^{(L)}\right\} 
       \subseteq \S_I 
\end{equation}
denote the {\it linkage-equilibrium manifold} in deme $\a$, and let
\begin{equation}\label{LE}
   \La_0 = \La_{0,1}\times\ldots\times\La_{0,\Ga} \subseteq \S_I^\Ga\,.
\end{equation}
If there is no position effect, i.e., if $w_{ij,\a} = w_{i_\K j_{\N},j_\K i_{\N};\a}$
for every $i$, $j$, and $\K$, then $D_{i,\a}=0$ for every $p_\dota\in\La_{0,\a}$. Hence,
\begin{equation}\label{La_0_D}
    \La_{0,\a} \subseteq \{p_\dota: D_\dota=0\}\,,
\end{equation}
where $D_\dota$ is defined in analogy to \eqref{p_{(a)}}.
In the absence of selection, equality holds in \eqref{La_0_D}.

We will often need the following assumption:
\begin{equation}\label{ergodic}
    \text{The backward migration matrix $M$ is ergodic, i.e., irreducible 
       and aperiodic.}    \tag{E}
\end{equation}
Given irreducibility, the biologically trivial condition that individuals
have positive probability of remaining in some deme, i.e., $m_{\a\a}>0$
for some $\a$, suffices for aperiodicity (Feller 1968, p.\ 426). Because $M$ is a finite
matrix, ergodicity is equivalent to primitivity.

If \eqref{ergodic} holds, there exists a principal left eigenvector $\xi\in\text{int} \S_\Ga$ such that
\begin{equation}\label{xi}
   \xi^\Tr M = \xi^\Tr\,. 
\end{equation}
The corresponding principal eigenvalue 1 of $M$ is simple and exceeds every other eigenvalue in modulus. 
The principal eigenvector $\xi$ is the unique stationary distribution of the Markov chain
with transition matrix $M$, and $M^n$ converges at a geometric rate to $e\xi^\Tr$ as $n\to\infty$, where
$e=(1,\ldots,1)^\Tr\in\Real^\Ga$ (Feller 1968, p.\ 393; Seneta 1981, p.\ 9).

We average $p_{i,\a}$ with respect to $\xi$,
\begin{equation}\label{P_i}
   P_i = \xi^\Tr p_i\,,\quad P = (P_1,\ldots,P_I)^\Tr\in\S_I\,,
\end{equation}
and define the gamete-frequency deviations $q$ from the average
gamete frequency $P$:
\begin{subequations}\label{q}
\begin{align}
   q_{i,\a}  &= p_{i,\a} - P_i\,, \label{q_ia}\\
   q_i  &= p_i - P_i e \in \Real^\Ga\,, \label{q_i}\\
   \qikk  &= \sum\nolimits_{i|i_k} q_i \in \Real^\Ga\,, \label{qikk}\\
   q_\dota  &= p_\dota - P \in\Real^I\,, \label{q_(a)}\\
   q  &= (q_{\cdot,1}^\Tr,\ldots,q_{\cdot,\Ga}^\Tr)^\Tr \in\Real^{I\Ga}\,.
\end{align}
\end{subequations}
Therefore, $q$ measures spatial heterogeneity or diversity. If $q=0$,
the gametic distribution is spatially homogeneous.

\subsection{Migration and recombination}\label{Mig_Rec}
We study migration and recombination in the absence of selection,  i.e., if $w_{ij,\a}=1$ for every $i$, $j$, $\a$.
Then the dynamics \eqref{dyn} reduces to 
\begin{equation}\label{mig_rec}
   p_{i,\a}' = \sum_\be m_{\a\be} (p_{i,\be} - D_{i,\be})\,,
\end{equation}
where 
\begin{equation}\label{D_nosel}
   D_{i,\be} = \sum_\K c_\K \left(p_{i,\be} 
            - p_{i_\K,\be}^{(\K)}p_{i_{\N},\be}^{(\N)}\right)
\end{equation}
and $\{\K,\N\}$ as above \eqref{ctot}. Here, 
\begin{equation}\label{piKaK}
   p_{i_\K,\be}^{(\K)} = \sum_{i|i_\K} p_{i,\be}\,,
\end{equation}
where $\sum_{i|i_\K}$ runs over all multi-indices $i$ with the components in $\K$ fixed as $i_\K$,
denotes the marginal frequency of the gamete with components $i_k$ fixed for the loci $k\in \K$.
In vector form, i.e., with the notation \eqref{p_i}, \eqref{mig_rec} becomes
\begin{equation}\label{mig_rec_vec}
  p_i' = M(p_i - D_i)\,, \quad i\in\I\,.
\end{equation}

The following theorem shows that in the absence of selection trajectories quickly approach LE and spatial homogeneity.
\begin{theorem}[B\"urger 2009a]\label{mig_rec_theorem}
Suppose that \eqref{mig_rec_vec} and \eqref{ergodic} hold. Then, the manifold
\begin{subequations}\label{Ps_0}
\begin{align}
   \Ps_0 &= \bigl\{p\in\S_I^\Ga: p_i =  \prod_k\pikk \text{ and } \qikk = 0
     			   \text{ for every $k$ and $i_k$}\bigr\} \\
         &= \bigl\{p\in\S_I^\Ga: D=0 \text{ and } q =0\bigr\} 
\end{align}
\end{subequations}
is invariant under \eqref{mig_rec_vec} and globally attracting at a 
uniform geometric rate. Furthermore, every point on $\Ps_0$ is an equilibrium
point. Thus, (global) LE and spatial homogeneity are approached  quickly under recombination 
and (ergodic) migration.
\end{theorem}

This theorem generalizes the well known fact that in multilocus systems in which recombination is the only evolutionary force,
linkage disequilibria decay to zero at a geometric rate (Geiringer 1944, Lyubich 1992). The proof is very technical and uses induction on the number of embedded loci of the full $L$-locus system. Theorem \ref{mig_rec_theorem} does not hold if the migration matrix is reducible or periodic (Remark 3.3 in B\"urger 2009a).

In the following, we investigate several limiting cases in which useful general results can be proved and perturbation theory allows important extensions.

\subsection{Weak selection}\label{sec_weak_sel}
We assume that selection is weaker than recombination and migration. The main result will be that
all trajectories converge to an invariant manifold $\Ps_\ep$ close to $\Ps_0$ \eqref{Ps_0}, 
on which there is LE and allele frequencies are deme independent. On $\Ps_\ep$, the dynamics can be described 
by a small perturbation of a gradient system on $\Ps_0$. This implies that all trajectories converge, i.e., no cycling can occur, and the
equilibrium structure can be inferred from that of the much simpler gradient system.

Throughout this section, we assume \eqref{ergodic}, i.e., the backward migration matrix is ergodic.
We assume that there are constants $r_{ij,\a}\in\Real$ such that
\begin{equation}\label{weak_sel}
    w_{ij,\a} = 1 + \ep r_{ij,\a}\,,
\end{equation}
where $\ep\ge0$ is sufficiently small and $\abs{r_{ij}}\le1$. Migration and recombination rates, 
$m_{\a\be}$ and $c_\K$, are fixed so that fitness differences are small compared with them.
From \eqref{marg_mean} and \eqref{weak_sel}, we deduce
\begin{equation}\label{r_ia}
   w_{i,\a}(p_\dota) = 1 + \ep r_{i,\a}(p_\dota)\,,\quad
   \wbara(p_\dota) = 1 + \ep \rbara(p_\dota)\,,
\end{equation}
in which
\begin{equation}\label{r_ia_rbara}
   r_{i,\a}(p_\dota) = \sum_j r_{ij,\a}p_{j,\a}\,,\quad
   \rbara(p_\dota) = \sum_{i,j} r_{ij,\a}p_{i,\a}p_{j,\a}\,.
\end{equation}

When selection is dominated by migration and recombination, we expect that
linkage disequilibria within demes as well as gamete- and gene-frequency differences between demes
decay rapidly to small quantities. In particular, we expect 
approximately panmictic evolution of suitably averaged gamete frequencies
in `quasi-linkage equilibrium'. We also show that all trajectories converge
to an equilibrium point, i.e., no complicated dynamics, such as cycling, 
can occur. In the absence of migration, this
was proved by Nagylaki et al.\ (1999, Theorem 3.1). For a single locus under selection and
strong migration, this is the content of Theorem 4.5 in Nagylaki and Lou (2007). 
Theorem \ref{weak_sel_th} and its proof combine and extend these results
as well as some of the underlying ideas and methods.

To formulate and prove this theorem, we define the vector
\begin{equation}
   \rh_\a = \left(p_{1,\a}^{(1)},\ldots,p_{I_1,\a}^{(1)},
               \ldots,p_{1,\a}^{(L)},\ldots,p_{I_L,\a}^{(L)}\right)^T 
                              \in\S_{I_1}\times\cdots\times\S_{I_L}
\end{equation}
of all allele frequencies at every locus in deme $\a$, and the vector
\begin{equation}\label{pi}
   \pi = \left(P_1^{(1)},\ldots,P_{I_1}^{(1)},
               \ldots,P_1^{(L)},\ldots,P_{I_L}^{(L)}\right)^T 
                              \in\S_{I_1}\times\cdots\times\S_{I_L}
\end{equation}
of all averaged allele frequencies at every locus. We
note that in the presence of selection the $\Pikk$, hence $\pi$, 
are time dependent. Instead of $p$, we will use $\pi$, $D$, and $q$ to 
analyze \eqref{dyn}, and occasionally write $p=(\pi,D,q)$.

On the LE manifold $\La_{0,\a}$ \eqref{LE_a}, which is characterized by the $\rh_\a$ ($\a\in\G$),
the selection coefficients of gamete $i$, allele $i_n$ at locus $n$, and of the entire population are
\begin{subequations}\label{r(rh)}
\begin{align}
   r_{i,\a}(\rh_\a) &= \sum_j r_{ij,\a}\prod_k \pjkak\,, \\
   \rinan(\rh_\a) &=\sum_{i|i_n} r_{i,\a}(\rh_\a)
               \prod_{k:k\ne n} \pikak\,,    \label{rinan}\\
   \rbara(\rh_\a) &=\sum_i r_{i,\a}(\rh_\a)\prod_k\pikk\,, \label{rbara}
\end{align}
\end{subequations}
cf.\ \eqref{r_ia_rbara}.
As in \eqref{xi}, let $\xi$ denote the principal left eigenvector of $M$.
We introduce the average selection coefficients of genotype $ij$,
gamete $i$, allele $i_n$ at locus $n$, and of the entire population:
\begin{subequations}\label{om(pi)}
\begin{align}
   \om_{ij} &= \sum_\a \xi_\a r_{ij,\a}\,, \label{omij}\\
   \om_i(\pi) &= \sum_j \om_{ij}\prod_k \Pjkk 
                = \sum_\a \xi_\a r_{i,\a}(\pi)\,, \label{om_i}\\
   \ominn(\pi) &= \sum_{i|i_n} \om_{i}(\pi)\prod_{k\ne n} \Pikk
           = \sum_\a \xi_\a \rinan (\pi)\,, \label{om_inn}\\
   \bar\om(\pi) &= \sum_i \om_{i}(\pi)\prod_k \Pikk
           = \sum_\a \xi_\a \rbara(\pi)\,.    \label{bar_om_pi}
\end{align}
For $\bar\om$, we obtain the alternative representations
\begin{equation}
   \bar\om(\pi) = \sum_n\sum_{i_n} \ominn\Pinn \label{bar_om_Pinn}
     		    = \sum_{i,j} \om_{ij} \left(\prod_n\Pinn\right)  \left(\prod_k P_{j_k}^{{(k)}}\right)\,,
\end{equation}
\end{subequations}
and
\begin{equation}\label{dif_barom}
   \der{\bar\om(\pi)}{\Pinn} = 2\ominn(\pi)\,.
\end{equation}

For reasons that will be justified by the following theorem, we call the differential equation
\begin{subequations}\label{weak_sel_lim}
\begin{gather}
  \dot{P}_{i_n}^{(n)} = \Pinn \left[\ominn(\pi) - \bar\om(\pi)\right]\,, \label{wsl1}\\
   		D = 0\,,\; q = 0       \label{wsl2}
\end{gather}
\end{subequations}
on $\S_I^\Ga$ the {\it weak-selection limit} of \eqref{dyn}. 
In view of the following theorem, it is more convenient to consider
\eqref{wsl1} and \eqref{wsl2} on $\S_I^\Ga$ instead of \eqref{wsl1} on $\S_{I_1}\times\cdots\times\S_{I_L}$.
The differential equation \eqref{wsl1} is a Svirezhev-Shashahani
gradient (Remark \ref{rem:Svir_Shash}) with potential function $\bar\om$. 
In particular, $\bar\om$ increases strictly along nonconstant solutions of \eqref{wsl1} because
\begin{equation}\label{om_Lyap}
   \dot{\bar\om} = 2\sum_n\sum_{i_n}\Pinn \left[\ominn(\pi)-\bar\om(\pi)\right]^2 \ge0 \,.
\end{equation}

We will also need the assumption
\begin{equation}\label{hyperbolic}
   \text{All equilibria of \eqref{wsl1} are hyperbolic.} \tag{H}
\end{equation}

\begin{theorem}[B\"urger 2009a]\label{weak_sel_th}
Suppose that \eqref{dyn}, \eqref{weak_sel}, \eqref{ergodic} and \eqref{hyperbolic} hold, the backward migration matrix
$M$ and all recombination rates $c_\K$ are fixed, and $\ep>0$ is sufficiently small.

(a) The set of equilibria $\Xi_0\subset\S_I^\Ga$ 
of \eqref{weak_sel_lim}
contains only isolated points, as does the set of equilibria
$\Xi_\ep\subset\S_I^\Ga$ of \eqref{dyn}. As $\ep\to0$, each equilibrium
in $\Xi_\ep$ converges to the corresponding equilibrium in $\Xi_0$.

(b) In the neighborhood of each equilibrium in
$\Xi_0$, there exists exactly one equilibrium point in $\Xi_\ep$.
The stability of each equilibrium in $\Xi_\ep$ is the same as that of
the corresponding equilibrium in $\Xi_0$; i.e., each pair is either
asymptotically stable or unstable.

(c) Every solution $p(t)$ of \eqref{dyn} converges to one of the 
equilibrium points in $\Xi_\ep$.
\end{theorem}

The essence of this theorem and its proof is that under weak selection
(i) the exact dynamics quickly leads to spatial quasi-homogeneity and quasi-linkage equilibrium, and (ii) after this time, the exact dynamics can be
perceived as a perturbation of the weak-selection limit \eqref{weak_sel_lim}.
The latter is much easier to study because it is formally equivalent to a panmictic one-locus
selection dynamics. Theorem \ref{weak_sel_th} is a singular perturbation result because in the absence of selection every point on $\Ps_0$ is an equilibrium.

Parts (a) and (b) of the above theorem follow essentially from
Theorem 4.4 of Karlin and McGregor (1972b) which is an application of the implicit function
theorem and the Hartman-Grobman theorem. Part (c) is much stronger and relies, among others, on the notion of
chain-recurrent points and their properties under perturbations of the dynamics (see Section \ref{sec:Weak_mig}).

\begin{proof}[Outline of the proof of Theorem \ref{weak_sel_th}]
Theorem \ref{mig_rec_theorem} and the theory of normally hyperbolic manifolds imply that for sufficiently
small $s$, there exists a smooth invariant manifold $\Ps_\ep$ close to $\Ps_0$, and $\Ps_\ep$ is
globally attracting for \eqref{dyn} at a geometric rate (see Nagylaki et al.\ 1999, and the
references there). The manifold $\Ps_\ep$ is characterized by an equation of the form
\begin{equation}\label{q_ps}
	(D,q) = \ep \ps(\pi,\ep)\,,	
\end{equation}
where $\ps$ is a smooth function of $\pi$. Thus, on $\Ps_\ep$, and more
generally, for any initial values, after a long time,
\begin{equation}\label{Dq(t)_Oep}
   D(t) = O(\ep) \quad\text{and}\quad q(t) = O(\ep)\,.
\end{equation}

The next step consists in deriving the recurrence equations in an $O(\ep)$ neighborhood
of $\Ps_0$ which, in particular, contains $\Ps_\ep$. By applying \eqref{weak_sel} and $D(t) = O(\ep)$ to \eqref{dyn_rec2},
straightforward calculations yield\begin{equation}\label{tilde_pinan}
   {p_{i_n,\a}^{(n)}}^\# = \pinan 
      + \ep\pinan \frac{\rinan(\rh_\a)-\rbara(\rh_\a)}{\wbara(\rh_\a)}
      + O(\ep^2)
\end{equation}
for every $\a\in\G$ (see eq.\ (3.10) in Nagylaki et al.\ 1999). By averaging, invoking \eqref{dyn}, \eqref{tilde_pinan}, and \eqref{q_ps}
one obtains
\begin{align}\label{weak_sel_diff}
   {\Pinn}' &= \mu^T{\pinn}' = \mu^T M {p_{i_n}^{(n)}}^\# = \mu^T {p_{i_n}^{(n)}}^\# \notag\\
     &= \Pinn + \ep \Pinn \left[\ominn(\pi) - \bar\om(\pi)\right] +O(\ep^2)\,.
\end{align}
By rescaling time as $\ta=\ep t$ and letting $\ep\to0$, the leading term in \eqref{weak_sel_diff}, 
${\Pinn}' = \Pinn + \ep \Pinn \left[\ominn(\pi) - \bar\om(\pi)\right]$, approximates
the gradient system \eqref{wsl1}. Therefore, we have $\bar\om(\pi')>\bar\om(\pi)$ unless $\pi'=\pi$. In particular,
the dynamics \eqref{weak_sel_diff} on $\Ps_0$ is gradient like.

The eigenvalues of the Jacobian of \eqref{weak_sel_diff} have the form $1+\ep\nu+O(\ep^2)$, where $\nu$ signifies an
eigenvalue of the Jacobian of \eqref{wsl1}. Therefore, \eqref{hyperbolic} implies that also every equilibrium of \eqref{weak_sel_diff} is hyperbolic. 

The rest of the proof is identical to that of Theorem 3.1 in Nagylaki et al.\ (1999). The heart of its
proof is the following. Since solutions of \eqref{dyn} are in phase with solutions on the invariant
manifold $\Ps_\ep$, it is sufficient to prove convergence of trajectories for initial conditions $p\in\Ps_\ep$.
With the help of the qualitative theory of numerical approximations, it can be concluded that the 
chain-recurrent set of the exact dynamics \eqref{dyn} on $\Ps_\ep$ is a small perturbation of
the chain-recurrent set of \eqref{wsl1}. The latter dynamics, however, is gradient like. Because all
equilibria are hyperbolic by assumption, the chain-recurrent set consists exactly of those equilibria. Therefore,
the same applies to the chain-recurrent set of \eqref{dyn} if $\ep>0$ is small.

A final point to check is that unstable boundary equilibria remain in the simplex after perturbation. The 
reason is that an equilibrium $\hat p$ of \eqref{weak_sel_lim} on the boundary of $\S_I^\Ga$ satisfies
$\hat p_i=0$ for every $i$ in some subset of $\I$, and this condition is not altered by migration. The
positive components of $\hat p$ are perturbed within the boundary, where the equilibrium remains.
\end{proof}

Our next result concerns the average of the (exact) mean fitnesses over demes:
\begin{equation}
   \bar w(p) = \sum_\a \xi_\a \wbara (p_\dota)\,.
\end{equation}

\begin{theorem}[B\"urger 2009a]\label{weak_sel_mean_fit_increase}
Suppose the assumptions of Theorem \ref{weak_sel_th} apply. If 
\eqref{Dq(t)_Oep} holds, $\pi$ is bounded away from the equilibria of
\eqref{wsl1}, and $p$ is within $O(\ep^2)$ of $\Ps_\ep$, then $\De\bar w(p)>0$\,.
\end{theorem}

The time to reach an $O(\ep^2)$ neighborhood of $\Ps_\ep$ can be estimated (Remark 4.6 in B\"urger 2009a and Remark 4.13 in Nagylaki and Lou 2007). Unless migration is very weak or linkage very tight, convergence occurs in an evolutionary short period. It follows that mean fitness is increasing during the long period when most allele-frequency change occurs, i.e., after reaching $\Ps_\ep$ and before reaching a small neighborhood of the stable equilibrium.

An essential step in the proof is to show that
\begin{equation}\label{Delta(D,q)}
    (D',q') - (D,q) = O(\ep^2)
\end{equation}
is satisfied if \eqref{Dq(t)_Oep} holds. Therefore, linkage disequilibria and the measure
$q$ of spatial diversity change very slowly on $\Ps_\ep$. This justifies to call states on $\Ps_\ep$
spatially quasi-homogeneous and to be in quasi-linkage equilibrium.

Theorem \ref{weak_sel_th} enables the derivation of the following result about the equilibrium structure of the multilocus migration-selection dynamics \eqref{dyn}. It establishes that arbitrarily many loci with arbitrarily many alleles can be maintained polymorphic by migration-selection balance: 

\begin{theorem}\label{poly_weak_sel_th}
Let $L\ge1$, $\Ga\ge2$, $I_n\ge2$ for every $n\in\Loc$, and let all recombination rates $c_\K$ be positive and fixed.

(a) There exists an open set $\ssQ$ of migration and selection parameters, such that for every parameter combination in $\ssQ$, there is 
a unique, internal, asymptotically stable equilibrium point. This equilibrium is spatially quasi-homogeneous, is
in quasi-linkage equilibrium, and attracts all trajectories with internal initial condition. Furthermore,
every trajectory converges to an equilibrium point as $t\to\infty$.

(b) Such an open set, $\ssQ'$, also exists if the set of all fitnesses is restricted to be nonepistatic and to display intermediate dominance.
\end{theorem}

For a more detailed formulation and a proof, see Theorem 2.2 in B\"urger (2009b). The constructive proof suggests that with increasing number of alleles and loci,
the proportion of parameter space that allows for a fully polymorphic equilibrium shrinks rapidly because at every locus some form of overdominance of 
spatially averaged fitnesses is required. In addition, the proof shows that this set $\ssQ$ exists in the parameter region where migration is strong relative to selection. 

\begin{remark}\rm
The set $\ssQ'$ in the above theorem does not contain parameter combinations such that
all one-locus fitnesses are multiplicative or additive, or more generally satisfy DIDID \eqref{DIDID}. In each of these cases, one gamete becomes fixed as $t\to\infty$ (Proposition 2.6 and Remark 2.7 in B\"urger 2009b).
\end{remark}

\subsection{Strong migration}\label{sec:mult_strong_mig}
Now we assume that selection and recombination are weak relative to migration, i.e., in addition to \eqref{weak_sel}, we posit
\begin{equation}\label{weak_rec}
    c_\K = \ep \ga_\K\,,\quad \K\subseteq\Loc\,,
\end{equation}
where $\ep\ge0$ is sufficiently small and $\ga_\K$ is defined by this relation. Then every solution $p(t)$ of the full dynamics
\eqref{dyn} converges to a manifold close to $q=0$. On this manifold, the dynamics is approximated by the differential equation
\begin{subequations}\label{str_mig_lim}
\begin{equation}\label{str_mig_lim_a}
   \dot P_i = P_i[\om_i(P)-\bar\om(P)] - \sum_\K \ga_\K \left[P_i - \PiKK \PiNN\right]\,,
\end{equation}
which we augment with
\begin{equation}
    q=0\,.
\end{equation}
\end{subequations}
This is called the {\it strong-migration limit} of \eqref{dyn}.

In general, it cannot be expected that the asymptotic behavior of solutions of \eqref{dyn} under
strong migration is governed by \eqref{str_mig_lim} because its chain-recurrent set does not always consist of
finitely many hyperbolic equilibria. The differential equation \eqref{str_mig_lim_a} is the continuous-time version of the discrete-time dynamics
\eqref{dyn_rec2} which describes evolution in a panmictic population subject to multilocus selection and recombination. Although the
recombination term is much simpler than in the corresponding difference equation, the dynamics is not necessarily less complex.
Akin (1979, 1982) proved that \eqref{str_mig_lim} may exhibit stable cycling.
Therefore, under strong migration and if selection and recombination are about equally weak, convergence of trajectories of \eqref{dyn} 
will not generally occur, and only local perturbation results can be derived (B\"urger 2009a, Proposition 4.10).

The dynamics \eqref{str_mig_lim} becomes simple if there is a single locus (then the recombination term in \eqref{str_mig_lim_a} vanishes) or if there is no epistasis, i.e., if fitnesses can be written in the form
\begin{equation}\label{no_ep}
    \om_{ij} = \sum_n u_{i_nj_n}^{(n)} 
\end{equation}
for constants $u_{i_nj_n}^{(n)} \ge0$. Essentially, this means that fitnesses can be assigned to each one-locus genotype and there is no nonlinear interaction between loci. Then mean fitness is a Lyapunov function, all equilibria are in LE, and global convergence of trajectories occurs generically (Ewens 1969, Lyubich 1992, Nagylaki et al.\ 1999).

\subsection{Weak recombination}
If recombination is weak relative to migration and selection, the limiting dynamics is formally equivalent to a multiallelic one-locus 
migration-selection model, as treated in Sections \ref{sec:gen_mig-sel} -- \ref{sec:mult_allele_arb_mig}. Then local perturbation results can be applied but, in general, global convergence of trajectories cannot be concluded and, in fact, does not always occur.

\subsection{Weak migration}
In contrast to the one-locus case treat in in Section \ref{sec:Weak_mig}, in the multilocus case the assumption of weak migration is insufficient to guarantee convergence
of trajectories. The reason is that in the absence of migration the dynamics \eqref{dyn} reduces to
\begin{equation}\label{dyn_ep0}
   p_{i,\a}' = p_{i,\a}^\# = p_{i,\a}\frac{w_{i,\a}}{\bar w_\a} - D_{i,\a}
\end{equation}
for every $i\in I$ and every $\a\in\G$; cf.\ \eqref{dyn_rec2}. 
Therefore, we have $\Ga$ decoupled multilocus selection dynamics, one 
for each deme. Because already for a single deme, stable cycling has been established (Hastings 1981, Hofbauer and Iooss 1984), global perturbation
results can not be achieved without additional assumptions. Such an assumption is weak epistasis.

We say there is weak epistasis if we can assign (constant) fitness
components $u_{i_nj_n,\a}^{(n)}>0$ to single-locus genotypes, such that
\begin{equation}\label{weak_ep}
    w_{ij,\a} = \sum_n u_{i_nj_n,\a}^{(n)} + \et s_{ij,\a}\,,
\end{equation}
where the numbers $s_{ij,\a}$ satisfy $\abs{s_{ij,\a}}\le1$ and $\et\ge0$ is a measure of the strength of 
epistasis. It is always assumed that $\et$ is small enough, so that $w_{ij,\a}>0$. 

For the rest of this section, we assume weak migration, i.e., \eqref{weak_mig} and \eqref{ass_muab}, weak epistasis \eqref{weak_ep}, and
\begin{equation}\label{eps}
   \et = \et(\ep) \,,
\end{equation}
where $\et:[0,1)\to[0,\infty)$ is $C^1$ and satisfies $\et(0)=0$. Therefore, migration and epistasis need not be `equally' weak.
In particular, no epistasis ($\et\equiv0$) is included.

In the absence of epistasis and migration ($\ep=\et=0$), $p$ is an equilibrium point of \eqref{dyn_ep0}
if and only if for every $\a\in\G$, $p_\dota$ is both a selection 
equilibrium for each locus and is in LE (Lyubich 1992, Nagylaki et al.\ 1999). In addition, the only chain-recurrent points of \eqref{dyn_ep0}
are its equilibria (Lemmas 2.1 and 2.2 in Nagylaki et al.\ 1999). Therefore, one can apply the proof of Theorem 2.3 in Nagylaki et al.\ (1999) to deduce the following result which simultaneously generalizes Theorem 2.3 in Nagylaki et al.\ (1999) and Theorem \ref{theorem_weak_mig}:

\begin{theorem}[B\"urger 2009a, Theorem 5.4]\label{theorem_weak_mig_mult}
Suppose that in the absence of epistasis every equilibrium of \eqref{dyn_ep0} is hyperbolic, and $\ep>0$ is sufficiently  small.

(a) The set of equilibria $\Si_0\subset\S_I^\Ga$ of \eqref{no_mig_dyn}
contains only isolated points, as does the set of equilibria
$\Si_\ep\subset\S_I^\Ga$ of \eqref{dyn}. As $\ep\to0$, 
each equilibrium in $\Si_\ep$ converges to the corresponding equilibrium in $\Si_0$.

(b) In the neighborhood of each asymptotically stable equilibrium in
$\Si_0$, there exists exactly one equilibrium in $\Si_\ep$, and it 
is asymptotically stable. In the neighborhood of each unstable internal
equilibrium in $\Si_0$, there exists exactly one equilibrium
in $\Si_\ep$, and it is unstable. In the neighborhood of each unstable 
boundary equilibrium in $\Si_0$, there exists at most one equilibrium 
in $\Si_\ep$, and if it exists, it is unstable. 

(c) Every solution $p(t)$ of \eqref{dyn} converges to one of the equilibrium points in $\Si_\ep$.
\end{theorem}

In contrast to the case of weak selection (Theorem \ref{weak_sel_th}), unstable boundary equilibria can leave the state space under weak migration (Karlin and McGregor 1972a). For an explicit example in the one-locus setting, see Remark 4.2 if Nagylaki and Lou (2007).

If $\ep=0$, then all equilibria are in $\La_0$, i.e., there is LE within each deme, but not between demes. If $\ep>0$ is 
sufficiently small, then there is weak LD within each deme, i.e., $D_{i,\a} = O(\ep)$ for every $i$ and every $\a$.

\begin{remark}\label{rem:weak_mig_mult}\rm
If in Theorem \ref{theorem_weak_mig_mult} the assumption of weak epistasis is not made, statements (a) and (b) remain valid, but (c) does not necessarily follow. Statement (c) follows if, for $\ep=0$, the chain recurrent set consists precisely of the hyperbolic equilibria.
\end{remark}

One of the consequences of the above theorem is the following, which contrasts sharply with Theorem \ref{poly_weak_sel_th}:

\begin{theorem}[B\"urger 2009b, Theorem 3.9]\label{th:upper_bound_weak_mig}
For an arbitrary number of loci, sufficiently weak
migration and epistasis, and partial dominance, the number of demes 
is the generic maximum for the number of alleles that can be maintained at 
any locus at any equilibrium (stable or not) of \eqref{dyn}.
\end{theorem}

\subsection{Weak evolutionary forces}
If all evolutionary forces are weak, i.e., if $\ep\to0$ in \eqref{weak_sel}, \eqref{weak_mig}, and 
\begin{equation}
	c_\K = \ep \ga_\K \quad\text{for all subsets } \K\subseteq \Loc\,,
\end{equation}
the limiting dynamics of \eqref{dyn} on $\S_I^\Ga$ becomes
\begin{equation}\label{weak_all}
    \dot p_{i,\a} = p_{i,\a}[r_{i,\a}(p_\dota) - \rbara(p_\dota)] 
                - \sum_\K \ga_\K \left(p_{i,\a} - \piKaK \piNaN\right) + \sum_\be \mu_{\a\be} p_{i,\be} \,.
\end{equation}
Hence, selection, recombination, and migration are decoupled. This dynamics may be viewed as the continuous-time version of \eqref{dyn}. However,
it may exhibit complex dynamical behavior already if either migration or recombination is absent. Analogs of Theorems \ref{mig_rec_theorem}, \ref{weak_sel_th}, \ref{weak_sel_mean_fit_increase}, \ref{poly_weak_sel_th}, and \ref{theorem_weak_mig_mult} apply to \eqref{weak_all}.

\subsection{The Levene model}\label{sec:mult_Levene}
For the Levene model, the recurrence equations \eqref{dyn} simplify to
\begin{equation}\label{dyn_mult_Levene}
   p_i' = \sum_{j,\ell,\a} R_{i,j\ell} p_j p_\ell c_\a w_{j\ell,\a}/\wb_\a\,, \quad i\in\I\,,
\end{equation}
where this form is sufficient to deduce the main results, i.e., it is not necessary to express the recombination probabilities $R_{i,j\ell}$ in terms of the linkage disequilibria.

Many of the results on the one-locus Levene model in Sections \ref{sec:Levene_general} and \ref{sec:Levene_nodom} can be generalized to the multilocus Levene model if epistasis is absent or weak. These generalizations are based on two key results. The first is that in the absence of epistasis geometric mean fitness is again a Lyapunov function, and the internal equilibria are its stationary points (Theorems 3.2 and 3.3 in Nagylaki 2009b). The second key result is that in the absence of epistasis, generically, every trajectory converges to an equilibrium point that is in LE (Theorem 3.1 in B\"urger 2010). This is true for the haploid and the diploid model. Then a proof analogous to that of Theorem \ref{theorem_weak_mig_mult} yields a global perturbation result for weak epistasis, in particular, generic global convergence to an equilibrium point in quasi-linkage equilibrium (Theorem 7.2 in B\"urger 2010).

With the aid of these results, the next two theorems can be derived about the maintenance of multilocus polymorphism in the Levene model.

\begin{theorem}[B\"urger 2010, Result 5.1]\label{multi_int_dom_fully}
Assume an arbitrary number of multiallelic loci, $\Ga\ge2$ demes of given size, and let all recombination rates be positive and fixed. 
Then there exists a nonempty open set of fitness parameters, exhibiting partial dominance between every pair of alleles at every locus,  
such that for every parameter combination in this set, there is a unique, internal, asymptotically stable equilibrium
point of the dynamics \eqref{dyn_mult_Levene}. This equilibrium is globally asymptotically stable.
\end{theorem}

Theorem \ref{multi_int_dom_fully} is the multilocus extension of Theorem \ref{int_dom_fully}. From a perturbation theorem for weak epistasis in the Levene model, analogous to Theorem \ref{theorem_weak_mig_mult} (Theorem 7.2 in B\"urger 2010), it follows that such an open set exists also within the set of non-epistatic fitnesses. 

The following result is of very different nature and generalizes Proposition 3.18 in Nagylaki (2009b):

\begin{theorem}[B\"urger 2010, Theorem 7.4]\label{theorem_weak_ep_DIDID}
Assume weak epistasis and diallelic loci with DIDID.

(a) If $L\ge\Ga$, at most $\Ga-1$ loci can be maintained polymorphic for an open
set of parameters.

(b) If $L\le\Ga-1$ and if the dominance coefficients in \eqref{DIDID} are arbitrary for each locus but fixed, 
there exists an open nonempty set $\ssW$ of fitness schemes satisfying \eqref{weak_ep} and
of deme proportions $(c_1,\ldots,c_{\Ga-1})$,
such that for every parameter combination in $\ssW$, there is a unique, internal, 
asymptotically stable equilibrium point of \eqref{dyn_mult_Levene}. 
This equilibrium is in quasi-linkage equilibrium and globally attracting. 
The set $\ssW$ is independent of the choice of the recombination rates.
\end{theorem}

Whereas Theorem \ref{multi_int_dom_fully} suggests that spatially varying selection has the potential to maintain considerable multilocus polymorphism, Theorem \ref{theorem_weak_ep_DIDID} shows that properties of the genetic architecture of the trait (no dominance or DIDID) may greatly constrain this possibility. Numerical results in B\"urger (2009c) quantify how frequent polymorphism is in the two-locus two-allele Levene model, and how this depends on the selection scheme and dominance relations. With increasing number of loci or number of alleles per locus, the volume of parameter space in which a fully polymorphic equilibrium is maintained will certainly decrease rapidly.

A few other aspects of the multilocus Levene model have also been investigated. Zhivotovsky et al.\ (1996) employed a multilocus Levene model to study the evolution of phenotypic plasticity. Wiehe and Slatkin (1998) explored a haploid Levene model in which LD  
is caused by epistasis. More recently, van Doorn and Dieckmann (2006) performed a numerical study that admits new mutations (of variable effect) and substitutions at many loci. They showed that genetic variation becomes increasingly concentrated on a few loci. Roze and Rousset (2008) derived recursions for the allele frequencies and for various types of genetic associations
in a multilocus infinite-island model. Barton (2010) studied certain issues related to speciation using a generalized haploid multilocus Levene model 
that admits habitat preferences.

\subsection{Some conclusions}
Local adaptation of subpopulations to their specific environment occurs only if the alleles or genotypes that perform well are maintained within the population. Similarly, differentiation between subpopulations occurs only if they differ in allele or gamete frequencies. Therefore, maintenance of polymorphism is indispensable for evolving or maintaining local adaptation and differentiation. The more loci and alleles are involved, the higher is the potential for local adaptation and differentation. This is the main reason, why conditions for the maintenance of polymorphism played such an important role in this survey.  

Because for a single randomly mating population, polymorphic equilibria do not exist in the absence of epistasis
and of overdominance or underdominance, it is natural to confine attention to the investigation of the maintenance of genetic variation in subdivided populations if in each subpopulation epistasis is absent (or weak) and dominance is intermediate. In addition, intermediate dominance seems to be by far the most common form of dominance.

We briefly recall the most relevant results concerning maintenance of genetic variation and put them into perspective. Particularly relevant results for a single locus are provided in Sections \ref{Two diallelic demes} -- \ref{sec:submult}, and by 
Theorems \ref{no_dom_bound}, \ref{RB_condition},  \ref{Karlin 1977, Result IA}, \ref{th:NL01_nodom}, and \ref{weak_mig_gen_max}. Theorems \ref{poly_weak_sel_th}, \ref{th:upper_bound_weak_mig}, \ref{multi_int_dom_fully}, and \ref{theorem_weak_ep_DIDID} treat multiple loci. Theorems \ref{poly_weak_sel_th} and \ref{multi_int_dom_fully} establish that in structured populations and in the Levene model, respectively, spatially varying selection 
can maintain multiallelic polymorphism at arbitrarily many loci under 
conditions for which in a panmictic population no polymorphism at all can be maintained.

Interestingly, for strong migration and only two demes, an arbitrary number
of alleles can be maintained at each of arbitrarily many loci, whereas
for weak migration, the number of demes is a generic upper bound for
the number of alleles that can be maintained (Theorem \ref{th:upper_bound_weak_mig}).
At first, this appears counterintuitive given the widespread opinion that it is easier to maintain 
polymorphism under weak migration than under strong migration. 

This opinion derives from the fact that for a single diallelic locus and two demes, the parameter region for which a protected
polymorphism exists increases with decreasing migration rate, provided there is a single parameter that measures the strength of selection,
as is the case in the Deakin model; see condition \eqref{PP_Deakin}. 
Additional support for this opinion comes from the study of weak migration in homogeneous and heterogeneous environments 
(Karlin and McGregor 1972a,b; Christiansen 1999), from the analysis of the general Deakin (1966) model (Section \ref{sec: Ga_diallelic_demes}), 
as well as from various numerical studies (e.g., Spichtig and Kawecki 2004, Star et al.\ 2007a,b).
However, Karlin (1982, p.\ 128) observed that for the non-homogeneous 
Deakin model, in which the `homing probabilities' vary among demes,
``it is possible to increase a single homing rate and reduce or even
abrogate the event of $A$-protection''. This is already obvious from the protection condition \eqref{A1_protected_2demes}: If only one of the migration rates in \eqref{kappa} is reduced, it depends on the sign of the corresponding selection coefficient if protection is facilitated or not.

Theorem \ref{poly_weak_sel_th} is not at variance with Theorem \ref{th:upper_bound_weak_mig} or the widespread opinion expressed above. It is complimentary, and
its proof yields deeper insight into the conditions under which genetic variation can be
maintained by migration-selection balance. With strong migration, a stable 
multiallelic polymorphism requires some form of overdominance of suitably averaged
fitnesses for the loci maintained polymorphic. The reason is that strong migration
leads to strong mixing, so that gamete and allele frequencies
become similar among demes. Because the fraction of volume of parameter space, in which average overdominance holds for multiple alleles, shrinks rapidly with increasing number of alleles or loci, we expect very stringent constraints on selection and dominance coefficients for maintaining polymorphism at many loci if migration is strong. 

The conditions for maintaining loci polymorphic under weak migration are much
weaker. With arbitrary intermediate dominance, the proof of Theorem \ref{th:upper_bound_weak_mig} shows that
those alleles will be maintained that are the fittest in at least one niche. This, obviously, limits the number of alleles that can be maintained.

Also in the Levene model arbitrarily many loci can be maintained polymorphic in the absence of epistasis and if dominance is intermediate in every deme (Theorem \ref{multi_int_dom_fully}). But there, the maximum number of polymorphic loci depends on the pattern of dominance
and the number of demes (Nagylaki 2009b, B\"urger 2010, and Theorem \ref{theorem_weak_ep_DIDID}). It is an open problem if for every (ergodic) migration scheme, arbitrarily
many loci can be maintained polymorphic in the absence of epistasis, overdominance, and underdominance.

Although the fraction of parameter space, in which the conditions for 
maintenance of multiallelic multilocus polymorphisms are satisfied,
decreases rapidly, stable multiallelic polymorphism involving small or moderate numbers of alleles or loci does not seem
unlikely. What is required if dispersal is weak, basically, is that 
there is a mosaic of directional selection pressures and different genes or genotypes that are locally well adapted.

\section{Application: A two-locus model for the evolution of genetic incompatibilities with gene flow}
Here we show how to apply some of the above perturbation results in conjunction with Lyapunov functions and an index theorem to derive the complete equilibrium and stability structure for a two-locus continent-island model, in which epistatic selection may be strong. The model was developed and studied by Bank et al.\ (2012) to explore how much gene flow is needed to inhibit speciation by the accumulation of so-called Dobzhansky-Muller incompatibilities (DMI). 

The idea underlying the concept of DMIs is the following. If a population splits into two, for instance because at least one of the subpopulations moves to a different habitat, and the two subpopulations remain separated for a sufficiently long time, in each of them new, locally adapted alleles may emerge that substitute previous alleles. Usually, such substitutions will occur at different loci. For instance, if $ab$ is the original haplotype (gamete), also called wild type, in one population $Ab$ may become fixed (or reach high frequency), whereas in the other population $aB$ may become fixed. If individuals from the derived populations mate, then hybrids of type $AB$, which occur by recombination between $aB$ and $Ab$ have sometimes reduced fitness, i.e., they are incompatible (to a certain degree). Therefore, the accumulation of DMIs is considered a plausible mechanism for the evolution of so-called intrinsic postzygotic isolation between allopatric (i.e., geographically separated) populations. Since complete separation is an extreme situation, Bank et al.\ (2012) studied how much gene flow can inhibit the accumulation of such incompatibilities, hence parapatric speciation.
 
As a simple scenario, Bank et al.\ (2012) assumed a continent-island model, i.e., a continental population, which is fixed for one haplotype (here, $aB$), sends migrants to an island population in which, before the onset of gene flow, $Ab$ was most frequent or fixed. To examine how much gene flow is needed to annihilate differentiation between the two populations, the equilibrium and stability structure was derived. A stable internal equilibrium, at which all four haplotypes are maintained, is called a DMI because, only when it exists, is differentiation between the two populations maintained at both loci. Bank et al.\ investigated more general biological scenarios than the one outlined above, and they investigated a haploid and a diploid version of their model. Here, we deal only with the haploid case, as it admits a much more complete analysis and is also representative for an important subset of diploid models. 

Our goal is to derive the equilibrium and bifurcation structure of the haploid model. We convey the main ideas and only outline most of the proofs. Detailed proofs are given in the Online Supplement S1 of Bank et al.\ (2012). Also comprehensive biological motivation and discussion is provided by Bank et al., as well as illuminating figures.

\subsection{The haploid DMI model}\label{haploid model}
Let $x_1$, $x_2$, $x_3$, and $x_4$ denote the frequencies of the four haplotypes $ab$, $aB$, $Ab$, and $AB$ on the island. 
They satisfy $x_i\ge0$ for every $i$ and $\sum_{i=1}^4 x_i=1$. 
We assume that selection acts on individuals during the haploid phase of their life cycle according to the 
following scheme for Malthusian fitness values: 

\begin{table}[ht]\label{DMI_fitnesses}
\begin{center}
 $\begin{array}{l|cccc}
       &\text{wild type}&\text{continental}&\text{island}&\text{recombinant}\\
\text{haplotype} &ab&aB&Ab&AB
\\ &\hphantom{a+b-c} &\hphantom{a+b-c} &\hphantom{a+b-c} &\hphantom{a+b-c}\\[-5mm]
\hline\\[-3mm]
\text{fitness}& w_1=0 & w_2=\be & w_3=\a & w_4=\a+\be-\ga
\\[1mm]
\text{frequencies}& x_1 & x_2 & x_3 & x_4
\\[1mm]
\end{array}$
\caption{Haplotype fitnesses and frequencies in the DMI model}
\end{center}
\end{table}

This scheme is entirely general. The fitness of the wild type is arbitrarily normalized to 0. 
The parameters $\a$ and $\be$ measure a potential selective advantage of the island and continental haplotypes, respectively, on the
island. They can be positive or negative. However, further below we will assume that
$\a>0$ because, as will be shown, otherwise a DMI cannot exist. 
Finally, the epistasis parameter $\ga $ measures the strength of the incompatibility among the $A$ and $B$ alleles. We assume $\ga \ge 0$, so that epistasis is negative or absent ($\ga=0$). 

Assuming weak evolutionary forces in continuous time, the haplotype dynamics is given by
\begin{equation}
	\dot x_i = x_i(w_i-\wb) - \et_i r D - mx_i + \de_{2i}m,
\end{equation}
where $r$ is the recombination rate between the two loci, $D=x_1x_4-x_2x_3$ is the LD, $\et_1=\et_4=-\et_2=-\et_3$, and $\de_{i2}=1$ if $i=2$ and
$\de_{i2}=0$ otherwise; cf.\ \eqref{recosel_2L2A} and \eqref{weak_all}. With the settings as in Table \ref{DMI_fitnesses}, we obtain
\begin{subequations}\label{OS_hapdyn}
\begin{align}
	\dot x_1 &= x_1[-\a(x_3+x_4)-\be(x_2+x_4)+\ga x_4] - rD -mx_1\,,\\
	\dot x_2 &= x_2[-\a(x_3+x_4)+\be(x_1+x_3)+\ga x_4] + rD +m(1-x_2)\,,\\
	\dot x_3 &= x_3[\a(x_1+x_2)-\be(x_2+x_4)+\ga x_4] + rD -mx_3\,,\\
	\dot x_4 &= x_4[\a(x_1+x_2)+\be(x_1+x_3)-\ga(x_1+x_2+x_3)] - rD -mx_4\,.
\end{align}
\end{subequations}
This is a dynamical system on the simplex $\S_4$ which constitutes our state space.
We always assume $m\ge0$, $r\ge0$, and $\ga\ge0$.

For many purposes, it is convenient to describe the dynamics in terms of the allele frequencies
$p=x_3+x_4$ of $A$, $q=x_2+x_4$ of $B$, and the measure $D$ of LD. Then the dynamical equations read
\begin{subequations}\label{OS_hapdynp}
\begin{align}
  \dot p &= \alpha p (1-p)-\ga (1-p)(pq+D)+\be D-m p\,, \label{OS_hapdynp_pA}\\
  \dot q &= \be q(1-q)-\ga (1-q)(pq+D)+\alpha D+m(1-q)\,, \label{OS_hapdynp_pB}\\
  \dot D &=[\alpha (1-2p) + \be(1-2q)]D - \ga[(1-p)(1-q)-D](pq+D) \notag \\
  		&\quad-rD-m[p(1-q)+D]\,. \label{OS_hapdynp_D}
\end{align}
\end{subequations}
It is straightforward to show that the condition $x\in\S_4$ translates to $0\le p \le0$, $0\le q \le0$, and
\begin{equation}
	-\min[pq,(1-p)(1-q)] \le D \le \min[p(1-q),(1-p)q]\,.
\end{equation}

\subsection{Existence and stability of boundary equilibria}\label{hapappbound}
We denote the monomorphic equilibria $x_i=1$ by $\mathbf{M}_i$. If $m=0$, then all monomorphic equilibria exist. However,
if $\a>0$ and $\ga>\be$, the conditions most relevant for this investigation (see \eqref{OS_conditions} below),
$\mathbf{M}_1$ and $\mathbf{M}_4$ are always unstable. 

It is straightforward to determine local stability of $\Mtwo$ and $\Mthree$. The latter can be stable only if $m=0$.

Next, there may exist two equilibria at which one locus is polymorphic and the other is fixed.
The equilibrium $\SA$ has the coordinates $(p,q,D) = \big(1- \frac{m}{\a -\ga},1,0 \big)$ 
and is admissible if and only if $m < \a-\ga$. 
The equilibrium $\SB$ has coordinates $\big(0,-\frac{m}{\be},0\big)$ and is admissible if and only if
$m <-\be$. The respective local stability conditions are easy to derive.

The following lemma collects the important observations from these analyses.

\begin{lemma}\label{lemma_boundary}
(a) For given $m>0$, at most one of the boundary equilibria $\Mtwo$, $\SA$, or $\SB$ can be stable. 
$\Mtwo$ is asymptotically stable if
\begin{equation}\label{OS_m2}
	m > \max[-\be, \a-\ga, \a -\be -r]\,.
\end{equation}
$\SA$ is asymptotically stable if
\begin{equation}\label{OS_sa}
	\ga<\a+\be \quad\text{and}\quad \frac{(\a -\ga )(\ga-\be)}{\a} \bigg(1 + \frac{\a+\be -\ga }{r} \bigg) < m < \a-\ga\,,
\end{equation}
which requires $r>\ga-\be$. $\SB$ is asymptotically stable if
\begin{equation} \label{OS_sb}
	\ga>\a+\be \quad\text{and}\quad  \frac{-\be\a}{\ga-\be} \left(1+ \frac{\ga-\be -\a}{r} \right) < m < - \be\,,
\end{equation}
which requires $r>\a$.

(b) If $\Mtwo$ is asymptotically stable, then $\SA$ and $\SB$ are not admissible.

(c) As a function of $m$, boundary equilibria change stability at most once. If a change in stability occurs, then it
is from unstable to stable (as $m$ increases).
\end{lemma}

Finally, if $r=0$, there is a fully polymorphic equilibrium $\mathbf{R}_0$ 
on the edge $x_1 = x_4 = 0$ of $\S_4$. Thus, only the island and the continental haplotypes are present. Its coordinates are easily calculated.

In the following we derive global asymptotic stability of boundary equilibria for various sets of parameters by applying 
the theory of Lyapunov functions (e.g.\ LaSalle 1976, in particular, Theorem 6.4 and Corollary 6.5). By global asymptotic
stability of an equilibrium we mean that every trajectory, such that initially all alleles are present, converges to this equilibrium.
By Lemma \ref{lemma_boundary} there is at most one asymptotically stable boundary equilibrium for any given set of
parameters. Hence, convergence of all trajectories to the boundary is sufficient for demonstrating global stability. 
Because global convergence to the boundary precludes the existence of an internal equilibrium, 
it yields necessary conditions for a stable DMI.

\begin{proposition}\label{DMI_ness_cond}
A DMI can exist only if each of the following conditions is satisfied:
\begin{equation}\label{ness_cond}
	\a>0 \text{ and  } \ga>\be\,,
\end{equation}
\begin{equation}\label{m2_reversed}
	m < \a-\be+r\,,
\end{equation}
\begin{equation}\label{m_mmax_many}
	m < \max\{\a-\be,\tfrac14\a,\tfrac14(\ga-\be)\}\,.
\end{equation}
\end{proposition}

Condition \eqref{ness_cond} means that $w(Ab)>\max\{w(ab),w(aB)\}$, i.e., the island type has higher fitness than its one-step mutational neighbors. 

The \emph{proof} of \eqref{ness_cond} uses the Lyapunov functions
\begin{equation}
	Y = \frac{x_1+x_3}{x_3+x_4} = \frac{1-q}{p}
\end{equation}
and
\begin{equation}
	X = \frac{x_1+x_3}{x_1+x_2} = \frac{1-q}{1-p}
\end{equation}
to show that convergence to $\SA$ or $\Mtwo$ occurs if \eqref{ness_cond} is violated.

Condition \eqref{m2_reversed} follows because, if \eqref{ness_cond} is satisfied,
\begin{align}
	\dot x_2 &= x_2 [\be(x_1+x_3) +\ga x_4 -\a(x_3+x_4) -r x_3]  + m(x_1+x_3+x_4) + r x_1x_4 \notag \\
			 & \ge x_2 [ (m+\be)x_1 + (m+\be-\a-r)x_3 +(m-\a+\ga)x_4 ] \label{lya2}
\end{align}
holds if 
\begin{equation}\label{m2}
	m > \max\{-\be,\a-\ga,\a-\be+r\} = \a-\be+r\,.
\end{equation}
Therefore, global convergence to $\Mtwo$ occurs if \eqref{m2} holds.
In particular, $\Mtwo$ is globally asymptotically stable for every $m$ if $r<\be-\a$.

Condition \eqref{m_mmax_many} follows in a similar way.

\subsection{Properties of internal equilibria}\label{section_internal_equi}
\begin{lemma}
Every trajectory eventually enters the region $D\le0$ and remains there. Convergence
to $D=0$ occurs if and only if at least one allele is eventually lost.
Thus, every internal equilibrium satisfies $D<0$. 
\end{lemma}

To \emph{prove} this lemma, we define
\begin{equation}
	Z = \frac{x_2 x_3}{x_1 x_4}\,,
\end{equation}
where $x_1>0$ and $x_4>0$ is assumed. We note that $Z=1$ if and only if $D=0$, and $Z<1$ if and only if $D>0$.
Then
\begin{equation}\label{dotZ}
  \dot Z = x_1x_3x_4(m + \ga x_2) + r D (x_1x_2x_3+x_2x_3x_4+x_1x_2x_4+x_1x_3x_4)\,.
\end{equation}
We observe that $\dot Z\ge0$ holds whenever $D\ge0$. In addition, it follows immediately that
$\dot Z>0$ if $rD>0$ and $x_2+x_3>0$. If $x_2+x_3=0$ and $x_1x_4>0$, then $\dot x_1+\dot x_4<0$ if
$r+m>0$. Hence, all trajectories leave $D>0$ if $r>0$.
If $rD=0$, then $\dot Z=0$ only if $x_3=0$ or if $m=0$ and $\ga x_2=0$. Thus, our result follows by investigating
(i) the dynamics on $x_3=0$ if $r=0$, (ii) the dynamics on $x_2=0$ if $m=r=0$, and (iii) the case $m=\ga=0$. 
We leave the simple first two cases to the reader. The third case is also not difficult and follows immediately from Section 3.4.1
in B\"urger and Akerman (2011).

From here on, we assume
\begin{equation}\label{OS_conditions}
	\a>0 \text{ and } \ga>\be  \text{ and } r>\be-\a\,,
\end{equation}
because we have proved that internal equilibria can exist only if \eqref{OS_conditions} is satisfied. We note that \eqref{OS_conditions}
holds if and only if $\Mthree$ (island haplotype fixed) is linearly stable in the absence of migration.

Our next aim is the derivation of a cubic equation from which the coordinate $p$ of an internal equilibrium $(p,q,D)$ can be obtained. 
Given $p$, the coordinates $q$ and $D$ can be computed from relatively simple explicit formulas. 

By solving $\dot p=0$, we find that, for given $p$ and $q$, and if $p\ne1-\be/\ga$, the value of LD 
at equilibrium is
\begin{equation}\label{D_eq}
		D = D(p,q)= p\,\frac{m+(1-p)(\ga q-\a)}{\be-\ga+\ga p}\,.
\end{equation}
Substituting this into \eqref{OS_hapdynp_pB}, assuming $\be\ne0$, and solving $\dot q=0$ for $q$, we obtain 
\begin{equation}\label{q_12_m}
	q_{1,2}(p) = \frac12\left[\left(1-\frac{m}{\be}\right) \pm \sqrt{Q}\right]\,,
\end{equation}
where
\begin{equation}
	Q = \left(1+\frac{m}{\be}\right)^2 - \frac{4\a m p}{\be(\ga-\be)} - \frac{4\a(\ga-\a)}{\be(\ga-\be)}p(1-p)
\end{equation}
needs to be nonnegative to yield an admissible equilibrium.
Finally, we substitute $q=q_1(p)$ and $D=D(p,q_1(p))$ into \eqref{OS_hapdynp_D} and obtain that an
equilibrium value $p$ must solve the equation
\begin{equation}\label{ABQ}
	(\ga-\be)A(p) - \sqrt{Q} B(p) = 0\,,
\end{equation}
where
\begin{subequations}\label{A(p)B(p)}
\begin{align}
 	A(p) &=  (\ga-\be)\left\{\be[(\a-r)(\ga-2\a)+\ga(\a-\be)] + [\be(\ga+\be-2\a)+r(2\be-\ga)]m + \be m^2 \right\} \notag\\
 			&\quad + \bigl\{\be\left[2\a(2\ga-3\be)(\ga-\a)-\be\ga(\ga-\be)\right] 
 						- (2\ga-\be)(2\a-\ga)r \notag\\
 			&\qquad+ [2\a\be(\ga-2\be)-\be\ga(\ga-\be)+\ga(2\ga-3\be)r]m \bigr\}p \notag\\
 			&\quad-[2\a\be(\ga-\a)(\ga-2\be)+\be\ga(\ga-2\a)r+\ga^2rm]p^2\,,\\
 	B(p) &= (\ga-\be)[-2\a\be+\ga(\be+r)+\be m] \notag \\
 	 		&\quad+ [2\a\be(\ga-\be)-\be\ga(\ga-\be)+r(\be-2\ga)]p + \ga^2rp^2\,.
\end{align}
\end{subequations}
If we substitute $q=q_2(p)$ and $D=D(p,q_2(p))$ into \eqref{OS_hapdynp_D}, we obtain 
\begin{equation}\label{ABQ2}
	(\ga-\be)A(p) + \sqrt{Q} B(p) = 0\,,
\end{equation}
instead of \eqref{ABQ}. 

By simple additional considerations, it is shown that solutions $p$ of \eqref{ABQ} or \eqref{ABQ2} satisfy
\begin{align}
	0 &=(\ga-\be)^2A(p)^2-QB(p)^2 \notag \\
	  &= \frac{4\be}{\ga-\be}(\be-\ga+\ga p)^2[m+(\ga-\a)(1-p)]P(p)\,,
\end{align}
where
\begin{align}\label{P(p)}
	P(p) &=  (\ga-\be)(m+r+\be-\a)[\a\be(\a+\be-\ga-r)-mr(\ga-\be)] \notag\\
		&\quad + \bigl\{\a\be(\a-\be)(\ga-\be)(\a+\be-\ga)+\a(\ga-\be)[3\be(\ga-2\a)+m(4\be-\ga)]r \notag \\
		&\qquad\;+ [\a(\ga^2+\be\ga-\be^2)+\ga(\ga-\be)m]r^2\bigr\}p \notag\\
		&\quad - 2\a r[\be(\ga-\be)(\ga-2\a)+\ga^2r]p^2 + \a\ga^2r^2 p^3\,.
\end{align}
Because $p=1-\be/\ga$ never gives an equilibrium of \eqref{OS_hapdynp} and $p=1-m/(\ga-\a)$ can give rise only to a single-locus
polymorphism, any internal equilibrium value $p$ must satisfy $P(p)=0$. 
 
We can summarize these findings as follows.

\begin{theorem}\label{theorem_max3_equil}
(a) The haploid dynamics \eqref{OS_hapdynp} can have at most three internal equilibria, and the coordinate $\hat p$ of an internal
equilibrium $(\hat p,\hat q,\hat D)$ is a
zero of the polynomial $P$ given by \eqref{P(p)}. 

(b) For given $\hat p$ with $P(\hat p)=0$, only one of $q_1(\hat p)$ or $q_2(\hat p)$ defined in \eqref{q_12_m} can yield an equilibrium value $\hat q$. $\hat D$ is calculated from $\hat p$ and $\hat q$ by \eqref{D_eq}. This procedure yields an internal equilibrium if and only if $0<\hat p<1$, $0<\hat q<1$, and 
\begin{equation}\label{D_admiss_cond}
	-\min[\hat p\hat q,(1-\hat p)(1-\hat q)] < \hat D<0\,.
\end{equation}
\end{theorem}

\begin{remark}\rm
In the absence of epistasis, i.e., if $\ga=0$, there can be at most two internal equilibria. Their coordinates are obtained from a
quadratic equation in $p$. The admissibility conditions are given by simple formulas (B\"urger and Akerman 2011).

If $r=0$, the equilibrium $\mathbf{R}_0$ mentioned in Section \ref{hapappbound} is obtained from the unique zero of $P(p)$.
\end{remark}

For a number of important special or limiting cases, explicit expressions or approximations can be obtained for the internal equilibria.

\subsection{Bifurcations of two internal equilibria}\label{section_mast}
A bifurcation of two internal equilibria can occur if and only if $P(p^\ast)=0$, where $p^\ast\in(0,1)$ is a critical point of $P$,
i.e., $P'(p^\ast)=0$. There are at most two such critical points, and they are given by
\begin{subequations}
\begin{equation}
	p^\ast_{1,2} = \frac{1}{3\a\ga^2r}\left\{2\a[\be(\ga-2\a)(\ga-\be) + \ga^2 r] \pm \sqrt{R}\right\}\,,
\end{equation}
where
\begin{align}
	R &= \a^2\bigl\{-\be(\ga-\be)\left[16\a\be(\ga-\a)(\ga-\be)+\ga^2(3\a^2+\be^2)-\ga^3(3\a+\be) \right] \notag \\
		&\qquad -\be\ga^2(\ga-\be)(\ga-2\a)r +\ga^2(3\be^2-3\be\ga+\ga^2)r^2 \bigr\} \notag \\
	   &\quad	-3mr\a\ga^2(\ga-\be)[4\a\be+\ga(r-\a)].
\end{align}
\end{subequations}

Solving either $P(p^\ast_1)=0$ or $P(p^\ast_2)=0$ for $m$, we obtain after some straightforward manipulations that the
critical value $m^\ast$ must be a solution of the following quartic equation:
\begin{subequations}
\begin{equation}\label{mast_1}
	[\a\be(2\be-\ga)(\a+\be-\ga+r)+\ga(\ga-\be)r m]^2[-\psi_1+2\psi_2 r m + 27\a\ga^2(\ga-\be)r^2 m^2]=0\,,	
\end{equation}
where
\begin{align}
	\psi_1 &= \a(\ga-\be)(\a-\be+r)^2[4\a\be(\ga-\a)(\ga-\be)+\ga^2r^2]\,, \\
	\psi_2 &= 2\a^2(\ga-\be)^2(9\be\ga-8\a\be-\a\ga) + 3\a\ga(\ga-\be)(3\be\ga-2\a\be-\a\ga)r \notag\\
			&\quad+3\a\ga^2(\ga-\be)r^2+2\ga^3r^3\,.
\end{align}
\end{subequations}

The zero $m=m_0$ arising from the first (linear) factor in \eqref{mast_1} does not give a valid bifurcation point for internal
equilibria. 

The second (quadratic) factor in \eqref{mast_1} provides two potential solutions. However,
because $\psi_1\ge0$, one is negative. Therefore, the critical value we are looking for is given by
\begin{equation}\label{m_ast}
	m^\ast = \frac{1}{27\a\ga^2(\ga-\be)r}\left(-\psi_2 + \sqrt{\psi_2^2+27\a\ga^2(\ga-\be)\psi_1}\right)\,.
\end{equation}
At this value, two equilibria with non-zero allele frequencies collide and annihilate each other. 
Thus, $m^\ast$ is the critical value at which a saddle-node bifurcation occurs. 
This gives an admissible bifurcation if both equilibria are internal (hence admissible)
for either $m<m^\ast$ or $m>m^\ast$.  

If $p_1^\ast=p_2^\ast$, i.e., if $R=0$, a pitchfork bifurcation could occur at $m^\ast$. As a function of
$\a$, $\be$, and $\ga$, the condition $p_1^\ast=p_2^\ast$ can be satisfied at $m^\ast$ only for three different values of $r$, of
which at most two can be positive. It can be shown that at each of these values, one of the emerging zeros of $P(p)$
does not give rise to an admissible equilibrium (because $D>0$ there). Thus, only a saddle-node bifurcation can occur.

Simple and instructive series expansions for $m^\ast$ may be found in the Online Supplement of Bank et al.\ (2012).

\subsection{No migration}\label{OS_haploid_no_mig}
We assume $m=0$. From Section \ref{section_internal_equi}, we obtain the following properties of internal equilibria $(p,q,D)$. 
The LD is given by
\begin{equation}\label{D_eq_m0}
	D = D(p,q)= p(1-p)\frac{\ga q-\a}{\be-\ga+\ga p}\,,
\end{equation}
where $p\ne1-\be/\ga$; cf.\ \eqref{D_eq}. For admissibility, \eqref{D_admiss_cond} needs to be satisfied.
For given $p$ and if $\be\ne0$, the coordinate $q$ of an internal equilibrium can assume only one of the following forms:
\begin{equation}\label{q_12}
	q_{1,2}(p) = \frac12\left(1 \pm \sqrt{1-\frac{4\a(\ga-\a)p(1-p)}{\be(\ga-\be)}}\right) \,.
\end{equation}
By Theorem \ref{theorem_max3_equil}, for given $p$, at most one of $q_1=q_1(p)$ or $q_2=q_2(p)$ can give rise to an equilibrium.

The following theorem characterizes the equilibrium and stability structure.

\begin{theorem}\label{theorem_m0_1equ}
Suppose \eqref{OS_conditions} and $m=0$.

(a) The haploid dynamics \eqref{OS_hapdynp} admits at most one internal equilibrium.

(b) Depending on the parameters, the internal equilibrium is given by either \newline
$(p,q_1(p),D(p,q_1(p)))$ or $(p,q_2(p),D(p,q_2(p)))$, where $p$ is one of the two zeros of
\begin{equation}\label{P(p)_m0}
	P(p) = \ga^2 r^2 p^2 - r[2\be(\ga-\be)(\ga-2\a)+r\ga^2]p + \be(\ga-\be)(\a-\be-r)(\a+\be-\ga-r)\,,
\end{equation} 
and $q_i(p)$ and $D(p,q_i(p))$ are given by \eqref{q_12} and \eqref{D_eq_m0}, respectively.

(c) An internal equilibrium exists if and only if both $\Mtwo$ and $\Mthree$ are asymptotically stable. This is the case if and only if
\begin{equation}\label{int_eq_existence}
	\ga>\a \;\text{ and }\; \be>0 \;\text{ and }\; r>\a-\be\,.
\end{equation}

(d) The internal equilibrium is unstable whenever it exists.

(e) If \eqref{int_eq_existence} does not hold, then $\Mthree$ is globally asymptotically stable.
\end{theorem}

This theorem complements the results derived by Feldman (1971) and Rutschman (1994) for the discrete-time dynamics of the 
haploid two-locus selection model. Rutschman proved global convergence to a boundary equilibrium for all parameter combinations
for which no internal equilibrium exists. If transformed to the parameters used by Rutschman, condition \eqref{int_eq_existence}
yields precisely the cases not covered by Rutschman's Theorem 14. Because our model is formulated in continuous time, the
internal equilibrium can be determined by solving quadratic equations. This is instrumental for our proof, as is the following index theorem.

\begin{remark}[Hofbauer's index theorem]\label{IndexTheorem}\rm
Theorem 2 in Hofbauer (1990) states the following. For every dissipative semiflow on $\Real^n_+$ such that all fixed points are regular, the sum of the indices of all saturated equilibria equals +1.

In our model, the index of an equilibrium is $(-1)^m$, where $m$ is the number of negative eigenvalues 
(they are always real). An internal equilibrium is always saturated. If it is asymptotically stable, it has index 1. Equilibria on the boundary of the simplex
are saturated if and only if they are externally stable. This is the case if and only if no gamete that is missing at the
equilibrium can invade. Because $\S_4$ is attracting within $\Real^4_+$, the index of an asymptotically stable (hence, saturated) boundary equilibrium is 1. 
\end{remark}

\begin{proof}[Outline of the proof of Theorem \ref{theorem_m0_1equ}]
Statements (a) and (b) follow from four technical lemmas that provide properties of the polynomial $P(p)$ and resulting admissibility conditions for $(p,q,D$).

The proof of (c) and (d) is the mathematically most interesting part. In view of the above,
it remains to prove that the internal equilibrium exists if \eqref{int_eq_existence} holds and that it is 
unstable. This follows readily from the index theorem of Hofbauer (1990) in Remark \ref{IndexTheorem}. 
In our model, the only boundary equilibria are the four monomorphic states. $M_1$ and $M_4$ are never saturated because they are
unstable within $\S_4$. $\Mtwo$ and $\Mthree$ are saturated if and only if they are asymptotically stable within $\S_4$. Then, $\text{ind}(\Mtwo)=\text{ind}(\Mthree)=1$. Hence there must exist an internal equilibrium with index -1. Such an equilibrium cannot be stable.
Because $\Mtwo$ and $\Mthree$ are both asymptotically stable if and only if \eqref{int_eq_existence} holds, statements (c) and (d) 
are proved. 

(e) follows with the help of several simple Lyapunov functions.
\end{proof}

\subsection{Weak migration}\label{hapappm0}
Theorem \ref{theorem_weak_mig_mult} and Remark \ref{rem:weak_mig_mult} in conjunction with Theorem \ref{theorem_m0_1equ} and Lemma \ref{lemma_boundary}
yield the equilibrium configuration for weak migration.

\begin{theorem}\label{OS_theorem_weakmig}
If $m>0$ is sufficiently small, the following equilibrium configurations occur. 

(a) If \eqref{int_eq_existence} holds, there exists one unstable internal equilibrium and one asymptotically stable internal
equilibrium (the perturbation of $\Mthree$). In addition, the monomorphic equilibrium $\Mtwo$ is asymptotically stable. Neither $\SA$ nor $\SB$ is admissible.

(b) Otherwise, i.e., if $\ga<\a$ or $\be<0$ or $r<\a-\be$, the perturbation of the equilibrium $\Mthree$ is globally 
asymptotically stable (at least if $\ga$ is small).
The equilibrium $\Mtwo$ is unstable, and the equilibria $\SA$ and $\SB$ may be admissible. If $\SA$ or $\SB$ 
is admissible, it is unstable.
\end{theorem}

Throughout, we denote the stable internal equilibrium by $\IDMI$. To first order in $m$ its coordinates are
\begin{equation}
	 \left(1- \frac{m(\a+r)}{\a(\a -\be+r)} , 
			\frac{m(\ga -\be+r)}{(\ga-\be)(\a -\be+r)}, - \frac{m}{\a -\be+r} \right).
\end{equation}

\subsection{The complete equilibrium and stability structure}\label{OS_main_result}
Now we are in the position to prove our main results about the equilibrium and bifurcation structure. We continue to assume
\eqref{OS_conditions}, so that a DMI can occur. Throughout, we always consider bifurcations as a function of (increasing) $m$.

We define
\begin{subequations}
\begin{align}
	\mA &= \frac{(\a -\ga )(\ga-\be)}{\a}\biggl(1+\frac{\a+\be -\ga }{r}\biggr)\,, \\
	\mB &= \frac{-\be\a}{\ga-\be} \biggl(1+\frac{\ga-\be -\a}{r}\biggr)\,, \\
	m_2 &= \a-\be-r  
\end{align}	
\end{subequations}
and note that $\mA$, $\mB$, and $m_2$ are the critical values of $m$ above which $\SA$, $\SB$, and $\Mtwo$, respectively, are asymptotically
stable provided they are admissible. If we set
\begin{subequations}
\begin{numcases}
	{\mmax^- =} 
    	\mA    & if $\ga-\a < 0 < \ga-\be<\a$ and $r > \ga-\be$,  \label{mminus1} \\
		\mB    & if  $\be<0<\a<\ga-\be$ and $r > \a$,      \label{mminus2} \\
		 m_2    & if $r\le\min[\a,\a-\be,\ga-\be]$, \label{mminus3}\end{numcases}
\end{subequations}
then, by Lemma \ref{lemma_boundary}, all boundary equilibria are repelling if and only if $m<\mmax^-$. As the following theorem shows, if $m<\mmax^-$, an asymptotically stable internal equilibrium exists which, presumably, is globally attracting. Hence, a DMI will evolve from every initial condition. Of course, $\mmax^-$ can be zero (if $0<\a<\be<\ga$ and $r\le\be-\a$). We also recall the definition of $m^\ast$ from \eqref{m_ast}.

\begin{theorem}\label{theorem_main1}
The following three types of bifurcation patterns can occur:

Type 1. 
\begin{itemize}[itemsep=-1ex]
\item If $0<m<m^\ast$, there exist two internal equilibria; one is asymptotically stable $(\IDMI)$, the other $(\mathbf{I}_0)$ is unstable.
The monomorphic equilibrium $\Mtwo$ is asymptotically stable.
\item At $m=m^\ast$, the two internal equilibria collide and annihilate each other by a saddle-node bifurcation.
\item If $m>m^\ast$, $\Mtwo$ is the only equilibrium; it is asymptotically stable and, presumably, globally stable.  
\end{itemize}

Type 2. There exists a critical migration rate $\tilde m$ satisfying $0<\tilde m<m^\ast$ such that:
\begin{itemize}[itemsep=-1ex] 
\item If $0<m<\tilde m$, there is a unique internal equilibrium $(\IDMI)$. It is asymptotically stable and, presumably,
globally stable.
\item At $m=\tilde m$, an unstable equilibrium $(\mathbf{I}_0)$ enters the state space by an exchange-of-stability
bifurcation with a boundary equilibrium.
\item If $\tilde m<m<m^\ast$, there are two internal equilibria, one asymptotically stable $(\IDMI)$, 
the other unstable $(\mathbf{I}_0)$, and one of the boundary equilibria is asymptotically stable.
\item At $m=m^\ast$, the two internal equilibria merge and annihilate each other by a saddle-node bifurcation.
\item If $m>m^\ast$, a boundary equilibrium asymptotically stable and, presumably, globally stable.
\end{itemize}

Type 3.
\begin{itemize}[itemsep=-1ex]
\item If $0<m<\mmax^-$, a unique internal equilibrium $(\IDMI)$ exists. It is asymptotically stable and, presumably,
globally stable.
\item At $m=\mmax^-$, $\IDMI$ leaves the state space through a boundary equilibrium by an exchange-of-stability
bifurcation. 
\item If $m>\mmax^-$, a boundary equilibrium is asymptotically stable and, presumably, globally stable.
\end{itemize}
\end{theorem}

\begin{figure}
\begin{center}
	\includegraphics[width=12.0cm]{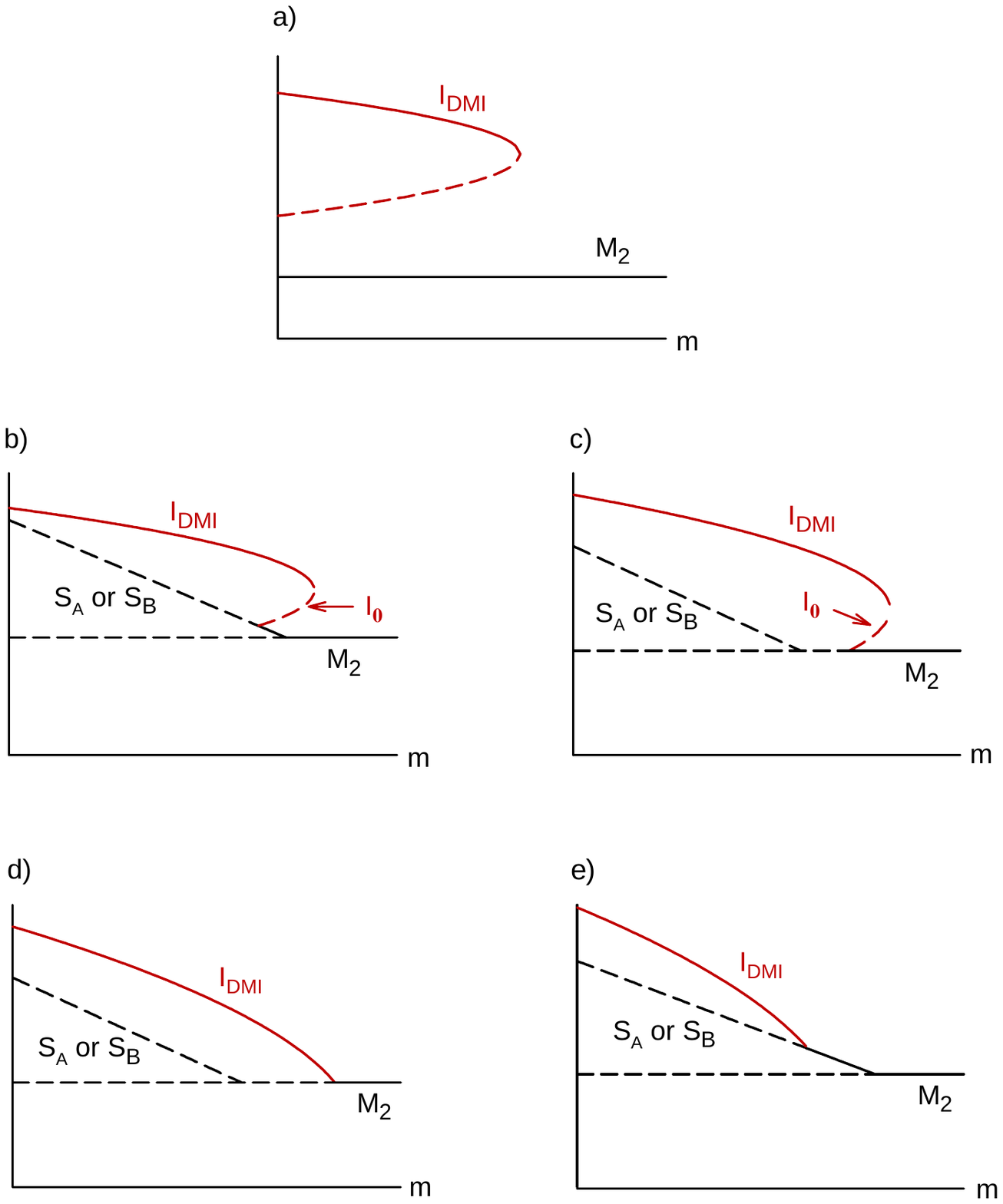}
	\caption{\small Bifurcation patterns in the DMI model according to Theorem \ref{theorem_main1}. Panel a shows the bifurcation pattern of Type 1.
	Panels b and c display the two cases that are of Type 2, and panels d and e those of Type 3.}
	 \label{Survey_DMI_fig1}
\end{center}
\end{figure}

\begin{proof}
Theorem \ref{OS_theorem_weakmig} provides all equilibrium configurations for small $m$. Lemma \ref{lemma_boundary} provides
control over the boundary equilibria. As $m$ increases, they can vanish but not emerge. They can also become asymptotically
stable as $m$ increases. For sufficiently large $m$, there is always a globally asymptotically stable boundary equilibrium. 
By Theorem \ref{theorem_max3_equil}, the number of internal equilibria is at most three. In
addition, internal equilibria can emerge or vanish only either by a saddle-node bifurcation 
(Section \ref{section_mast}) or because an equilibrium
enters or leaves $\S_4$ through one of boundary equilibria, when an exchange of stability occurs. 
A bifurcation involving the two internal equilibria can occur at most
at one value of $m$, namely at $m^\ast$ \eqref{m_ast}. An exchange-of-stability bifurcation can occur only at the
values $\mA$, $\mB$, or $m_2$. If it occurs, the respective boundary equilibrium is asymptotically stable for every
larger $m$ for which it is admissible.
By the index theorem of Hofbauer (1990), Remark \ref{IndexTheorem}, the sum of the indices of all saturated equilibria is 1.

If Case 1 of Theorem \ref{OS_theorem_weakmig} applies, then $\Mtwo$ is asymptotically stable for every $m>0$, and
it is the only boundary equilibrium. Hence, its index is 1. The index of the stable internal equilibrium is also 1. Because the sum of the indices
of the internal equilibria must be 0, the index of the unstable equilibrium is -1. Because at most one 
bifurcation involving the two internal equilibria can occur and because for large $m$, $\Mtwo$ is globally asymptotically stable,
the bifurcation must be of saddle-node type in which the equilibria
collide and annihilate each other (but do not emerge). In principle, the internal equilibria could also leave $\S_4$
through a boundary equilibrium (in this case, it must be $\Mtwo$). However, by the index theorem, they can do so only simultaneously.
This occurs if and only if $m^\ast=m_2$, which is a non-generic degenerate case. 
Because the sum of indices of the internal equilibria must be zero, no
equilibrium can enter the state space. These considerations settle the bifurcation pattern of Type 1.

If Case 2 of Theorem \ref{OS_theorem_weakmig} applies, then, for small $m$, the boundary equilibria are unstable, 
hence not saturated,
and do not contribute to the sum of indices. Since then the indices of internal equilibria must sum up to 1, 
the only possible bifurcation that does not entail the stability of a boundary equilibrium
would be a pitch-fork bifurcation of the internal equilibrium which, by Section \ref{section_mast}, does not
occur. Indeed, because $m^\ast$ is the only value at which a bifurcation among internal
equilibria can occur, and because for large $m$ a boundary equilibrium is globally asymptotically stable, the three
equilibria emerging by a pitchfork bifurcation would have to leave the state space through boundary equilibria. This, however,
cannot occur, as follows easily from the results about linear stability in Section \ref{hapappbound}. 
Thus, any further bifurcations involve a boundary equilibrium. There are two possibilities. 

(i) An equilibrium enters $\S_4$ at some value
$\tilde m$ (which can only be one of $\mA$, $\mB$, or $m_2$) through
one of the unstable boundary equilibria by an exchange-of-stability bifurcation. If $m>\tilde m$, there is one
asymptotically stable boundary equilibrium, an unstable internal equilibrium (the one that entered $\S_4$), and 
one asymptotically stable internal equilibrium. Now a reasoning analogous to that applied above to Case 1 of 
Theorem \ref{OS_theorem_weakmig} establishes the bifurcation pattern of Type 2.

(ii) The internal equilibrium leaves $\S_4$ by exchange of stability through one of the boundary equilibria at $\mmax^-$. 
This becomes asymptotically stable then and, presumably, globally stable. At larger values of $m$ no equilibrium can enter $\S_4$ through one of the
(other) unstable boundary equilibria, because this would either lead to two simultaneously stable boundary equilibria, 
which is impossible (Lemma \ref{lemma_boundary}), or this had to occur at the same value at which the hitherto stable 
boundary equilibrium
merges with $\Mtwo$ and leaves the state space. This, too, is impossible because the sum of the indices of the
new stable boundary equilibrium and the new unstable internal equilibrium would be zero. 
Thus, we have established the bifurcation pattern of Type 3 and excluded all other possibilities.
\end{proof}

Our final goal is to assign the respective parameter combinations to the three types of bifurcation patterns determined above.
To this aim we define four selection scenarios that reflect different biological situations. We say that the fitness landscape is of slope type if at least one of the one-mutant neighbors ($ab$ or $AB$) has fitness intermediate between the continental ($aB$) and the island type ($Ab$). Otherwise, there is a double-peak landscape.

\emph{Selection scenario 1}: $0 < \a < \be <\ga$. It represents the parameter regime in which ``selection against hybrids" (of continental and island haplotypes) is driving DMI evolution. The fitness landscape has two peaks with the higher at the continental type. The genetic incompatibility is strong ($\ga$ large).

\emph{Selection scenario 2}: $0 < \be < \a < \ga$. It represents an intermediate parameter regime in which the two mechanisms of ``selection against hybrids" and ``selection against immigrants" (below) interfere. This is also a double-peak landscape, but with maximum at the island type. The incompatibility is strong.

\emph{Selection scenario 3}: $\be<\ga<\min\{\a,\a+\be\}$. It represents part of the parameter regime in which ``selection against immigrants" drives DMI evolution. The fitness landscape is of slope type and the incompatibility is weak. $\be$ may be positive or negative. 

\emph{Selection scenario 4}: $\be < 0 < \a < \ga -\be$. It represents part of the parameter regime in which ``selection against immigrants" drives DMI evolution. The fitness landscape is of slope type and there is strong local adaptation ($\be < 0 < \a)$.

To formulate the main result, we need the following critical recombination rates:
\begin{subequations}
\begin{align}
	\rA &= (\ga-\a)\,\frac{3(\ga-\be)-\a}{2\ga-\a}\,, \\
	\rB &= \be\,\frac{3\a+\be-\ga}{\be+\ga}\,, \\
	r_2 &= \frac{3\a(\ga-\be)-\sqrt{\a(\ga-\be)(4\be\ga+5\a\ga-9\a\be)}}{2\ga}\,.  
\end{align}	
\end{subequations}

\begin{theorem}\label{theorem_main2}
1. Bifurcation patterns of Type 1 occur in 

\noindent
Selection scenario 1;\newline 
Selection scenario 2 if and only if $r\ge\a-\be$.

2. Bifurcation patterns of Type 2 occur in

\noindent
Selection scenario 2 if and only if $r_2<r<\a-\be$;

\noindent
Selection scenario 3 if and only if one of the following holds:

(a) $r_2^\ast < r \le \ga-\be$,

(b) $r>\max[\ga-\be,\rA]$ and $\ga>\tfrac12\a$,

(c) $\ga-\be<r<\infty$ and $\ga=\tfrac12\a>3\be$,

(d) $\ga-\be < r < \rA$ and $\ga<\tfrac12\a$.

\noindent
Selection scenario 4 if and only if one of the following holds:

(a) $r_2^\ast < r \le \a$,

(b) $r>\max[\a,\rB]$ and $\ga>-\be$,

(c) $\a<r<\infty$ and $\ga=-\be<3\a+\be$,

(d) $\a < r < \rB$ and $\ga<-\be$.

3. Bifurcation patterns of Type 3 occur in

\noindent
Selection scenario 2 if and only if $r\le r_2$; \newline
Selection scenario 3 if and only if one of the following holds:

(a) $r\le\min[\ga-\be,r_2^\ast]$,

(b) $\ga-\be < r \le \rA$ and $\ga>\tfrac12\a$,

(c) $\ga-\be<r<\infty$ and $\ga=\tfrac12\a<3\be$. 

(d) $r\ge\max[\ga-\be,\rA]$ and $\ga<\tfrac12\a$.

\noindent
Selection scenario 4 if and only if one of the following holds:

(a) $r\le\min[\a,r_2^\ast]$,

(b) $\a < r \le \rB$ and $\ga>-\be$,

(c) $\a<r<\infty$ and $\ga=-\be\ge 3\a+\be$,

(d) $r\ge\max[\a,\rB]$ and $\ga<-\be$.
\end{theorem}

For the quite tedious proof, we refer to the Online Supplement of Bank et al.\ (2012).

The theorem shows that for selection scenario 1, in which the continental type is superior to all others, only the bifurcation pattern of Type 1 occurs (provided \eqref{OS_conditions} is assumed). Hence, if a DMI exists, it is only locally stable and initial conditions, or historical contingencies, determine whether differentiation between the populations is maintained or not. The only other situation in which the bifurcation pattern of Type 1 can occur is for the second double-peak scenario, provided linkage between the two loci is sufficiently tight. For slope-type fitnesses (selection scenarios 3 and 4) as well as for scenario 2 with strong recombination, the DMI is always globally stable for weak migration.

If, in selection scenario 4, $\ga(\a+\be)<\be(2\a+\be)$ and $\ga<-\be$ hold, then a bifurcation pattern of Type 3 occurs for every $r>0$. 
These two conditions are satisfied whenever $-\be>\max[2\a,\ga]$. 
Therefore, in the strong local adaptation scenario a globally asymptotically stable DMI occurs whenever the selection intensity on the two loci differs
by more than a factor of two. A bistable equilibrium pattern can occur in this scenario
only if the selection strength on both loci is sufficiently similar and the recombination rate is about as strong as the selection intensity.

In summary, two mechanisms can drive the evolution of parapatric DMIs. In a heterogeneous environment, a DMI can emerge as a by-product of selection against maladpated immigrants. In a homogeneous environment, selection against unfit hybrids can maintain a DMI. No DMI can be maintained if the migration rate exceeds one of the bounds given in Proposition \ref{DMI_ness_cond}.
Therefore, it is the adaptive advantage of single substitutions rather than the strength of the incompatibility that is the most important factor for the evolution of a DMI with gene flow. In particular, neutral DMIs can not persist in the presence of gene flow. Interestingly, selection against immigrants is most effective if the incompatibility loci are tightly linked, whereas selection against hybrids is most effective if they are loosely linked. Therefore, opposite predictions result concerning the genetic architecture that maximizes the rate of gene flow a DMI can sustain.

\parindent=0pt
\section{References}

{\everypar={\hangindent=.5cm \hangafter=1}

Akin E. 1979. The Geometry of Population Genetics.
Lect.\ Notes Biomath.\ 31. Berlin Heidelberg New York: Springer.

Akin E. 1982. Cycling in simple genetic systems. J.\ Math.\ Biol.\ 13, 305-324.

Akin, E. 1993. The General Topology of Dynamical Systems. Providence, R.I.: Amer.\ Math.\ Soc.

Bank, C., B\"urger, R., Hermisson, J.  2012. The limits to parapatric speciation: Dobzhansky-Muller incompatibilities in a continent-island model. Genetics 191, 845–863.

Barton, N.H. 1999. Clines in polygenic traits. Genetical Research 74, 223-236.

Barton, N.H. 2010. What role does natural selection play in speciation? Phil.\ Trans. \ R.\ Soc.\ B 365, 1825-1840.

Barton, N.H., Turelli, M. 2011. Spatial waves of advance with bistable dynamics: cytoplasmic and genetic analogues of Allee effects.
Amer.\ Natur.\ 178, No.\ 3, pp. E48-E75.

Baum, L.E., Eagon, J.A. 1967. An inequality with applications to
statistical estimation for probability functions of Markov processes and
to a model for ecology. Bull.\ Amer.\ Math.\ Soc.\ 73, 360-363.

Berman, A., Plemmons, R.J. 1994. Nonnegative Matrices in the Mathematical Sciences. Philadelphia: SIAM.

Bulmer, M.G. 1972. Multiple niche polymorphism. Amer.\ Natur.\ 106, 254-257.

B\"urger, R. 2000. The Mathematical Theory of Selection, Recombination, and
Mutation. Wiley, Chichester.

B\"urger, R. 2009a. Multilocus selection in subdivided populations I.
Convergence properties for weak or strong migration. J.\ Math.\ Biol.\
58, 939-978.

B\"urger, R. 2009b. Multilocus selection in subdivided populations II.
Maintenance of polymorphism and weak or strong migration. J.\ Math.\ Biol.\
58, 979-997.

B\"urger, R. 2009c. Polymorphism in the two-locus Levene model with nonepistatic
directional selection. Theor.\ Popul.\ Biol.\ 76, 214-228.

B\"urger, R. 2010. Evolution and polymorphism in the multilocus Levene model with no or weak epistasis.
Theor.\ Popul.\ Biol.\ 78, 123-138. 

B\"urger, R. 2011. Some mathematical models in evolutionary genetics. In: The Mathematics of Darwin's Legacy (FACC Chalub \& JF Rodrigues, eds), pp. 67-89. Birkh\"auser, Basel.

B\"urger, R., Akerman, A. 2011. The effects of linkage and gene flow on local adaptation: A two-locus continent-island model. Theor.\ Popul.\ Biol.\ 80, 272-288.

Cannings, C. 1971. Natural selection at a multiallelic autosomal locus
with multiple niches. J.\ Genetics 60, 255-259.

Christiansen, F.B. 1974. Sufficient conditions for protected polymorphism in a subdivided
population. Amer.\ Naturalist 108, 157-166.

Christiansen, F.B. 1975. Hard and soft selection in a subdivided population.
Amer.\ Natur.\ 109, 11-16.

Christiansen, F.B. 1999. Population Genetics of Multiple Loci. Wiley, Chichester.

Conley, C. 1978. Isolated invariant sets and the Morse index. NSF CBMS Lecture Notes 38.
Providence, R.I.: Amer.\ Math.\ Soc. 

Deakin, M.A.B. 1966. Sufficient conditions for genetic polymorphism.
Amer.\ Natur.\ 100, 690-692.
 
Deakin, M.A.B. 1972. Corrigendum to genetic polymorphism in a subdivided population. Australian\ J.\ Biol.\ Sci.\ 25, 213-214.

Dempster, E.R. 1955. Maintenance of genetic heterogeneity. Cold Spring
Harbor Symp.\ Quant.\ Biol.\ 20, 25-32.

Ewens, W. J. 1969. Mean fitness increases when fitnesses are additive. Nature 221, 1076.

Ewens, W.J. 2004. Mathematical Population Genetics. 2nd edition. Springer, New York.

Ewens, W.J. 2011. What changes has mathematics made to the Darwinian theory? In: The Mathematics of Darwin's Legacy (FACC Chalub \& JF Rodrigues, eds), pp. 7-26. 
Birkh\"auser, Basel.

Eyland, E.A. 1971. Moran's island model. Genetics 69, 399-403.

Feller, W. 1968. An Introduction to Probability Theory and Its Applications, vol. I, 
third edn. Wiley, New York.

Fisher, R.A. 1937. The wave of advance of advantageous genes. Ann.\ Eugenics 7, 355-369.

Friedland, S., Karlin, S. 1975. Some inequalities for the spectral radius of nonnegative
matrices and applications. Duke Math.\ J.\ 42, 459-490.

Geiringer, H. 1944. On the probability theory of linkage in Mendelian
heredity. Ann.\ Math.\ Stat.\ 15, 25-57.

Hadeler, K.P., Glas, D., 1983. Quasimonotone systems and convergence to equilibrium
in a population genetic model. J.\ Math.\ Anal.\ Appl.\ 95, 297-303.

Haldane, J.B.S. 1930. A mathematical theory of natural and artificial
selection. Part VI. Isolation. Proc.\ Camb.\ Phil.\ Soc.\ 28, 224-248.

Haldane, J.B.S. 1932. The Causes of Evolution. London: Longmans, Green. (Reprinted with
a new introduction and afterword by E.G.\ Leigh, Jr., by Princeton University Press, 1992.)

Haldane, J.B.S. 1948. The theory of a cline. J.\ Genetics 48, 277-284.

Hardy, G.H. 1908. Mendelian proportions in a mixed population. Science 28, 49-50.

Hastings A. 1981a. Simultaneous stability of $D=0$ and $D\ne0$ for
multiplicative viabilities at two loci: an analytical study. J.\ Theor.\ Biol.\ 89, 69-81.

Hastings A. 1981b. Stable cycling in discrete-time genetic models. Proc.\
Natl.\ Acad.\ Sci.\ USA 78, 7224-7225.

Hirsch, M. W., 1982. Systems of differential equations which are competitive or
cooperative. I: Limit sets. SIAM J.\ Math.\ Anal.\ 13, 167-179.

Hofbauer J., Iooss G. 1984. A Hopf bifurcation theorem of difference equations 
approximating a differential equation. Monatsh.\ Math.\ 98, 99-113.

Hofbauer, J., Sigmund, K. 1988. The Theory of Evolution and Dynamical Systems. Cambridge: University Press.

Hofbauer, J., Sigmund, K. 1998. Evolutionary Games and
Population Dynamics. Cambridge: University Press.

Karlin, S. 1977. Gene frequency patterns in the Levene subdivided
population model. Theor.\ Popul.\ Biol.\ 11, 356-385.

Karlin, S. 1982. Classification of selection-migration structures and
conditions for a protected polymorphism. Evol.\ Biol.\ 14, 61-204.

Karlin, S., Campbell, R.B. 1980. Selection-migration regimes characterized
by a globally stable equilibrium. Genetics 94, 1065-1084.  

Karlin S., Feldman, M.W. 1978. Simultaneous stability of $D=0$ and 
$D\ne0$ for multiplicative viabilities at two loci. Genetics 90, 813-825.

Karlin, S., McGregor, J. 1972a. Application of method of small parameters
to multi-niche population genetics models.
Theor.\ Popul.\ Biol.\ 3, 186-208.

Karlin, S., McGregor, J. 1972b. Polymorphism for genetic and ecological
systems with weak coupling. Theor.\ Popul.\ Biol.\ 3, 210-238.

Kingman, J.F.C. 1961. An inequality in partial averages. Quart.\ J.\ Math.\ 12, 78-80.

Kolmogoroff, A., Pretrovsky, I., Piscounoff, N. 1937.
\'Etude de l'\'equation de la diffusion avec croissance de la quantite de 
mati\'ere et son application \`a un probl\`eme biologique. (French)
Bull.\ Univ.\ Etat Moscou, Ser.\ Int., Sect.\ A, Math.\ et Mecan.\ 1, 
Fasc. 6, 1-25.

LaSalle, J.P. 1976. The Stability of Dynamical Systems. 
Regional Conf.\ Ser.\ Appl.\ Math.\ 25. Philadelphia: SIAM.


Levene, H. 1953. Genetic equilibrium when more than one ecological niche
is available. Amer.\ Natur.\ 87, 331-333.

Lessard S. 1997. Fisher's fundamental theorem of natural selection
revisited. Theor.\ Pop.\ Biol.\ 52, 119-136.

Lewontin, R.C., Kojima K.-I. 1960. The evolutionary dynamics of complex polymorphisms. Evolution 14, 458-472.

Li, C.C. 1955. The stability of an equilibrium and the average fitness of
a population. Amer.\ Natur.\ 89, 281-295.

Li, C.C. 1967. Fundamental theorem of natural selection. Nature 214, 505-506.

Lou, Y., Nagylaki, T., 2002. A semilinear parabolic system for migration
and selection in population genetics. J.\ Differential Equations 181, 388–418.

Lou, Y., Nagylaki, T., Ni, W.-M. 2013. An introduction to migration-selection PDE models. 
Discrete Continuous Dynamical Systems, Series A, in press.

Lyubich, Yu.I. 1971. Basic concepts and theorems of evolutionary genetics
of free populations. Russ.\ Math.\ Surv.\ 26, 51-123.

Lyubich, Yu.I. 1992. Mathematical Structures in Population Genetics.
Berlin Heidelberg New York: Springer.

Maynard Smith, J. 1970. Genetic polymorphism in a varied environment. 
Am.\ Nat.\ 104, 487-490.

Nagylaki, T. 1977. Selection in One- and Two-Locus Systems. Lect.\ Notes
Biomath.\ 15. Berlin, Heidelberg, New York: Springer.

Nagylaki, T. 1989. The diffusion model for migration and selection. Pp.\ 55–75 
in: Hastings, A. (Ed.), Some Mathematical Questions in Biology. Lecture Notes
on Mathematics in the Life Sciences, vol. 20. American Mathematical Society,
Providence, RI.

Nagylaki, T. 1992. Introduction to Theoretical Population Genetics.
Berlin: Springer.

Nagylaki, T. 1993. The evolution of multilocus systems under weak 
selection. Genetics 134, 627-647.

Nagylaki, T. 2009a. Polymorphism in multiallelic migration-selection
models with dominance. Theor.\ Popul.\ Biol.\ 75,239-259.

Nagylaki, T. 2009b. Evolution under the multilocus Levene model without epistasis.
Theor.\ Popul.\ Biol.\ 76, 197-213.

Nagylaki, T., Hofbauer, J., Brunovsk\'y, P., 1999. Convergence of multilocus systems
under weak epistasis or weak selection. J.\ Math.\ Biol.\ 38, 103-133.

Nagylaki, T., Lou, Y. 2001. Patterns of multiallelic poylmorphism
maintained by migration and selection. Theor.\ Popul.\ Biol.\ 59, 297-33.

Nagylaki, T., Lou, Y. 2006a. Multiallelic selection polymorphism.
Theor.\ Popul.\ Biol.\ 69, 217-229.

Nagylaki, T., Lou, Y. 2006b. Evolution under the multiallelic Levene model.
Theor.\ Popul.\ Biol.\ 70, 401-411.

Nagylaki, T., Lou, Y. 2007. Evolution under multiallelic migration-selection
models. Theor.\ Popul.\ Biol.\ 72, 21-40. 

Nagylaki, T., Lou, Y. 2008. The dynamics of migration-selection models.
In: Friedman, A.\ (ed) Tutorials in Mathematical Biosciences IV.
Lect.\ Notes Math.\ 1922, pp.\ 119 - 172. 
Berlin Heidelberg New York: Springer.

Novak, S. 2011. The number of equilibria in the diallelic Levene model with multiple demes. 
Theor.\ Popul.\ Biol.\ 79, 97–101.

Peischl, S. 2010. Dominance and the maintenance of polymorphism in multiallelic migration-selection models
with two demes. Theor.\ Popul.\ Biol.\ 78, 12-25.

Price, G.R. 1970. Selection and covariance. Nature 227, 520-521.

Prout, T. 1968. Sufficient conditions for multiple niche polymorphism.
Amer.\ Natur.\ 102, 493-496.

Seneta, E. 1981. Non-negative Matrices, 2nd ed. London: Allen and Unwin.

Spichtig, M., Kawecki, T.J. 2004. The maintenance (or not) of polygenic
variation by soft selection in a heterogeneous environment. 
Amer.\ Natur.\ 164, 70-84. 

Star, B., Stoffels, R.J., Spencer, H.G. 2007a. Single-locus polymorphism
in a heterogeneous two-deme model. Genetics 176, 1625-1633.

Star, B., Stoffels, R.J., Spencer, H.G. 2007b. Evolution of fitnesses and
allele frequencies in a population with spatially heterogeneous selection
pressures. Genetics 177, 1743-1751.

Strobeck, C. 1979. Haploid selection with $n$ alleles in $m$ niches.
Amer.\ Nat.\ 113, 439-444.

Svirezhev, Yu.M. 1972. Optimality principles in population genetics.
In: Studies in Theoretical Genetics, pp.\ 86--102. Novisibirsk:
Inst.\ of Cytology and Genetics. (In Russian).

van Doorn, G.S., Dieckmann, U. 2006. The long-term evolution of
multilocus traits under frequency-dependent disruptive selection.
Evolution 60, 2226-2238.

Wakeley, J. 2008. Coalescent Theory: An Introduction. Roberts \& Company Publishers, Greenwood Village, Colorado.

Weinberg, W. 1908. \"Uber den Nachweis der Vererbung beim Menschen.
Jahreshefte des Vereins f\"ur vaterl\"andische Naturkunde in W\"urttemberg 64, 368-382.

Weinberg, W. 1909. \"Uber Vererbungsgesetze beim Menschen. Zeitschrift f\"ur
induktive Ab\-stam\-mungs- und Vererbungslehre 1, 377-392, 440-460; 2, 276-330.

Wiehe, T., Slatkin, M., 1998. Epistatic selection in a multi-locus Levene model and
implications for linkage disequilibrium. Theor.\ Popul.\ Biol.\ 53, 75-84.

Wright, S. 1931. Evolution in Mendelian populations. Genetics 16, 97-159. 

Zhivotovsky, L.A., Feldman, M.W., Bergman, A. 1996. On the evolution of phenotypic plasticity
in a spatially heterogeneous environment. Evolution 50, 547-558.

}

\end{document}